%% file: trg-newton-v5.tex
\documentclass[a4paper,11pt]{article}
\pdfoutput=1 

\usepackage{jheppub}

\usepackage{amsmath, amsthm, amssymb}
\usepackage{dsfont}
\usepackage{wasysym, verbatim} 
\usepackage{color}
\usepackage{tensor}
\usepackage{mathtools}
\usepackage{enumitem}
\usepackage{physics}
\usepackage{subcaption}
\usepackage{graphicx}
\usepackage[utf8]{inputenc}
\usepackage{booktabs,multirow, quotes,tablefootnote}
\usepackage{nicematrix}
\usepackage{tikz}

\NiceMatrixOptions
{
	custom-line = 
	{
		letter = : ,
		command = dashedline , 
		ccommand = cdashedline ,
		tikz = dashed
	},
	custom-line = 
	{
		letter = . ,
		command = dotrule , 
		ccommand = cdotrule ,
		tikz = dotted
	}
}

\usepackage{cancel}
\usepackage[normalem]{ulem}

\usepackage{braket}
\usepackage{tcolorbox}
\usepackage{float}

\usepackage{notoccite} 

\usepackage{cleveref}
\creflabelformat{equation}{(#2\textup{#1}#3)}
\crefname{equation}{Eq.}{Eqns.}
\crefname{figure}{Fig.}{Figs.}
\crefname{section}{Section}{Sections}
\crefname{table}{Table}{Tables}
\crefname{chapter}{Chapter}{Chapters}
\crefname{appendix}{Appendix}{Appendices}
\crefname{subsection}{Section}{Sections}
\crefname{remark}{Remark}{Remarks}
\crefname{footnote}{footnote}{footnotes}

\usepackage{xcolor}
\usepackage{hyperref}

\numberwithin{equation}{section}

\theoremstyle{theorem}
\newtheorem*{conjecture*}{Conjecture}

\numberwithin{theorem}{section}
\numberwithin{definition}{section}
\newtheorem{lemma}{Lemma}\numberwithin{lemma}{section}
\theoremstyle{remark}
\newtheorem{remark}{Remark}\numberwithin{remark}{section}

\input{shorthand.tex}

\def\Jratio{a}

\makeatletter
\def\@fpheader{\vspace{3cm} }
\makeatother

\title{Rotations, Negative Eigenvalues, and Newton Method in Tensor Network Renormalization Group}

\author[a]{Nikolay Ebel,}
\author[b]{Tom Kennedy}
\author[a]{and Slava Rychkov}
\affiliation[a]{Institut des Hautes \'Etudes Scientifiques, 91440 Bures-sur-Yvette, France}
\affiliation[b]{Department of Mathematics, University of Arizona,
	Tucson, AZ 85721, USA}

\emailAdd{ebelnikola@gmail.com}
\emailAdd{slava@ihes.fr}
\emailAdd{tgk@arizona.edu}

\abstract{In the tensor network approach to statistical physics, properties of the critical point of a 2D lattice model are encoded by a four-legged tensor which is a fixed point of an RG map. The traditional way to find the fixed point tensor consists in iterating the RG map after having tuned the temperature to criticality. Here we develop a different and more direct technique, which solves the fixed point equation via the Newton method. This is challenging due to the existence of marginal deformations---linear transformations of the coordinate frame, which parametrize a two-dimensional family of fixed points. We address this challenge by including a 90 degree rotation into the RG map. This flips the sign of the problematic marginal eigenvalues, rendering the fixed point isolated and accessible via the Newton method. We demonstrate the power of this technique via explicit computations for the 2D Ising and $3$-state Potts models. Using the Gilt-TNR algorithm at bond dimension $\chi=30$, we find the fixed point tensors with $10^{-9}$ accuracy, much higher than what was previously achieved. 
}

\begin{document}

\maketitle

\section{Introduction}

In the last 15 years, the use of tensor networks \cite{Levin:2006jai} revolutionized renormalization group (RG) studies of two-dimensional statistical lattice models. Tensor RG\footnote{We use this term inclusively, not limited to the first TRG algorithm proposed in \cite{Levin:2006jai}, but to all RG maps manipulating tensor networks, whether or not they involve disentangling steps or only coarse-graining.} maps a tensor network representation of the partition function to an equivalent representation involving a smaller number of tensors. By now there exist many tensor RG algorithms \cite{Levin:2006jai,SRG,TEFR,SRG1,HOTRG,Evenbly-Vidal,LoopTNR,Bal:2017mht,Evenbly:2017dyd,GILT}. They show excellent numerical performance when benchmarked on exactly solvable two-dimensional models such as Ising or 3-state Potts, truncating to tensors of finite size. In three dimensions, there is currently no fully successful numerical algorithm \cite{HOTRG,GILT,parallelHOTRG,Lyu:2023ukj}.\footnote{Although significant progress was recently reported in \cite{Lyu:2024lqh}.} The analytic theory of tensor RG for infinite-size tensors was initiated in \cite{paper1,paper2} in two dimensions and in \cite{Ebel:2024jbz} in three dimensions.

Here we will be dealing with tensor RG for two-dimensional models. Our main result will be a modification of the existing algorithms which allows to recover the critical fixed point tensor in a novel and more efficient way. Traditionally, one obtains approximate fixed point tensors by starting from the lattice model at or nearby the exactly known critical temperature and performing a series of 10-20 RG steps, referred below as the ``shooting'' method. The resulting tensor is close to the fixed point but it is not quite the fixed point, because the irrelevant perturbations have not fully decayed to zero. In this work, we will be able to instead obtain the fixed point via a Newton method. The Newton method converges provided that the Jacobian has no eigenvalues 1. Naively, there are always perturbations with eigenvalues 1 in our problem. They are associated with aspect-ratio deformations, which break the rotation invariance of the model while keeping criticality. Our main idea will be to change the tensor RG map by including in it a $\pi/2$ rotation of the lattice, which changes the problematic eigenvalues 1 to $-1$, and allows the Newton method to converge.

Our Newton method setup is general and would apply to any fixed point without marginal scalar deformations. (A notable such exception is the compactified scalar boson theory, arising in the study of the BKT phase transition.) In this work, we use it to recover the critical isotropic tensor RG fixed points of 2D Ising and 3-state Potts models. We achieve accuracy $\sim 10^{-9}$ in the Hilbert-Schmidt distance, much better than what has been done via the shooting method (see Table \ref{tab:accuracy}), and limited only by roundoff and truncation errors. The key point however is not merely the improvement in numerical accuracy, but that the fixed point can be thus obtained by directly solving the fixed point equation, something which has not been previously appreciated.\footnote{We are aware of only one previous work \cite{early-Newton} which tried to apply the Newton method to locate tensor RG fixed points. Their search failed beyond bond dimension $\chi=8$, for reasons which were not clearly identified at the time, but likely due to a combination of not recognizing the importance of eliminating the stress tensor eigenvalue 1, and working with a tensor RG algorithm without disentangling.}

\begin{table}[ht]
	\centering
	\begin{tabular}{llllccc}\toprule
		Algorithm & Model &$\chi$ & Ref & Method & FP accuracy  \\
		\midrule
		TNR & Ising &42& \cite{Evenbly-Vidal} & shooting & $10^{-5}$ \\
		Gilt-HOTRG & Ising & 30 & \cite{Lyu:2021qlw} & shooting & $10^{-2}$ \\
		Gilt-HOTRG & Ising & 24 & \cite{PhysRevE.109.034111} & shooting & $3\times 10^{-4}$ \\
		Gilt-TNR & Ising & 30 & this work, Sec.~\ref{sec:Isofp} & shooting & $5\times 10^{-5}$\\
		\midrule
		Rotating Gilt-TNR & Ising &30 & this work, Sec.~\ref{sec:rot-iso}  & Newton & $10^{-9}$\\
		Rotating Gilt-TNR & Potts &30 & this work, App.~\ref{app:the-3-state-potts-model}  & Newton & $10^{-9}$\\
		\bottomrule
	\end{tabular}
	\caption{\label{tab:accuracy}Accuracy of 2D Ising and 3-state Potts fixed point determination by Newton method, compared to previous work using shooting method. Ising results of \cite{Evenbly-Vidal,Lyu:2021qlw,PhysRevE.109.034111} are reviewed in Sec.~\ref{sec:Tokyo-comp}. There is no prior work quantifying the FP accuracy for Potts.}
\end{table}

{Our second result will be to validate the method, using the fixed point tensor to extract physical observables---scaling dimensions of the operators of conformal field theory (CFT) describing the phase transition, and comparing with the exact solution.} Most tensor RG studies extract the scaling dimensions from the transfer matrix \cite{TEFR} or from the lattice dilatation operator\footnote{This is the term we propose hoping that it will stick, as there is no generally accepted term.} \cite{TNRScaleTrans}. Here, we will focus on an alternative technique---using the eigenvalues of the tensor RG map linearized around the fixed point (the Jacobian). This is a standard and general RG technique \cite{Wilson:1973jj,DombGreenVol6-Wegner}, {particularly natural in the context of our study. In tensor RG it has been previously used only a handful of times \cite{Lyu:2021qlw,PhysRevE.109.034111}.} (The transfer matrix and the lattice dilatation operator will be studied in the companion paper \cite{paper-DSO}.) With the error in fixed point tensor determination essentially eliminated by the Newton method, the bond dimension truncation error will dominate the error in determination of physical quantities. The main purpose of this paper being conceptual, we will be working at a relatively modest bond dimension $\chi=30$ using a specific tensor RG algorithm Gilt-TNR \cite{GILT}. We have not tried to compete in accuracy of physical quantities with other tensor RG studies using a higher bond dimension, where the fixed point was approached via the shooting method~\cite{Evenbly-Vidal,LoopTNR,GILT}. This is postponed to the future.

We will encounter two subtle effects when using the Jacobian method:
\begin{itemize}
	\item
	      The standard relation $\lambda=b^{d-\Delta}$ between the Jacobian eigenvalues and the scaling dimensions, where $b>1$ is the lattice rescaling factor, will be modified by a phase, depending on the spin of the CFT operator under spatial rotations. This extra phase arises due to the $\pi/2$ rotation present in our map. Thus, some eigenvalues will change sign and some even become complex. 
	      
	    \item We will see that the Jacobian eigenvalues of total derivative operators are not universal and not given by the above formula. This is expected from the CFT point of view, but it may be the first time that this is seen in a numerical study. We view the observation of this effect as our third main result. We will provide a reason why Refs.~\cite{Lyu:2021qlw,PhysRevE.109.034111} did not see this.
\end{itemize}

The paper is organized as follows. In the main text we focus on the 2D Ising model. The 3-state Potts case is completely analogous and is treated in \cref{app:the-3-state-potts-model}. We start in \cref{sec:intuition} by describing in general terms our problem, possible stumbling blocks, and how we plan to resolve them. In \cref{sec:no-rot} we illustrate the intuition developed in \cref{sec:intuition} via concrete calculations using Gilt-TNR \cite{GILT}. This map does not include a rotation, thus the Jacobian has eigenvalues 1 at the fixed point, and the Newton method is not applicable. We locate the approximate fixed point using the traditional ``shooting'' method. We extract scaling dimensions of leading CFT quasiprimaries from the Jacobian eigenvalues at the approximate fixed point. We check that the non-quasiprimary Jacobian eigenvalues are non-universal, i.e.~depend on the RG map. Finally, we present evidence for the existence of a manifold of anisotropic fixed points. We then proceed to describe our main new result in \cref{sec:with-rot}, where we add rotation to the Gilt-TNR map, which turns eigenvalues 1 to $-1$. We set up the Newton method and show that it converges. For anisotropic tensors, we exhibit period-2 oscillations of the RG iterates. In \cref{sec:conclusions} we conclude and suggest future research directions.

Our paper contains numerous detailed appendices which should be useful to anyone wishing to reproduce or improve our results. In \cref{app:conventions} we set up the basic conventions concerning tensor RG and symmetries of tensor networks. In \cref{app:init} we discuss how we initialize our tensor RG flow starting from the nearest-neighbor 2D Ising model, including the anisotropic case.  \cref{app:moreGilt} provides a detailed description of the Gilt-TNR algorithm \cite{GILT} used here, and of the Gilt-HOTRG algorithm used in \cite{Lyu:2021qlw,PhysRevE.109.034111}. \cref{app:GaugeFixing} describes our continuous and discrete gauge fixing procedures. In \cref{app:GiltDiff} we provide evidence that the Gilt-TNR map is differentiable near the fixed point, and optimize the step size used to extract the Jacobian via a finite difference approximation. 

Much or our work here will be numerical and at an intuitive level of rigor, but we hope that our ideas will also prove useful in the eventual fully rigorous computer-assisted construction of exact tensor network RG fixed points---the program started in \cite{paper1}.

\section{Intuition for RG eigenvalues}
\label{sec:intuition}
We start with general remarks about tensor RG as applied to the 2D Ising model. We then explain how conformal field theory (CFT) helps classify RG eigenvalues. Finally, we discuss when the Newton method can and cannot be used to search for the critical RG fixed point. We will see that the Newton method is problematic for non-rotating RG maps which do not preserve rotation symmetry, and we will explain how this problem is solved by adding rotation to the RG map.

\subsection{Tensor RG for the 2D Ising}
\label{sec:TRG2D}

We would like to use tensor RG to study phase transitions and critical points of 2D lattice models. For concreteness we consider the isotropic 2D Ising model with nearest-neighbor ferromagnetic interactions, on the square lattice. (See \cref{app:the-3-state-potts-model} for the 3-state Potts model.) The model is defined by its partition function:
\beq
Z = \sum_{\sigma_{x,y}=\pm 1} e^{\frac 1T \sum_{x,y} (\sigma_{x,y}\sigma_{x+1,y}+\sigma_{x,y}\sigma_{x,y+1})}\,.
\label{eq:ZNN}
\eeq 
It has a phase transition at $T=T_c=\frac{2}{\ln(1+\sqrt{2})}$ separating the ordered and disordered phases. The lattice model is exactly solvable but the exact solution will not be used here. 
As a preparation to using tensor RG, we translate the partition function of the model to a tensor network made of a four-legged tensor $A=A_{\rm NN}(T)$. While there are several ways to do so, probably the simplest way \cite{TEFR} is to rotate the lattice of Ising spins by $\pi/4$ and to define the tensor
\beq
\myinclude[scale=0.8]{fig-Eq2_2_and_EqB_1.pdf} \,  =
e^{\frac 1T (\sigma_1 \sigma_2 + \sigma_2 \sigma_3 +  \sigma_3 \sigma_4
+ \sigma_4 \sigma_1)}\,,
\label{app:ising-tensor2-intro}
\eeq
where every index can take values $\pm1$, so the bond dimension is $2$. The tensor network contraction made out of this tensor $A$ (called tensor network partition function):
\beq
\myinclude[scale=0.5]{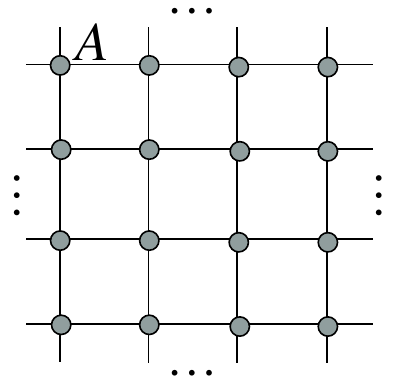}
\label{eq:TNpart-func}
\eeq
 is then identical to the spin model partition function \eqref{eq:ZNN}; see \cref{app:init}. We will be using periodic boundary conditions on the tensor network side. This corresponds to non-standard boundary conditions on the spin model side. This is not a problem since our main interest is anyway in the infinite volume limit.

 We now briefly explain the main idea of tensor RG (see \cite{Evenbly-review} and \cref{app:TRGconventions}). A tensor RG map takes the tensor network contraction \eqref{eq:TNpart-func} made of tensor $A$ and finds another tensor $A'=R(A)$ so that the partition function is the same but the size of the network is reduced by some constant integer $b$ (called lattice rescaling factor). E.g.~a square $N\times N$ network made of $A$ is mapped to a $\frac N b \times \frac N b$ network made of $A'$ with the same periodic boundary conditions and the same value of the partition function. It is convenient to rescale $A$ so that it has unit Hilbert-Schmidt norm:
 \beq
 \|A\|= (\sum|A_{ijkl}|^2)^{1/2}=1\,, \label{eq:norm}
 \eeq
 and similarly $A'$. Then the partition function remains the same up to a rescaling. 
 
 Furthermore, we will assume for now that the map $R$ is non-rotating, i.e.~it preserves the orientation of the lattice. This means that if we apply the map to a rectangular tensor network of size $N\times M$, then after the RG step the lattice has size $N/b \times M/b$. (While this is very natural, in \cref{sec:addrot} we will find it useful to change this.)
 
By now there exist many tensor RG algorithms, which typically involve two steps: 1) disentangling \cite{Evenbly-Vidal}, which cleans the network by removing ultra-short-range correlations known as Corner-Double-Line (CDL) tensors \cite{Levin-talk}, and 2) coarse-graining, which reduces the size of the network. In the early days of tensor RG, only coarse-graining was used \cite{Levin:2006jai,SRG,TEFR,SRG1,HOTRG}, but state-of-the-art algorithms use both disentangling and coarse-graining \cite{Evenbly-Vidal,LoopTNR,Bal:2017mht,Evenbly:2017dyd,GILT}. We will rely on the Gilt-TNR algorithm \cite{GILT}, reviewed in \cref{sec:GiltReview}.
For the moment let us keep the discussion general, applicable to any non-rotating RG map $R$ with disentangling.
 
The RG flow is initialized by $A^{(0)}=A_{\rm NN}(T)$, and generated by applying the RG map repeatedly:
\beq
A^{(n+1)}=R(A^{(n)}),\qquad n=0,1,2,\ldots,
\label{eq:RGevolution}
\eeq
To preserve the partition function exactly, the bond dimension must be allowed to grow. Tensor RG algorithms truncate the tensor $A^{(n+1)}$ in such a way that the bond dimension never exceeds a finite number $\chi$, a parameter of the algorithm. One would like to take $\chi$ as large as possible, within available computational resources. In this paper we will be content to use $\chi=30$.

It is expected that after many steps this RG evolution will converge to a fixed point, which is a tensor $A_{\rm fp}$ satisfying $R(A_{\rm fp})=A_{\rm fp}$. In fact we expect 3 fixed point tensors: high-temperature $A_{H}$, low-temperature $A_{L}$, and critical $A_*$. The RG evolution starting from $A^{(0)}=A^{(0)}(T)$ should converge to $A_{H}$ for $T>T_c$, to $A_{L}$ for $T<T_c$, and to $A_{*}$ for $T=T_c$ (see \cref{fig:exp}).
\begin{figure}
	\centering
	\includegraphics[width=0.5\textwidth]{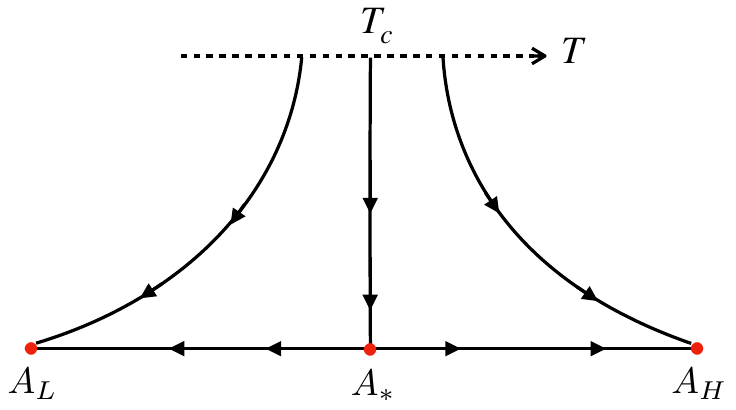}
	\caption{Expected pattern of tensor RG evolution for the 2D Ising model. The dashed line is the curve of initial tensors corresponding to varying $T$ in $A^{(0)}(T)$.
		\label{fig:exp}
	}
\end{figure}
There is much numerical evidence that these expectations can be realized for a tensor RG map with disentangling \cite{Evenbly-Vidal,LoopTNR,Bal:2017mht,Evenbly:2017dyd,GILT}.

Tensors $A_{H}$ and $A_{L}$ in \cref{fig:exp} are simple explicitly known tensors of bond dimension 1 and 2, respectively. Namely, $A_{H}$ has a single nonzero tensor element
\beq
(A_{H})_{0000}=1,
\eeq
and tensor $A_{L}$ is a direct sum of two copies of $A_{H}$. See \cite{paper1,paper2} for a rigorous theory of tensor RG in the neighborhood of these two fixed points. The quest of much of the literature, and our focus here, is the much more interesting and complicated critical fixed point tensor $A_*$, which is expected to have infinite bond dimension. To give meaning to infinite sums in the tensor network built out of $A_*$, it is natural to require that it be Hilbert-Schmidt, i.e.~have a finite Hilbert-Schmidt norm.\footnote{This is because Hilbert-Schmidt tensors can be contracted along one bond: the result is finite and Hilbert-Schmidt. This implies that the partition function of an $N\times M$ tensor network built of Hilbert-Schmidt tensors is finite \cite{paper2}, while this property is not guaranteed for tensors which are not Hilbert-Schmidt.} $A_*$ is expected to depend on the RG map, although the critical exponents extracted from $A_*$ should be universal. An explicit expression for $A_*$ is unknown, and we would find it surprising if it existed.\footnote{Recently, there were claims of exact tensor network RG fixed points constructed in \cite{Ueda:2023ihr,Ueda:2024ixq,Cheng:2023kxh}. We find some aspects of those works questionable. In particular it has not been verified that the proposed exact fixed point tensors remain Hilbert-Schmidt when the cutoff is removed. See \cite{paper-short}, where we also provide an argument that tensor RG without disentangling cannot have an exact RG fixed point which is Hilbert-Schmidt and corresponds to a nontrivial CFT.} At present, progress is possible by working at finite $\chi$ and computing numerical
 finite dimensional approximations to $A_*$. In the future one may be able to prove that an exact fixed point tensor exists in a small neighborhood of numerical approximations, but this is beyond the scope of this work. 

According to \cref{fig:exp}, to get $A_*$ approximately one can take the initial tensor $A^{(0)}$ corresponding to $T=T_c$ and do several RG iterations.\footnote{\label{note:bisection1}The critical temperature $T_c$ may be exactly known if the lattice model is exactly solvable, like the nearest-neighbor 2D Ising. In practical calculations at finite $\chi$, the critical temperature shifts a bit from the exact one. To find the shifted critical temperature, one may use the bisection method, see \cref{note:bisection}.} That's the method used in the prior literature \cite{Evenbly-Vidal,LoopTNR,Bal:2017mht,Evenbly:2017dyd,GILT,Lyu:2021qlw,PhysRevE.109.034111}, which we refer to as ``shooting.'' (Often the fixed point tensor itself is not exhibited and evidence for reaching the fixed point is presented in the form of gauge-invariant observables reaching limits.)
The ``shooting'' method is limited due to the fact that the RG map is not a contraction in the neighborhood of $A_*$. Indeed, the temperature perturbation\footnote{This standard terminology may be a bit confusing, since temperature is a microscopic variable and cannot be varied directly around $A_*$. What is meant is that the eigenvector arises from the temperature perturbation of the microscopic tensor $A^{(0)}(T_c)$ after many RG steps.} is relevant, with eigenvalue $\lambda_t=b^{y_t}>1$, where $y_t=1$ (below we recall how this follows from CFT). The number of iterations one can perform is limited by numerical instabilities, and by the residual dependence of $T_c$ on $\chi$. These limitations motivate our main goal---an alternative way to get $A_*$ via the Newton method. 

\subsection{CFT expectations for RG eigenvalues}
\label{sec:CFTexp}
The Newton method will involve the Jacobian $\nabla R$ of the RG map. To set things up, we need to have an idea about the eigenvalues and eigenvectors of $\nabla R$. This may be understood with the help of the 2D CFT which describes the critical 2D Ising model.\footnote{For prior discussions of relations between tensor RG and CFT, see e.g.~\cite{Lyu:2021qlw,PhysRevE.109.034111,TNRScaleTrans,Ueda2023}. Aspects related to rotation needed for us were not discussed in those references.}  This CFT was exactly solved long ago \cite{Belavin:1984vu}; in particular the spectrum of local operators is exactly known. We stress however that even rough approximate knowledge of the spectrum would suffice for the discussion below.

Since both the RG fixed point and the CFT describe the same critical state, we expect a broad correspondence between these two descriptions. In particular, perturbations of the RG fixed point $A_*$ should be in one-to-one correspondence with the local translationally  invariant perturbations of the CFT, given by integrals of nontrivial (i.e.~excluding the unit operator)\footnote{Perturbing by the unit operator $\cO=\mathds{1}$ simply rescales the partition function. In tensor RG, this would translate to rescaling $A$. Since we fixed the  normalization of $A$, \cref{eq:norm}, such a perturbation will not be present among the $\nabla R$ eigenvectors.} CFT local operators:
\beq
\kappa \int d^2x\, \cO(x)\,.
\label{eq:pertCFT}
\eeq
On the CFT side, the RG transformation corresponds to rescaling distances $x\to bx$. When we do this, the coupling $\kappa$ rescales as $\kappa \to \lambda \kappa$, where
\beq
\lambda = b^{2-\Delta_\cO},
\label{eq:lambdaCFT}
\eeq
and $\Delta_\cO$ is the scaling dimension of $\cO$. Note that this equation is only significant for operators which are not total derivatives. For $\cO$ a total derivative, perturbation \cref{eq:pertCFT} vanishes. Hence, the normalization of $\kappa$ is ambiguous and $\lambda$ can be anything.\footnote{In the standard RG parlance, perturbation by a total derivative is classified as "redundant" \cite{DombGreenVol6-Wegner}.}

In summary, CFT predicts that the RG Jacobian $\nabla R$ at the fixed point $A_*$ should have a series of eigenvectors in one-to-one correspondence with nontrivial local CFT operators $\cO$. Eigenvalues of perturbations associated with $\cO$ which are not total derivatives (such operators are called quasiprimaries) will be given by \cref{eq:lambdaCFT}. Eigenvalues of total derivative perturbations cannot be predicted by general principles; in particular they will depend on the choice of the RG map.

\begin{table}[ht]
	\centering
	\begin{tabular}{lccc}\toprule
		                                & $\mathbb{Z}_2$ & $\Delta$ & $ \ell$ \\
		\midrule
		$\mathds{1}$ & + & 0 & 0\\
		$T,\overline{T}$                & +              & 2        & $\pm 2$       \\
		$T\overline T$                  & +              & 4        & 0       \\
		$TT,\overline{T}\,\overline{T}$ & +              & 4        & $\pm 4$       \\
		\midrule
		$\epsilon$                      & +              & 1        & $0$     \\
		$(L_{-4}+\ldots)\epsilon$         & +              & 5        & $4$     \\
		\midrule
		$\sigma$                        & $-$            & 0.125    & 0       \\
		$(L_{-3}+\ldots)\sigma$           & $-$            & 3.125    & 3       \\
		\bottomrule
	\end{tabular}
	\caption{\label{tab:2D}2D Ising CFT Virasoro primary and quasiprimary operators of scaling dimension $\Delta\le 5$, with their quantum numbers (see e.g.~\cite[Table 8.1]{DiFrancesco:1997nk}). The $\ldots$ stands for a product of lower Virasoro generators of the same total weight as the shown leading terms.}
\end{table}

Quasiprimaries of the Ising CFT are classified as follows \cite{DiFrancesco:1997nk}. 
Recall that in the space of local operators we have action of the Virasoro algebra whose generators are denoted by $L_n,\bar L_n $, $n\in \mathbb{Z}$, and of its finite-dimensional subalgebra of global conformal transformations, generated by $L_n,\bar L_n $, $n=0,\pm 1$. Any local operator is characterized by its holomorphic and antiholomorphic weights $h,\bar h$ defined through equations $L_0 \calO(0)= h \calO(0)$, $\bar L_0 \calO(0)=\bar h \calO(0)$. Linear combinations $L_0\pm\bar L_0$ generating dilatations and rotations, the linear combinations of weights $\Delta=h+\bar h$ and $\ell=h-\bar h$ give scaling dimension and spin of a local operator.

One first classifies the Virasoro primary operators $\calO(x)$, which are singled out algebraically by the condition $L_n \calO(0)=\bar L_n \calO(0)=0$, $n\ge 1$, i.e.~are annihilated by all positive Virasoro generators. The Ising CFT has three such Virasoro primary operators $\mathds{1}$, $\epsilon$, $\sigma$, see \cref{tab:2D}. The quasiprimary operators satisfy a weaker constraint $L_n \calO(0)=\bar L_n \calO(0)=0$, $n=1$. In particular all Virasoro primaries are also quasiprimaries.
Taking a Virasoro primary and acting on it with sequences of negative Virasoro generators $L_{-n}, \bar L_{-n}$, $n<0$, produces infinitely many additional operators called Virasoro descendants. Other quasiprimaries are appropriate linear combinations of Virasoro descendants which are annihiliated by $L_1$, $\bar L_1$. Note that Virasoro generators $L_{-n},\bar L_{-n}$, $n\ge 2$, must necessarily be involved: acting with just $L_{-1}, \bar L_{-1}$ would not produce a quasiprimary, since these generators generate translations, and applying them is akin to taking a total derivative. 

The weights of quasiprimaries can be easily computed from the weights of primaries recalling that $L_{-n}$ increase $h$ by $n$ and  $\bar L_{-n}$ increases $\bar h$ by $n$. E.g.~the holomorphic stress tensor component $T = L_{-2} \mathds{1}$ has $h=2,\bar h =0$, and analogously the antiholomorphic component $\bar T$. These quasiprimaries are present in any CFT. 

 In \cref{tab:2D} we give the full list of quasiprimaries of scaling dimension $\Delta\le 5$, along with spin and the global $\mathbb{Z}_2$ symmetry quantum number. It's interesting to note that many candidate quasiprimaries, such as e.g.~$L_{-2}\epsilon+\ldots$ or $L_{-3}\epsilon +\ldots$ do not appear in this table, while the first quasiprimary built on top of $\epsilon$ has the form $L_{-4}\epsilon +\ldots$. These missing local operators are referred to as ``null descendants.'' They happen to vanish identically in the 2D Ising CFT---a fact closely related to its exact solvability.
 
Going back to the Jacobian of the RG map at the fixed point, perturbations corresponding to all the operators in \cref{tab:2D} should appear among eigenvectors of $\nabla R$, with universal eigenvalues. E.g.~the eigenvector corresponding to $\epsilon$ is the temperature perturbation mentioned above: $\lambda_t=b^{2-\Delta_\eps}=b$. 

A second relevant eigenvector corresponds to $\sigma$. This eigenvector is however $\bZ_2$-odd, and thus its existence does not require extra tuning to reach the fixed point. Put another way, in \cref{fig:exp} we were implicitly assuming that the RG map preserves the $\bZ_2$ symmetry of the 2D Ising model. Thus the whole RG evolution happens within the space of $\bZ_2$-even tensors, with all $\bZ_2$-odd perturbations set to zero. 

Let's discuss the interplay of tensor RG and symmetry in more detail.

\subsection{Symmetry preserved by the RG map}
\label{sec:symmetry}

The 2D Ising model has a global on-site $\bZ_2$ (spin flip) symmetry. It is also invariant under spatial symmetry transformations, which in the isotropic case form the group $D_4$ generated by the axis reflections and $\pi/2$ rotations.\footnote{By spatial symmetry we mean point group symmetry.  Lattice translational symmetry is always present and is left implicit. For unequal horizontal and vertical nearest-neighbor couplings $J_x\ne J_y$, spatial symmetry gets reduced to a smaller group. The final spatial symmetry depends on whether the reduction of the spin model to a tensor network is performed rotating by $\pi/4$ not. In this paper we will study the case $J_x\ne J_y$ with $\pi/4$ rotation, see \cref{app:init}.} Translating the model to a tensor network, one naturally obtains a tensor $A^{(0)}$ which is invariant under all these symmetries. See \cref{app:symmetries} for a recap of symmetries of tensor networks.

Suppose we have a tensor RG map $R$ which preserves symmetry $G$, which means that the subspace $\mathcal{T}^G$ of $G$-invariant tensors is invariant under $R$. Then, if the initial tensor $A^{(0)}$ is $G$-invariant, the whole RG evolution \eqref{eq:RGevolution} happens within the $G$-invariant subspace $\mathcal{T}^G$. We can also look for a fixed point $A_*$ within $\mathcal{T}^G$. This is convenient for several reasons. First, because it implies a reduction in the number of variables to consider. Second, because we do not have to worry about relevant or marginal directions orthogonal to $\mathcal{T}^G$.

Preserving on-site symmetries such as the Ising spin-flip is easy---all tensor RG maps naturally achieve this.  This is reflected in \cref{fig:exp}, as mentioned at the end of \cref{sec:CFTexp}. 

On the other hand, it turns out to be hard to implement disentangling while preserving rotations. Available tensor RG maps with disentangling preserve reflections along one axis \cite{Evenbly-Vidal}, along both axes \cite{Hauru:2015abi,Evenbly:2017dyd,Evenbly-TNR-website}, or no spatial symmetry at all \cite{LoopTNR,Bal:2017mht,GILT}.\footnote{\label{note:isotropic}The isotropic TNR scheme from \cite[App.~C]{Evenbly-review} would preserve rotation symmetry. This complicated map acts on networks built of three different tensors. To our knowledge it has not been implemented in practice.} Gilt-TNR \cite{GILT}, which we will use in practical calculations, belongs to the last group of algorithms.

For such RG maps, even though the initial tensor respects a spatial symmetry, we cannot impose it on subsequent tensors along the RG evolution. Also the critical fixed point tensor cannot be expected to have a spatial symmetry, and indeed our numerical calculations will show that it does not. This may seem puzzling because of course the critical point should enjoy the same spatial symmetry as the original tensor, and indeed a larger one because it should be fully $O(2)$ rotationally symmetric. But there is no contradiction: this simply means that the spatial symmetry of the critical point is not manifest.

One may consider this an inconvenience but not necessarily a problem from the practical point of view. And indeed this did not impede the prior tensor RG work, which found the approximate critical fixed point tensor via shooting from $A^{(0)}(T_c)$. However for our goal of using the Newton method this turns out to be a problem, as we will now explain.

\subsection{Rotation symmetry breaking and eigenvalues 1}
\label{sec:rotation-symmetry-breaking}

According to \cref{tab:2D}, the Jacobian $\nabla R$ at the fixed point will have two eigenvectors which are marginal, i.e.~have eigenvalues 1. These eigenvectors are associated with the CFT stress tensor holomorphic and antiholomorphic components $T$, $\bar T$. In the real stress tensor component basis, we can associate these eigenvectors with the linear combinations
\beq
T_{11}-T_{22}\quad \text{and}\quad T_{12},
\label{eq:Treal}
\eeq
conveniently aligned with the coordinate axes $x_1,x_2$ which we identify with the two axes of the tensor network. 
(Recall that the stress tensor trace $T_{11}+T_{22}$ vanishes in a CFT.) 

Consider deforming the CFT by the stress tensor perturbation $\int d^2x\, h_{\mu\nu} T^{\mu\nu}$ with an $x$-independent $h_{\mu\nu}$. This deforms the 2D metric: $\delta_{\mu\nu}\to g_{\mu\nu}= \delta_{\mu\nu}+h_{\mu\nu}$. Physically, this deforms the rotational symmetry of the critical point, turning circles $|x|=1$ into ellipses $|x|_g=1$ where $|x|_g$ is the distance computed with respect to the $g_{\mu\nu}$ metric. For example the two-point correlation functions in the deformed theory will scale as powers of $|x|_g$. The first perturbation in \eqref{eq:Treal} gives an ellipse whose axes are aligned with the coordinate axes, which breaks $\pi/2$ rotations but preserves axis reflections. The second perturbation gives a tilted ellipse and breaks both rotations and reflections. These deformations make sense for infinitesimal and for finite $h_{\mu\nu}$. We will call the ellipse $|x|_g=1$ the anisotropy ellipse. Because of scale invariance, the overall scale of the anisotropy ellipse does not play a role, so it is characterized by two parameters, e.g.~the main axes ratio and the tilt angle.

\begin{figure}
	\centering
	\includegraphics[width=0.3\textwidth]{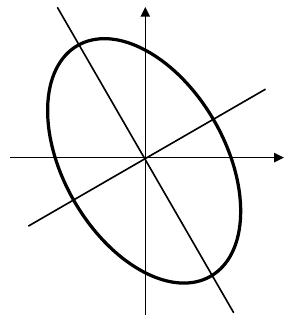}
	\caption{Anisotropy ellipse parametrizing translationally invariant stress-tensor perturbations of the CFT.
		\label{fig:ellipse}
	}
\end{figure}

Even though the above perturbations break rotations, they preserve criticality, since the correlation length remains infinite. We expect a two-dimensional manifold $\mathcal{M}_2$ of corresponding RG fixed points, in one-to-one correspondence with possible anisotropy ellipses. The critical fixed point $A_{*}$ in \cref{fig:exp} is but one point of $\mathcal{M}_2$. The complete picture looks like in \cref{fig:exp1}. The discussed perturbations with eigenvalues 1 span the tangent space of $\mathcal{M}_2$. 

Note that the anisotropy ellipse is a physical quantity which can be associated with any critical tensor, whether it is a fixed point or not. We should just compute two-point correlation functions at large distances and see for which metric $g$ they scale as powers of $|x|_g$. Defined this way, the anisotropy ellipse is an RG invariant on the critical manifold for the kind of RG maps we are considering (non-rotating RG maps changing both lattice dimensions by the same factor $b$), and all points on the RG trajectory will have the same anisotropy ellipse as the initial critical tensor. In the case at hand the initial tensor represents the nearest-neighbor isotropic Ising model on the square lattice, for which the anisotropy ellipse is of course the circle. So, the same will be true for the fixed point $A_{*}$ which we will therefore call "isotropic fixed point". The previous observation explains why RG iterations get attracted to a single fixed point rather than start wandering around the fixed point manifold $\mathcal{M}_2$.

Now that we know that the fixed point $A_{*}$ is isotropic, does this help us to find it? 
The answer would be yes if we had a tensor RG algorithm preserving lattice rotation symmetry. In that case we could be sure that $A_{*}$ belongs to a subspace of rotation-symmetric tensors. We could restrict RG evolution to this subspace, within which $\nabla R$ would be free of eigenvalues 1. It would have been trivial then to set up the Newton method. But in the absence of manifest rotation symmetry, the situation is not so rosy. Even though $A_{*}$ is isotropic, there is nothing manifestly special about $A_{*}$ as a tensor which could help us to limit the search. We are forced to work in a larger space of all tensors, which contains the above perturbations, and hence $\nabla R$ will have eigenvalues 1.\footnote{If we used a tensor RG algorithm preserving one or both axis reflections, such as \cite{Evenbly-Vidal,Hauru:2015abi,Evenbly:2017dyd,Evenbly-TNR-website}, we could project out the second perturbation in \eqref{eq:Treal}, but the first one would still be present.}

\begin{figure}
	\centering
	\includegraphics[width=0.5\textwidth]{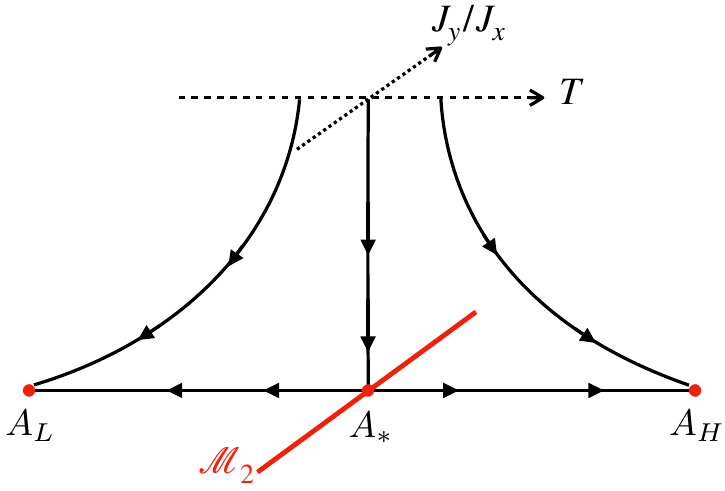}
	\caption{This figure corrects \cref{fig:exp} to show that we expect a two-dimensional manifold $\calM_2$ of fixed point tensors. (Since two of the three dimensions in the figure represent directions orthogonal to the fixed point manifold, we had to represent $\calM_2$ by a one-dimensional thick red line.) Microscopically, we may consider the anisotropic 2D Ising model with varying $J_y/J_x$. Tensor RG flows starting from the critical point of this model (dotted line) will be attracted to a one-dimensional submanifold of $\calM_2$. Adding the next to nearest neighbor interaction we may cover the whole $\calM_2$.
		\label{fig:exp1}
	}
\end{figure}

Now, it is well known that eigenvalues 1 present a problem for solving a fixed point equation via the Newton method. At an intuitive level, the problem can be explained as follows. If the Newton method does work, we can rewrite the fixed point equation in an equivalent form so that the map becomes a contraction in a neighborhood of the solution. A contraction has a unique fixed point. Thus, we conclude that the solution must be isolated. However we have seen that the solution is not isolated---there is a manifold of fixed points. This contradiction shows that there must be an obstruction to setting up the Newton method. In \cref{sec:newton} we will see the problem more clearly from a formal point of view. But first let us discuss how we will solve the problem and get rid of eigenvalues 1.

\subsection{Adding a rotation to the tensor RG map}
\label{sec:addrot}
In the discussion so far, we assumed (\cref{sec:TRG2D}) that the tensor RG map $R$ is non-rotating, i.e.~it preserves the orientation of the underlying lattice. This was needed e.g.~when applying the prediction \eqref{eq:lambdaCFT} for RG eigenvalues to perturbations associated with CFT operators with spin, like the stress tensor. We saw that such maps have features (eigenvalues 1 and a manifold of fixed points) which are problematic for our goal of using the Newton method. To solve this problem, we propose to consider rotating RG maps, which change the orientation. Concretely, we will consider a tensor RG map $R^\circ$ which is a composition of a non-rotating map $R$, followed by $\Gamma_{\pi/2}$, rotation of the tensor by $\pi/2$:
\beq
R^\circ = \Gamma_{\pi/2} R\,,
\label{eq:Rwithrot}
\eeq
or graphically:
\beq
\myinclude[scale=0.7]{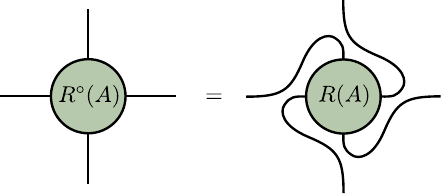}\,.
\label{eq:RGwithrot}
\eeq
So, the rotating map $R^\circ$ maps a tensor network of size $N\times M$ to one which has size $M/b\times N/b$, see Eq.~\eqref{eq:ZAN1mod}:
\beq
\myinclude[scale=0.6]{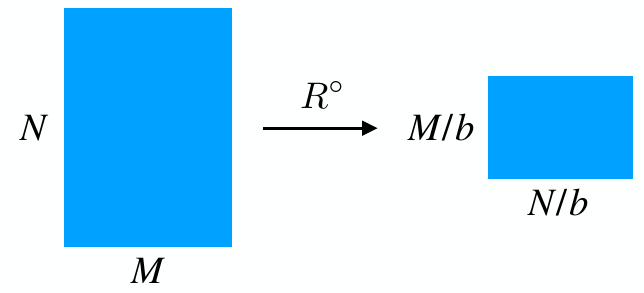}\,.
\label{eq:rotTN}
\eeq
Adding a rotation to the RG map may appear strange, and indeed recognizing that this is useful is the main idea of our paper.\footnote{The earliest tensor RG map \cite{Levin:2006jai} did change the orientation, however the point we will make---that this eliminates eigenvalues 1---did not play any role in that pioneering paper which was focused on other issues.}

Starting from the isotropic 2D Ising tensor $A^{(0)}(T)$, the RG evolution of $R^\circ$ is expected to show the same pattern of fixed points as that of $R$ in \cref{fig:exp}. In particular, the flow from $A^{(0)}(T_c)$ will lead to a critical isotropic fixed point tensor, which we denote $A_*^\circ$. (Note that $A_*^\circ\ne A_*$, the fixed point tensor of $R$, because $R$ does not preserve rotations.) Jacobian eigenvalues around the fixed point $A_*^\circ$ can be computed using the CFT data, up to a change having to do with the rotation. On the CFT side, a spin $\ell$ operator $\calO$ transforms under rotation by angle $\theta$ as
\beq
\calO\mapsto e^{i \theta \ell} \calO \,.
\eeq

In our case $\theta=\pi/2$. This implies that \cref{eq:lambdaCFT} for the Jacobian eigenvalues needs to be modified in case of $R^\circ$ as follows:
\beq
\lambda = (i)^{\ell_\cO} b^{2-\Delta_\cO}\,.
\label{eq:lambda-mod}
\eeq
Applying this to the CFT stress tensor which has spin $\pm 2$, we conclude that the eigenvectors of $\nabla R^\circ$ corresponding to the stress tensor components $T$, $\overline T$ will have eigenvalues $-1$, and not $1$ as for $\nabla R$. Relatedly, when extended to anisotropic and not-rotationally invariant tensors, the map $R^\circ$ will not have a two-dimensional manifold of fixed points as $R$ did, \cref{fig:exp1}. Instead, the isotropic fixed point $A_*^\circ$ will be isolated, and its neighborhood will contain a two-parameter manifold of period-2 orbits of $R^\circ$. 

Eigenvalues $-1$ pose no problem for the Newton method, and we expect to use the Newton method to find the fixed point $A_*^\circ$ of $R^\circ$. We stress that the fixed point $A_*^\circ$ of $R^\circ$ is as good as ${A}_*$ of $R$ for describing the critical point physics. It is only that $A_*^\circ$ can be found using the Newton method, while ${A}_*$ cannot.

We remark that in addition to changing the problematic stress tensor eigenvalues from $1$ to $-1$, rotation is not expected to have any deleterious effect on other eigenvalues. Before applying rotation, eigenvalues of quasiprimaries were positive real numbers, after rotation they acquire a spin-dependent phase while keeping the same absolute value. In particular they will become imaginary for odd spin, negative for spin $\ell=2 \text { mod } 4$, and will remain positive for spin $\ell =0 \text { mod } 4$.

A second remark is that the proposed solution of enabling the Newton method by including rotation is not limited to the Ising model. In \cref{app:the-3-state-potts-model} we will show that it works also for the 3-state Potts model. Generally, our solution will work as long as the fixed point does not have any scalar operators which are internal symmetry singlets and are marginal, i.e.~have dimension 2. Such operators have eigenvalue 1 and, being scalars, their eigenvalue will remain 1 under any rotation.\footnote{We don't know of any interesting model where operators of spin $\ell=4,8,\ldots$ have dimension 2. In reflection positive models they should have dimension at least $\ell$ by the unitarity bounds.} Experience shows that models with marginal scalars are not so frequent, although special models may have them. One such family of models is compactified scalar boson, where marginal scalar perturbation corresponds to changing the compactification radius \cite{DiFrancesco:1997nk}.

\subsection{Newton method}
\label{sec:newton}

Let us finally set up the Newton method. As mentioned, we will be solving the fixed point equation $R^\circ(A)=A$, which has an isolated solution at $A_*^\circ$.  We rewrite the equation as
\beq
f(A)=0,\quad\text{where } f(A)=A-R^\circ(A)\,.
\eeq
The Newton method solves the latter equation via iterations\footnote{Note that we use superscripts to number RG iterates \eqref{eq:RGevolution}, while subscripts for Newton iterates.}
\beq
A_{(m+1)}= g(A_{(m)}),\quad \text{where } g(A)=A- [J (A)]^{-1} f(A)\,,
\label{eq:newton}
\eeq
and 
\beq
J(A)=I-\nabla R^\circ(A)
\label{eq:J}
\eeq 
is the Jacobian of the map $f$. Note that since the RG map preserves the subspace of $\mathbb{Z}_2$-even tensors, here we only need the Jacobian restricted to this subspace.\footnote{Below, when extracting the CFT spectrum from the linearized RG map, we will also consider the Jacobian in the $\mathbb{Z}_2$-odd subspace.}

As discussed, $\nabla R^\circ(A_*^\circ)$ does not have eigenvalues 1. This implies that $J(A_*^\circ)$ is invertible and $g(A)$ is well defined in a neighborhood of $A_*^\circ$.\footnote{On the contrary, if we try to apply this to $R$ instead of $R^\circ$, $\nabla R$ has eigenvalues 1, $J$ is not invertible, and the Newton method fails.} Furthermore, if we compute $\nabla g$ and evaluate it at $A_*^\circ$ 
using $f(A_*^\circ)=0$, we find
\beq
\nabla g(A_*^\circ)=0\,.
\eeq
This implies in particular that $g(A)$ is a contraction near $A_*^\circ$. Therefore, Newton iterations \eqref{eq:newton} will converge for a sufficiently good initial approximation.

The Jacobian $J(A)$ entering the definition of $g(A)$ is somewhat expensive to evaluate and invert.\footnote{In particular, for the bond dimension $\chi=30$ used in our numerical calculations, just storing full $J(A)$ would require around a terabyte of RAM (assuming we use $64$ bit floating point numbers). Such memory requirement is far beyond our computational resources.} Several numerical strategies are available which provide a sufficiently good approximate inverse Jacobian, whose discussion is postponed to \cref{sec:inv-jac} below.

\section{Results for the non-rotating Gilt-TNR}
\label{sec:no-rot}

We now proceed to concrete tensor RG computations, based on the Gilt-TNR algorithm. We will start in this section with several in-depth studies using the original Gilt-TNR algorithm, which is non-rotating. We will thus be able to compare with the prior work from Refs.~\cite{GILT,Lyu:2021qlw,PhysRevE.109.034111}, and to test various intuitions about tensor RG from \cref{sec:intuition}, such as the existence of marginal eigenvalues corresponding to the stress-tensor deformations, and of a manifold of fixed points corresponding to different anisotropy parameters. We will also test the Jacobian method for extracting CFT scaling dimensions. After this thorough warm-up, in the next section we will switch to the rotating Gilt-TNR algorithm, and present our main numerical result: find the fixed point via the Newton method.

\begin{remark}\label{rem:chi-fp} Before we start, the following comment is in order. In the previous section, when we discussed the conceptual picture, the fixed points $A_*$ or $A^\circ_*$ denoted exact fixed points of the exact non-rotating or rotating RG map. From now on we will be discussing numerical RG maps operating at a finite bond dimension $\chi$. The logic of the previous section will still apply to some approximation, but of course the fixed points will have finite bond dimension and will be somewhat different from $A_*$ and $A^\circ_*$. To stress this we will denote the fixed points as $A^{[\chi]}_*$ and $A^{\circ[\chi]}_*$. Our goal here will be to extract these fixed points as precisely as possible, for a finite $\chi$. What happens with $A^{[\chi]}_*$ and $A^{\circ[\chi]}_*$ as $\chi\to\infty$, whether they converge to $A_*$ or $A^\circ_*$, is a separate and very interesting question which is beyond the scope of this work. See however comments at the end of Section \ref{sec:Isofp}.
	\end{remark}

\subsection{Review of the non-rotating Gilt-TNR algorithm}
\label{sec:GiltReview}

The original, non-rotating, Gilt-TNR RG map \cite{GILT} is defined as follows. Take a group of four contracted $A$ tensors inside the tensor network (a plaquette). One starts by inserting on the contracted bonds matrices $Q_1,\ldots,Q_4$ chosen so that the insertion 1) leaves the contraction intact (approximately); 2) eliminates short-range CDL-like correlations around the plaquette. How to choose the $Q_i$'s is the trademark of the algorithm, see \cite{GILT} and \cref{app:moreGilt}. Performing a singular value decomposition (SVD) of each $Q_i$ matrix, we get the following result:
\beq
\myinclude[scale=0.7]{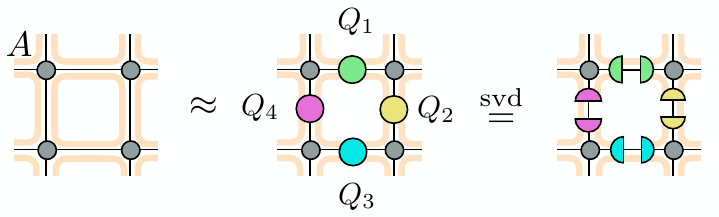}\,,
\label{eq:gilt1}
\eeq
where the light-brown shading represents CDL-like correlations. 

We now perform the above replacement, with the same $Q_1,\ldots,Q_4$, to every second plaquette of the tensor network, in a checkerboard pattern. We then contract the half-tensors with the $A$'s, defining two tensors $B_1$ and $B_2$ as:
\beq
\myinclude{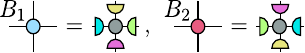}\,.
\label{eq:gilt2}
\eeq
The original tensor network made of $A$ is thus mapped to a tensor network made of $B_1$ and $B_2$, which has CDL-like correlations eliminated on every other plaquette:
\beq
\myinclude{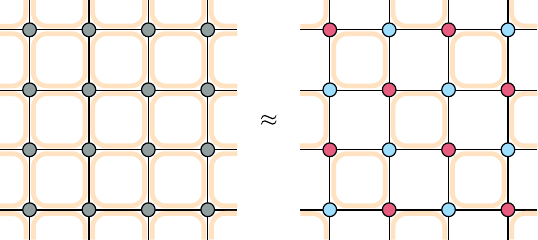}\,.
\label{eq:gilt3}
\eeq
This finishes the disentangling stage of Gilt-TNR.

The second stage of Gilt-TNR is coarse-graining. The default coarse-graining strategy of our code \cite{our-code} is via two consecutive TRG \cite{Levin:2006jai} iterations:
\beq
\text{TRG}_1:B_1,B_2\mapsto C,\quad \text{TRG}_2:C\mapsto \tilde A,
\eeq
Namely, TRG$_1$ performs SVD-decompositions of $B_1$ and $B_2$ diagonally, and then recombines the SVD pieces into a tensor $C$, so that CDL-like correlations remaining in \eqref{eq:gilt3} are contracted away:
\beq
\myinclude{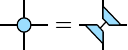}\,,
\myinclude{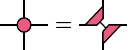}\,,\quad
\myinclude{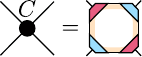}\,.
\label{eq:gilt4}
\eeq
Analogously, TRG$_2$ performs an SVD-decomposition of $C$ diagonally in two different ways and defines a tensor $\tilde A$ by recombining the pieces. The point is that after TRG$_1$ there are no CDL-like correlations in the networks, and so it makes sense to follow up with TRG$_2$ right away before any further disentangling.\footnote{We thank Cl\'ement Delcamp for explaining this point and for passing to us a minimal working version of the original Gilt-TNR code \cite{Gilt-TNR-code} where the described strategy is realized.} A moment's thought shows that after the two TRG iterations the lattice ends up rotated by $\pi/2$. Since the original Gilt-TNR code \cite{Gilt-TNR-code} aims to be non-rotating, the third and final stage of Gilt-TNR is to define the RG-transformed tensor $A'$ by rotating $\tilde A$ back to recover the original orientation.

The map $\mathcal{R}:A\mapsto A'$ is the non-normalized Gilt-TNR map. The corresponding normalized RG map $R(A)$ is defined by $R(A)=A'/\|A'\|$. The lattice rescaling factor is $b=2$. 

\subsection{Isotropic fixed point by the ``shooting'' method} 
\label{sec:Isofp}

Most of our work in this section will be for the 2D Ising model which is isotropic, i.e. has equal nearest-neighbor couplings $J_x=J_y$ (only in \cref{sec:anisotropy} we will consider the anisotropic case). We transform the partition function of the isotropic model to a tensor network, as outlined in \cref{sec:TRG2D} and in more detail in \cref{app:init}. The resulting tensor $A_{\rm NN}(t)$ has bond dimension 2. We parametrize it by the reduced temperature $t=T/T_c$ where $T_c$ is the critical temperature.

We thus have the tensor $A^{(0)}=A_{\rm NN}(t)$ to which we start applying the non-rotating Gilt-TNR algorithm, obtaining a sequence of tensors $A^{(n)}$, $n=1,2,\ldots$, see \cref{eq:RGevolution}. Our first goal is to verify the picture from \cref{fig:exp}. To see that a tensor $A^{(n)}$ converges to a fixed point we need to apply gauge-fixing after each RG step, as explained below. But as a proxy, to detect fixed point behavior, it is standard to use gauge-invariant observables. One such set of observables are the singular values of $A^{(n)}$ decomposed along a diagonal:
\beq
\myinclude{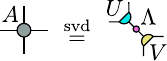}.
\label{fig:singval}
\eeq
Here $U,V$ are unitaries and $\Lambda$ is a diagonal matrix of singular values.
In \cref{fig:sval_traj_nr} we plot the singular values, normalized by the largest singular value, as a function of the RG step. The RG evolution was started at $t=1.0000110043$, the value found by bisecting\footnote{\label{note:bisection}As mentioned in \cref{note:bisection1}, the critical temperature will shift from the exact $t_c=1$ when working at finite bond dimension. To locate the finite-$\chi$ critical temperature, one uses the bisection algorithm: 1) Start from a `bracketing interval' $I=[t_l,t_h]$ such that $t_l$ is in the low-temperature phase and $t_h$ is in the high-temperature phase; 2) Test the midpoint of the interval $I$ - which phase is it in? This is done by running tensor RG for several steps until the tensor starts approaching the high-temperature or lower-temperature fixed points, which can be noticed by monitoring the singular values computed as in \eqref{fig:singval}; 3) Bisection: depending on the outcome of 2), replace the interval $I$ by its lower or upper half $I'$ whose endpoints are in different phases; 4) Repeat the above 3 steps until the bracketing interval has length smaller than the required accuracy $\Delta t$.} up to $\Delta t=10^{-10}$, to maximize the length of the plateaux visible in \cref{fig:sval_traj_nr}, whose presence signals that the RG evolution nears the fixed point, before visibly deviating from it for $n\gtrsim30$. In the plotted case, $t$ happens to be slightly below the true $t_c$ for the chosen Gilt-TNR parameters, and one can see that the deviation happens in the direction of the low-temperature phase where the fixed point tensor $A_{*L}$ has two identical singular values, all other ones being equal to zero. For $t$ slightly above $t_c$ one observes similar plateaux with eventual deviation towards the high-temperature phase where the fixed point tensor $A_{*H}$ has a single nonzero singular value. Our \cref{fig:sval_traj_nr} compares well e.g.~with the results reported in: 
\begin{itemize}
	\item \cite[Fig.~6]{GILT} using the non-rotating Gilt-TNR;
	\item \cite[Fig.~9]{Lyu:2021qlw} for another closely related non-rotating algorithm Gilt-HOTRG;
	\item \cite[Fig.~6, two upper panels]{2023arXiv230617479H} for Loop-TNR \cite{LoopTNR} and NNR-TNR.
	\end{itemize}

\begin{figure}
	\centering
	\includegraphics[width=0.5\textwidth]{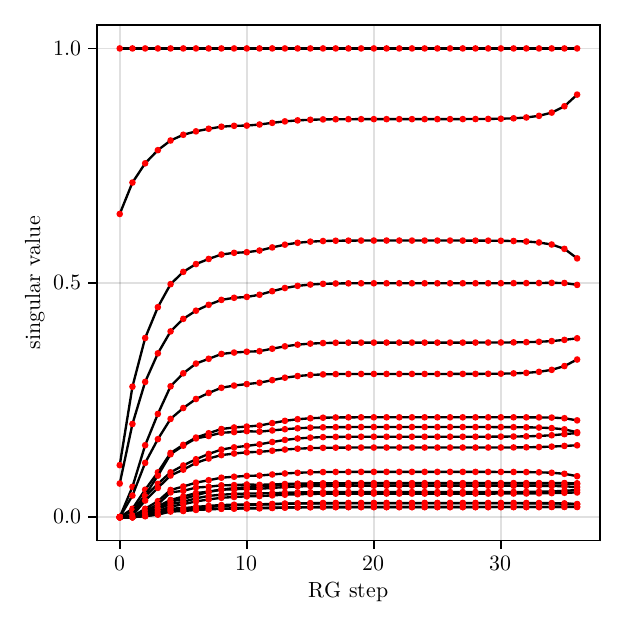}
	\caption{
	\label{fig:sval_traj_nr} The RG flow of tensors starting from $A^{(0)}$ corresponding to the Ising model at the reduced temperature $t=1.0000110043$ and with anisotropy $J_x/J_y=1$ (i.e. the isotropic tensor), in terms of the 20 largest singular values along the diagonal, \cref{fig:singval}, normalized by the first one. Gilt-TNR parameters (see \cref{tab:Giltparams}): $\chi=30$, $\epsilon_{\rm gilt}=6\times 10^{-6}$.
	}
\end{figure}

The plateau in \cref{fig:sval_traj_nr} is a strong indication that the family of tensors $A^{(n)}$ approaches a common gauge equivalence class, so that the gauge invariant observables do converge. Next, we would like to see this explicitly, by choosing a gauge for each tensor $A^{(n)}$ so that gauge-fixed tensors tend to a limit, in the Hilbert-Schmidt norm. This limit is then a fixed-point tensor. Our gauge-fixing procedure is described in \cref{app:GaugeFixing}. It involves two steps: continuous gauge fixing followed by discrete gauge fixing. Continuous gauge fixing matrices are obtained from diagonalizing appropriate environments constructed out of $A^{(n)}$. Discrete gauge fixing consists in multiplying by diagonal matrices with $\pm1$ on the diagonal.

\cref{fig:crit_convergence} demonstrates the stated convergence of the gauge-fixed RG flow to the fixed point. In this figure, we plot the Hilbert-Schmidt distance between gauge-fixed tensors $A^{(n)}$ and $A^{(n+1)}=R(A^{(n)})$ as a function of the RG step.\footnote{To compute the distance between tensors $A^{(n)}$ with $n\leq10$ we applied additional relative discrete gauge fixing (akin to that used in \cite{Lyu:2021qlw}), as using $I_+$ evaluated at the tensor approximating the critical one was not enough to bring tensors close to each other (see \cref{app:GaugeFixing}).} We see that the distance decreases until $n\sim 23$ and then starts increasing. The decrease and increase of the distance are rather smooth,\footnote{Except for a blip around step 11 where a small irregularity is also seen in the singular values, \cref{fig:sval_traj_nr}. } showing that our gauge-fixing procedure works well. The increase of the distance for $n\gtrsim 23$ is roughly consistent with the leading RG eigenvalue being $\lambda=b^{2-\Delta_\epsilon}=2$. This eigenvalue corresponds to the CFT quasiprimary $\epsilon$.

On the other hand the convergence rate for $n \lesssim 23$ roughly agrees with the largest irrelevant RG eigenvalue in the $\mathbb{Z}_2$-even sector being $\lambda_{\rm irr}\approx 0.63$ as we will see below when studying the Jacobian eigenvalues (\cref{tab:isospec1}, column ``$\lambda$ of $\nabla R$''). This eigenvalue does not correspond to a quasiprimary, the leading irrelevant quasiprimaries being $T^2,\, \overline{T}^2,\, T\overline T$ of dimension 4 and eigenvalue $0.25$.\footnote{For the trajectory originating in the nearest-neighbor Ising model, operator $T\overline T$ is known to be finetuned to zero \cite{Caselle_2002,Reinicke_1987,Poghosyan2019,Ueda2023}.\label{foot:noTTbar}} Although this eigenvalue corresponds to a total derivative operator, it does control the rate of approach to the fixed point.

The following general analysis is instructive. Given the critical temperature to $\Delta t$ accuracy, and the leading irrelevant $\mathbb{Z}_2$-even eigenvalue $\lambda_{\rm irr}$, which is the optimal number of steps $n_*$ for the shooting method, and which accuracy $\delta_*$ of approximating the fixed point tensor we may expect? Basic RG intuition says that after $n$ steps, the fixed point tensor is approximated by
\beq
\delta_n\sim \max((\lambda_{\rm irr})^n, (\Delta t) 2^n),
\eeq
where we use that the relevant $\mathbb{Z}_2$-even eigenvalue is $\approx 2$. Minimizing this expression over $n$ we get
\beq
n_*\sim  -\frac{\log_2 \Delta t}{1-\log_2 \lambda_{\rm irr}},
\qquad \delta_*\sim \Delta t^\frac{1}{1-1/\log_2 \lambda_{\rm irr}}\,.\label{eq:general}
\eeq
For $\lambda_{\rm irr}=0.63$, $\Delta t=10^{-10}$, we get $n_*\sim 20$, $\delta_*\sim 10^{-4}$, in reasonable agreement with \cref{fig:crit_convergence}.
	
In \cref{sec:Tokyo-comp}, we will compare our \cref{fig:crit_convergence} to the results of \cite{Evenbly-Vidal,Lyu:2021qlw,PhysRevE.109.034111}, who also considered the norm difference of two consecutive tensors as a function of the RG step..

\begin{figure}
	\centering
	\includegraphics[width=0.6\textwidth]{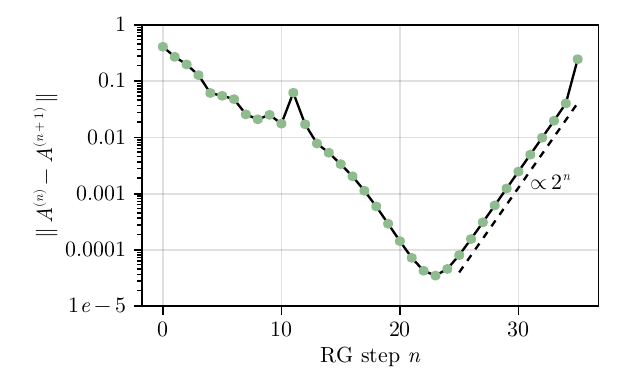}
	\caption{The Hilbert-Schmidt distance between two subsequent gauge-fixed tensors as a function of the RG step. Gilt-TNR parameters are the same as in \cref{fig:sval_traj_nr}.
		\label{fig:crit_convergence}
	}
\end{figure}

Let us discuss next the tensor elements of the approximate fixed point tensor, defined as the tensor
$A^{(n)}$ with $n=n_*$ corresponding to the minimum in the curve in
\cref{fig:crit_convergence}. From the above discussion we expect that $A^{(n_*)}$ approximates the fixed point $A^{[30]}_*$ of the $\chi=30$ Gilt-TNR map\footnote{See \cref{rem:chi-fp}. Apart from $\chi=30$, the fixed point $A^{[30]}_*$ also depends of course on other parameters of Gilt-TNR  such as $\eps_{\rm gilt}$ but we leave this dependence implicit.} with accuracy
\beq
\|A^{(n_*)}-A^{[30]}_*\|\sim 5\times 10^{-5}\,.
\label{eq:5e5}
\eeq
It is also interesting to record the distance from the critical nearest-neighbor Ising tensor from which we started the RG evolution:
 \beq
 \|A^{(n_*)}-A_{\rm NN}(1)\|\approx 0.90\,.
 \label{eq:d12}
 \eeq
 
The first few largest tensor elements of $A^{(n_*)}$ are reported in \cref{tab:fp30}, Column 2. (Column 3 of the same table gives the approximate fixed point tensor of the RG map with rotation, to be discussed in \cref{sec:with-rot}. Please ignore them for now.)

We divided the tensor elements of $A^{(n_*)}$ in \cref{tab:fp30} into several groups. If the fixed point tensor preserved reflection and/or rotation invariance, tensor elements in each group would be equal, perhaps up to a sign. Instead, we see that the tensor elements within groups are close but not exactly equal. How should we interpret this? Recall that we started the RG evolution from the isotropic tensor $A^{(0)}$ which is reflection and rotation invariant. However, as we already mentioned, Gilt-TNR does not preserve these symmetries (although it does preserve the spin flip $\mathbb{Z}_2$). That's why the fixed point tensor in \cref{tab:fp30} does not respect these symmetries.\footnote{\label{note:horvert}One can also see from \cref{tab:fp30} that the horizontal reflection is broken by a smaller amount than rotations and the vertical reflection. E.g.~in column 2 we have $A_{1^+1^+1^-1^-}\approx
		A_{1^-1^+1^+1^-}$, 
		$A_{1^-1^-1^+1^+}\approx
		A_{1^+1^-1^-1^+}$ (horizontal reflections) hold to a much better accuracy than say $A_{1^+1^+1^-1^-}\approx A_{1^-1^-1^+1^+}$ (vertical reflection). A likely reason for this is given in \cref{note:horvert1}.}

\begin{table}
	\centering
	\begin{NiceTabular}{cc@{\hskip 1em}|c}
			\toprule
		$	ijkl $        & $\left(A^{(n_*)}\right)_{ijkl} $ & $\left(A_{(m_*)}\right)_{ijkl}$\\
			\midrule
		$	1^+1^+1^+1^+$ & 0.3411 & 0.3520  \\
			\dashedline
		$	1^+1^+1^-1^- $& 0.2107 &  0.2119\\
		$	1^-1^+1^+1^- $& 0.2105  & 0.2027\\
		$	1^-1^-1^+1^+ $& 0.1969   & 0.2042\\
		$	1^+1^-1^-1^+ $& 0.1971   & 0.2135\\
			\dashedline
		$	1^-1^-1^-1^- $& 0.1807   & 0.1829\\
			\dashedline
		$	1^+1^-1^+1^- $& 0.1481   & 0.1474\\
		$	1^-1^+1^-1^+ $& 0.1438   & 0.1503\\
			\dashedline
		$		1^+1^+2^-2^-$ & $-0.1067$ & $-0.1056$ \\
		$		2^-1^+1^+2^- $& 0.1068   & 0.1024\\
		$		1^+2^-2^-1^+ $& 0.1015   & 0.1069\\
		$		2^-2^-1^+1^+ $& $-0.1016$  & $-0.1033$\\
				\dashedline
		$		1^-1^+1^-2^+ $& 0.1048   & 0.08593\\
		$		2^+1^-1^+1^- $& 0.08648  & 0.08152\\
		$		1^+1^-2^+1^- $& 0.08705  & 0.09959\\
		$		1^-2^+1^-1^+ $& 0.07821  & 0.08898\\
				\bottomrule
			\end{NiceTabular}
	\caption{
		\label{tab:fp30} The first few largest tensor elements of the (approximate) fixed point tensor. The tensor legs are ordered as left-top-right-bottom. The leg indices are numbered in the notation $p^\pm$ where $\pm$ is the $\mathbb{Z}_2$ quantum number and $p=1,2,3,\ldots$ is the index in the corresponding $\mathbb{Z}_2$ sector on the corresponding leg. Column 1: tensor element indices. Column 2: the approximate fixed point tensor $A^{(n_*)}$ of the non-rotating Gilt-TNR. Column 3: the approximate fixed point tensor $A_{(m_*)}$ of the rotating Gilt-TNR, to be discussed in \cref{sec:with-rot}.}
\end{table}

In \cref{fig:tensor_elements_magnitudes} we show absolute values of tensor elements of approximate fixed point tensors for $\chi=30$ as well as for $\chi=10,20$, for comparison.\footnote{For $\chi=10,20$ we follow the same procedure as for $\chi=30$. First, we find approximate $t_c$ by bisecting to $10^{-10}$ and looking for the longest plateau in singular values as in \cref{fig:sval_traj_nr}. We find $t_{\chi=10}=1.0012778632$ ($\eps_{\rm gilt}=10^{-4}$) and $t_{\chi=20}=1.0000618664$ ($\eps_{\rm gilt}=2\times 10^{-5}$). Then, we produce a plot of Hilbert-Schmidt distances between subsequent gauge-fixed tensors, as in \cref{fig:crit_convergence}. The approximate fixed point tensor is defined as $A^{(n)}$ corresponding to the minimum of this plot.} Absolute values of tensor elements are sorted in descending order and indexed as $|A|_I$ where $I=1,2,\ldots$ is the position in the sorted list. We then plot $|A|_I$ as a function of $I$. The left figure shows the first 50 tensor elements, while the right figure gives a global of view of all tensor elements, in the log-log scale. The dashed line in this figure corresponds to the slope $|A|_I\sim I^{-1/2}$, which marks the boundary for the finiteness of the Hilbert-Schmidt norm for infinite-dimensional tensors. The curves remain below that line for the values of  $\chi$ that we studied, but just barely so. We cannot determine from this plot whether or not our finite $\chi$ fixed points will converge to a fixed point tensor with finite Hilbert-Schmidt norm as $\chi$ goes to infinity. It may be that the Gilt-TNR algorithm does not have a fixed point with a finite Hilbert-Schmidt norm in the $\chi\to\infty$ limit. If so, the search for a rigorous infinite-dimensional fixed point will need to employ a different tensor network RG map. We will keep using Gilt-TNR in the rest of this work to demonstrate other ideas such as the effect of rotation and the possibility to set up the Newton method. Those ideas are general and will still apply if in the future we have to change the algorithm. See also Remark \ref{rem:chi-fp} and Section \ref{sec:conclusions} for further related comments.

\begin{figure}
	\centering
	\includegraphics[width=\textwidth]{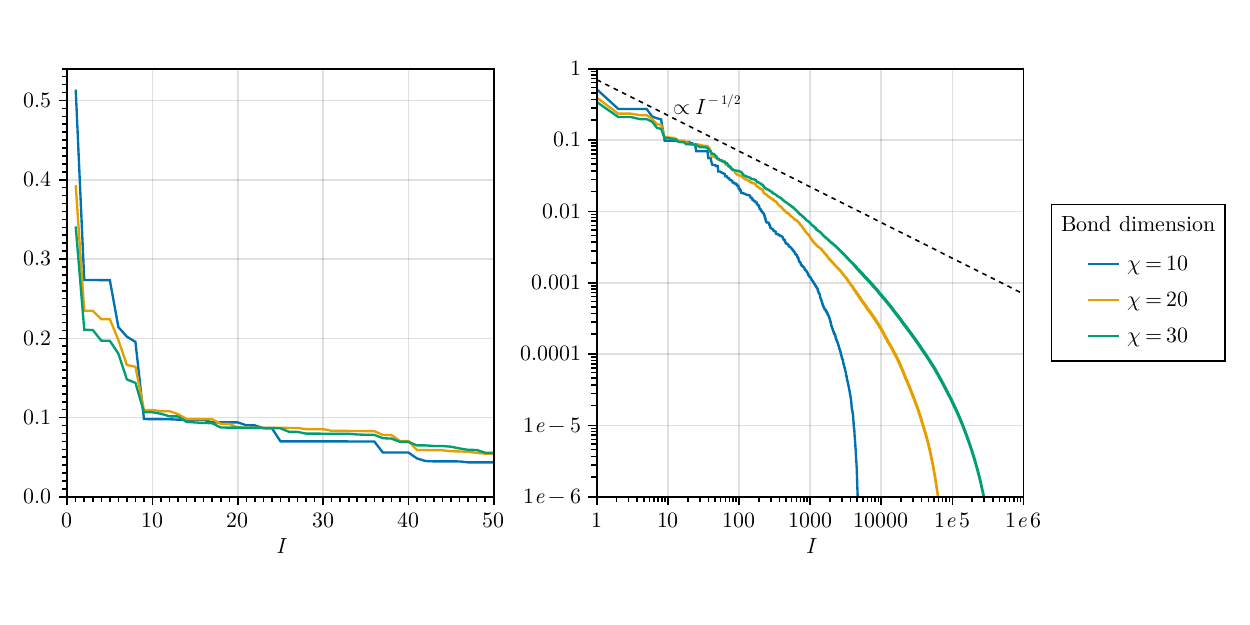}
	\caption{Absolute values of tensor elements of approximate fixed point tensors for $\chi=10,20,30$, see the text.
		\label{fig:tensor_elements_magnitudes}
	}
\end{figure}

\subsection{Jacobian eigenvalues}
\label{sec:jac-norot}
Now that we provided evidence for the existence of a fixed point tensor, we would like to explore the RG map Jacobian at the approximate fixed point. First of all we can speak of the Jacobian as our map is differentiable in a neighborhood of the fixed point. Evidence for this is provided in \cref{app:GiltDiff}. This statement is less trivial than it may seem. To make Gilt-TNR robustly differentiable we had to tweak the algorithm a bit. In the original Gilt-TNR (see \cref{app:moreGilt} where this is referred to as "dynamic $Q$-choice"), matrices $Q_1$, $Q_2$, $Q_3$, $Q_4$ are chosen iteratively, and the number of iterations is determined at runtime depending on a certain internal parameter called $\eps_{\rm conv}$. For constant $\eps_{\rm conv}$, the number of iterations may jump for small variations of tensor $A$ causing a tiny but noticeable discontinuity (and hence non-differentiability) in $\mathcal{R}(A)$. 
To avoid this problem, we are using a modified algorithm in which the number of iterations is kept fixed, referred to as ``static $Q$-choice'' in \cref{app:moreGilt}.\footnote{\label{note:TNRdanger}Other tensor RG algorithms involving disentanglers, such as e.g.~TNR \cite{Evenbly-Vidal}, also involve iterative optimization of disentanglers, to minimize a cost function expressing a truncation error. It would be interesting to see if those algorithms are, or may be made, continuous or differentiable in a small neighborhood of the critical fixed point. To our knowledge this has never been tested carefully.}

So, let $R$ be the non-rotating Gilt-TNR RG map with $\chi=30$ and other parameters as in \cref{fig:sval_traj_nr}. Let $\nabla R$ be the Jacobian of $R$ at the tensor $A^{(n_*)}$ from \cref{tab:fp30}. We are interested in the eigenvalues of $\nabla R$. 

Note that tensor $A^{(n_*)}$ preserves $\mathbb{Z}_2$ flip symmetry, in the sense that each of its vertical and horizontal indices can be assigned a $\pm 1$ quantum number, such that the tensor remains invariant when acted upon by all these sign flips, i.e., the tensor is $\mathbb{Z}_2$-even.\footnote{A tensor $A$ is $\mathbb{Z}_2$-even or $\mathbb{Z}_2$-odd if
$A_{ijkl} s_i s_j s_k s_l = \pm A_{ijkl}$ where $s_i,s_j,s_k,s_l$ are the quantum numbers of the four legs and the sign $\pm$ is $+$ for even, $-$ for odd.}

  The full tangent space at $A^{(n_*)}$, on which $\nabla R$ acts, splits into a direct sum of subspaces of perturbations which are $\mathbb{Z}_2$-even and $\mathbb{Z}_2$-odd. Since $R$ preserves $\mathbb{Z}_2$, $\nabla R$ maps  $\mathbb{Z}_2$-even perturbations to $\mathbb{Z}_2$-even, and similarly for $\mathbb{Z}_2$-odd. Each eigenvector will be either $\mathbb{Z}_2$-even or $\mathbb{Z}_2$-odd.

The first few largest eigenvalues of $\nabla R$ at $A^{(n_*)}$ are reported in \cref{tab:isospec1}, column "$\lambda$ of $\nabla R$." Since $\nabla R$ is a large matrix of size $\chi^4\times \chi^4$, its full diagonalization would be too time consuming. We instead compute $O(10)$ largest eigenvalues by the Arnoldi iteration method, implemented in the {\tt Julia} \cite{Julia,Julia-2017} package {\tt KrylovKit.jl} \cite{haegeman_krylovkit_2024}. The Arnoldi method does not require evaluation and storage of the full matrix $\nabla R$. Instead, one provides a function $v\mapsto\nabla R.v$, which computes the directional derivative of $R$ in the direction $v$. We evaluate this directional derivative via a symmetric finite difference approximation:
\begin{equation}
	D_h = \frac{R(A+hv)-R(A-hv)}{2h},
	\label{eq:symfinite}
\end{equation}
where $h=10^{-4}$ is a small parameter, set at this value to balance truncation and roundoff errors; see \cref{app:GiltDiff}.\footnote{Computing $O(10)$ largest eigenvalues and the corresponding eigenvectors of $\nabla R$ by this method requires $O(100)$ calls of the RG map, and for $\chi=30$ takes a few minutes on a laptop.}

\definecolor{light-blue}{HTML}{D1EDFF}
\definecolor{light-red}{HTML}{FF999C}
\begin{table}[]
	\centering
	\begin{NiceTabular}{c|cc|c|c}
		\CodeBefore
		\rowcolor{light-blue}{8}
		\rowcolor{light-red}{2,3,4}
		\Body
		\toprule
		$\mathbb{Z}_2$ & $\mathcal{O}$ & $\lambda_{\rm CFT}=2^{2-\Delta_\mathcal{O}}$ & $\lambda$ of $\nabla R$ &  $\lambda$ of $\nabla R^{\circ}$ \\
		\midrule
		$+$            & $\epsilon$            & 2                                            & 1.9996                  & 1.9996                          \\
		$+$            & $T$                   & 1                                            & 1.0015                  &  $-1.0010  $                       \\
		$+$            & $\overline{T}$             & 1                                            & 0.9980                  & $ -0.9982 $                        \\
		$+$            &                       &                                              & 0.6322                  &  0.7757                          \\
		$+$            &                       &                                              & $0.5941\pm 0.1195 i$    &  0.6209                          \\
		&                       &                                              & \ldots                  &  \ldots                          \\
		\midrule
		$-$            & $\sigma$              & $2^{15/8}\approx 3.668016$                   & 3.6684                  &  3.668014                        \\
		$-$            &                       &                                              & 1.5328                  &  $0.0027\pm 1.5364i$             \\
		$-$            &                       &                                              & 1.5300                  &  0.8600                          \\
		$-$            &                       &                                              & 0.8869                  &  $0.6318 \pm 0.0583i$            \\
		&                       &                                              & \ldots                  &  \ldots                      
		\\\bottomrule   
	\end{NiceTabular}
	\caption{Column 1: $\mathbb{Z}_2$ quantum number. Column 2,3 (colored rows): the first few low-dimension CFT quasiprimaries, and their exact eigenvalues. Columns 4,5: The first few largest, in absolute value, Jacobian eigenvalues at the approximate fixed point for the non-rotating Gilt-TNR ($\nabla R$) and for the rotating Gilt-TNR ($\nabla R^\circ$). Uncolored rows show RG eigenvalues corresponding to derivative interactions, which are not universal.}
	\label{tab:isospec1}
\end{table}

Colored rows of \cref{tab:isospec1} correspond to CFT quasiprimaries. The corresponding Jacobian eigenvalues can be compared to the CFT predictions in columns 2,3, computed via Eq.~\eqref{eq:lambdaCFT}. We observe  reasonably good agreement, both in the $\mathbb{Z}_2$-even and in the $\mathbb{Z}_2$-odd sectors. Note in particular the two eigenvalues close to 1 corresponding to the stress tensor components $T$ and $\overline T$. 

In the few uncolored rows in \cref{tab:isospec1} we report eigenvalues which do not correspond to quasiprimaries on the CFT side. These rows correspond to derivative operators, whose eigenvalues are not universal, as discussed in \cref{sec:CFTexp}. For this reason, the CFT entries in these rows are left empty.\footnote{Note that the Jacobian matrix is real but in general not symmetric, and so some eigenvalues come in complex-conjugate pairs.} 

To demonstrate the latter point, let us compare the Jacobian eigenvalues between $R$ and the rotating Gilt-TNR map $R^\circ$ which will be discussed in \cref{sec:with-rot} (\cref{tab:isospec1}, last column). We see that the quasiprimary eigenvalues change very little, apart from the $-1$ factor picked up by the $T$, $\overline T$ eigenvalues, as discussed in \cref{sec:addrot}. On the other hand, several eigenvalues in the uncolored rows show significant changes. For example, the leading irrelevant $\mathbb{Z}_2$-even eigenvalue $\lambda_{\rm irr}$ changes from $\approx 0.63$ to $\approx 0.78$.

The spectrum of even lower eigenvalues beyond those shown in \cref{tab:isospec1} becomes rather messy. Unfortunately we do not have enough precision to compare with the irrelevant quasiprimaries in the $\mathbb{Z}_2$-even and $\mathbb{Z}_2$-odd sectors. Their eigenvalues are buried among the non-universal eigenvalues corresponding to  derivative operators.

\subsection{Comparison to prior shooting method results}
\label{sec:Tokyo-comp}

Our results in Sections \ref{sec:Isofp} and \ref{sec:jac-norot} can be compared with the results of Ref.~\cite{Evenbly-Vidal,Lyu:2021qlw,PhysRevE.109.034111}. This was already partly summarized in \cref{tab:accuracy} (upper part of the table). We will give a detailed comparison, divided into two parts: locating the fixed point via the "shooting" method, and computing the Jacobian eigenvalues.

\subsubsection{Locating the fixed point}

In this section we review past results on the accuracy of approximate fixed points and compare with our results. We emphasize that our goal here is not to compare Gilt-TNR with other algorithms, but rather to compare the accuracy of our Newton method for finding the fixed point with past results which use a shooting method to find it. 

Ref.~\cite{Evenbly-Vidal,Lyu:2021qlw,PhysRevE.109.034111} are the only prior works known to us which implemented gauge fixing and plotted the norm difference of two consecutive tensors as a function of the RG step, like our  \cref{fig:crit_convergence}.\footnote{Other works \cite{TNRScaleTrans,Bal:2017mht} said they could reach an approximate fixed point tensor by gauge fixing, but without a plot like \cref{fig:crit_convergence} the accuracy is hard to know. Ref.~\cite{2023arXiv230617479H} quantified the rate of convergence to the fixed point by plotting the norm difference of the vectors of singular values for consecutive tensors as a function of the RG step, but without performing gauge fixing.} 
Refs.~\cite{Lyu:2021qlw,PhysRevE.109.034111} used methods closely related to ours, so we discuss them first.

 Ref.~\cite{Lyu:2021qlw} used an algorithm closely related to Gilt-TNR, which we call Gilt-HOTRG, as it uses Gilt for disentangling, and HOTRG for coarse-graining (see \cref{app:Tokyo}). It is an algorithm closely related to Gilt-TNR. Several parameter choices were identical for us and for them: $\chi=30$, $\epsilon_{\rm gilt}=6\times 10^{-6}$.
 
 Ref.~\cite{PhysRevE.109.034111} also used essentially Gilt-HOTRG with some minor modifications. They use $\chi=24$, $\epsilon_{\rm gilt}=2\times 10^{-7}$. They also split an RG step in two and only use disentangling on odd steps. This doubling of the number of steps should be kept in mind when comparing with the plots of \cite{PhysRevE.109.034111}.
 
 Ref.~\cite{Lyu:2021qlw} used gauge fixing similar to the one used by us, while Ref.~\cite{PhysRevE.109.034111} used a more sophisticated two-step gauge fixing. In the first step they gauge fix by general invertible matrices to achieve the Minimal Canonical Form (MCF) \cite{Acuaviva:2022lnc}, which minimizes the norm of the gauge-fixed tensor. In the second step they gauge fix by orthogonal matrices, similarly to \cite{Lyu:2021qlw} (with some improvements) and to us.
 
The norm-difference plot \cite[Fig.~3, red curve]{PhysRevE.109.034111}, is rather similar in shape to our \cref{fig:crit_convergence}, although their minimal distance $\delta\sim 3\times 10^{-4}$ is larger than our $\delta\sim 5 \times 10^{-5}$. This may be due to them locating the critical temperature not as precisely as us.

On the other hand the norm difference plot \cite[Fig.~9(b)]{Lyu:2021qlw} reaches a much worse minimal distance $\delta\sim 10^{-2}$ in spite of them locating the critical temperature to the same $\Delta t=10^{-10}$ accuracy as us. Ref.~\cite{PhysRevE.109.034111} investigated this issue. In their Fig.~13 they reproduce the results of \cite{Lyu:2021qlw} and note that the rate of approach to the fixed point tensor is notably slower for the algorithm of \cite{Lyu:2021qlw} than for their algorithm. In the language of \cref{sec:Isofp}, we could say that the algorithm of \cite{PhysRevE.109.034111} has $\lambda_{\rm irr}$ about the same as our algorithm, while the algorithm of \cite{Lyu:2021qlw} has $\lambda_{\rm irr}$ closer to 1. This interpretation is conjectural since \cite{Lyu:2021qlw,PhysRevE.109.034111} don't report their $\lambda_{\rm irr}$.

Therefore, the MCF step in gauge fixing, which is the only significant difference between the algorithms of Refs.~\cite{Lyu:2021qlw} and \cite{PhysRevE.109.034111}, is responsible for the improvement of the convergence rate to the fixed point in their case. On the other hand it should be noted that our algorithm is able to reach the same convergence rate without MCF but using Gilt-TNR instead of Gilt-HOTRG. There is a little bit of luck involved here, the leading irrelevant eigenvalue $\lambda_{\rm irr}$ being non-universal.

Finally, let us discuss the earlier Ref.~\cite{Evenbly-Vidal}. Using the TNR algorithm that they proposed in that work, they reached $\sim 10^{-5}$ accuracy in a plot like our \cref{fig:crit_convergence}, using the TNR bond dimension $\chi_{\rm TNR}=42$ \cite[Suppl.~Mat., Fig.~6]{Evenbly-Vidal}. While typically the TNR algorithm needs a smaller bond dimension than Gilt-TNR to achieve the same accuracy \cite{GILT}, here it's the other way around, likely because they did not bisect to get closer to the $\chi$-dependent $T_c$, but started the RG flow at the theoretical $T_c$. A feature of the plots in \cite{Evenbly-Vidal} and in \cite{Lyu:2021qlw}, possibly due to imperfect gauge-fixing, is that they flatten near the minimum, compared to the clean V-shaped plots in our work and in \cite{PhysRevE.109.034111}.

\subsubsection{Jacobian eigenvalues}

In \cref{sec:jac-norot} we studied the Jacobian of Gilt-TNR close to the critical point. In particular, we reported the leading irrelevant eigenvalue $\lambda_{\rm irr}$ and observed that the rate of approach to the fixed point is consistent with it. Since Refs.~\cite{Lyu:2021qlw,PhysRevE.109.034111} do not report the spectrum of the Jacobian of the map used to produce their norm-difference plots, the $\lambda_{\rm irr}$ test is not possible based on their available data. Such a test would be interesting in the future. 

Refs.~\cite{Lyu:2021qlw,PhysRevE.109.034111} do study the spectrum of a tensor RG map Jacobian (procedure called there "linearized tensor RG", or "lTRG"), but only for a different RG map,  constructed once the (approximate) fixed point is known. This new ``frozen'' (our term) RG map has the Gilt disentanglers and the HOTRG isometries frozen to their fixed point values.\footnote{The change in the RG map can be inferred from the linearization expression \cite[Eq.~(40)]{Lyu:2021qlw} and is also stated explicitly in \cite{PhysRevE.109.034111} below Eq.~(18).}  

Ref.~\cite{Lyu:2021qlw} uses the spectrum of the frozen RG map Jacobian to extract the scaling dimensions of CFT operators with $\Delta\le 2$ \cite[Table I]{Lyu:2021qlw}, and Ref.~\cite{PhysRevE.109.034111} extends this analysis to $\Delta\le 3$ \cite[Table I]{PhysRevE.109.034111}. Compared to our $\nabla R$ spectrum, their accuracy is one order of magnitude worse for $\sigma$ and $\eps$, and is about the same for $T$, $\overline T$. Our better performance for $\sigma$ and $\epsilon$ may be due to our better approximation of the fixed point tensor.

Their results contain one curious feature: in addition to quasiprimaries $\sigma$, $\epsilon$, $T$, $\overline{T}$, it contains many derivative operators identified with first and second derivatives of $\sigma$, $\epsilon$ and first derivatives of $T$, $\overline{T}$, whose scaling dimensions are correctly reproduced (up to numerical errors) from frozen RG eigenvalues. This is puzzling: we emphasized many times that eigenvalues of derivative operators are not universal and do not have to agree between RG and CFT, and indeed they do not agree for us. So why do they agree for them? This was not discussed in \cite{Lyu:2021qlw,PhysRevE.109.034111}.

The answer to this puzzle is instructive. It turns out that the frozen RG map of \cite{Lyu:2021qlw,PhysRevE.109.034111} is very special: its Jacobian coincides with the lattice dilatation operator. This coincidence, which does not hold for our RG maps $R$ and $R^{\circ}$, explains the puzzle, because lattice dilatation operator eigenvalues are universal and reproduce all CFT operator eigenvalues, including total derivatives. It would require too much preparatory work to explain this point in detail here; it will be done in \cite{paper-DSO}.
 
We should add that the frozen RG map is only defined once one has already found an approximate fixed point, which can only be done with a full RG map where disentanglers and isometries depend on the tensor. So while the frozen RG map is an easier map and it has a conceptually simpler Jacobian, it would not be useful for our purposes of setting up a Newton method.

\subsection{Anisotropy: manifold of fixed points}
\label{sec:anisotropy}

In the previous sections we studied the RG flow starting from the isotropic 2D Ising model. We would now like to show the new features which appear when introducing anisotropy. Namely, we consider the nearest neighbor Ising model with anisotropy parameter $\Jratio=J_y/J_x \leq 1$. We transform the partition function of this model to a tensor network as detailed in \cref{app:init}. The resulting tensor has bond dimension $2$ and depends on the reduced temperature $t$ and the anisotropy parameter $\Jratio$. Starting from this tensor $A^{(0)}=A_{\rm NN}(t,\Jratio)$ we repeatedly apply the non-rotating Gilt-TNR algorithm to get a sequence $A^{(n)}(t,\Jratio)$, $n=1,2,\ldots,$ see \cref{eq:RGevolution}.

We have two goals in this section: 1) We want to verify \cref{fig:exp1} showing that $\Jratio$ parametrizes a one-dimensional manifold of fixed points. To our knowledge, this has never been checked using tensor RG. 2) We want to demonstrate that eigenvalues of $\nabla R$ corresponding to quasiprimaries are universal across this manifold and eigenvalues that do not correspond to quasiprimaries can vary.

\begin{figure}
	\centering
	\includegraphics[width=0.7\textwidth]{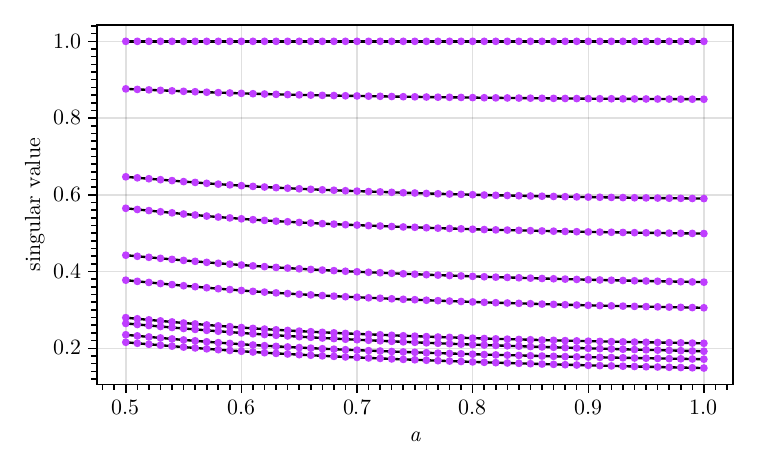}
	\caption{Dependence of the first few singular values of $A^{(20)}(t \approx t_c(a), \Jratio)$ on $\Jratio$. The singular values are obtained via SVD decomposition as in \cref{fig:singval}. The largest singular value is normalized to 1 as in \cref{fig:sval_traj_nr}.
		\label{fig:spec_of_A_vs_J}
	}
\end{figure}

Let us start with the first goal. As discussed in \cref{sec:rotation-symmetry-breaking}
we expect a two-dimensional manifold of RG fixed points for a non-rotating RG map. The critical model corresponds to $t=1$ in our notation, for all values of $a$. However, as in the isotropic case, for finite bond dimension the value of $t$ that gives convergence to the approximate fixed point is not exactly $1$ and will depend on $a$. We denote it by $t_c(a)$.\footnote{Approximate values of $t_c(a)$ were obtained by the bisection method as explained after \cref{fig:singval} with $\Delta t=10^{-10}$.}
If we start with a critical anisotropic nearest neighbor Ising model ($t=t_c(a)$), then iterations of the RG map will converge to a fixed point that depends on the anisotropy parameter $\Jratio$. To show that different values of $\Jratio$ do indeed produce different fixed points, we consider the singular values of the fixed point decomposed along a diagonal as defined in \cref{fig:singval}. If we plot these singular values as a function of the RG step, we see curves that are qualitatively similar to \cref{fig:sval_traj_nr}. In particular for $t \approx t_c(a)$ there is a range of RG steps where the singular values are almost constant. However, these constant values depend on $\Jratio$. We compute these singular values at the 20th RG step for $\Jratio$ ranging from $0.5$ to $1.0$ and $t \approx t_c(a)$, and plot them as a function of $\Jratio$ in \cref{fig:spec_of_A_vs_J}. The variation of these values with $\Jratio$ shows that as we vary $\Jratio$, the RG fixed point is varying along a one-dimensional manifold. We expect that by considering a more general initial Ising model with two parameters we could obtain the full two-dimensional manifold of fixed points.

\begin{figure}
	\centering
	\includegraphics[width=0.7\textwidth]{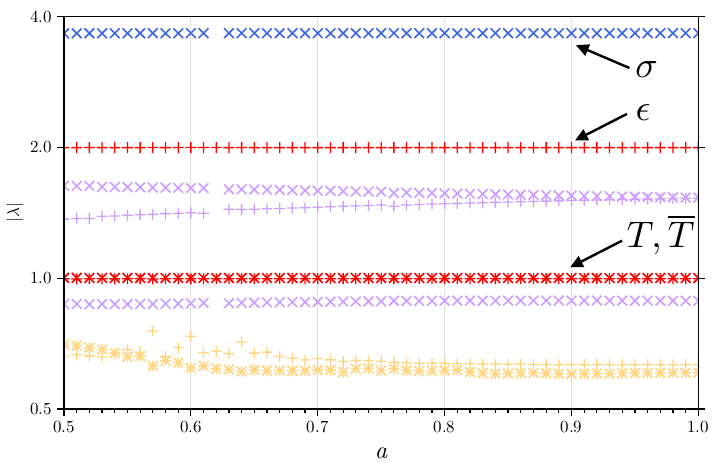}
	\caption{The first few largest eigenvalues of $\nabla R$ computed at $A^{(20)}(t\approx t_c(a), \Jratio)$ as a function of the anisotropy parameter $\Jratio$. Note the vertical $\log_2$-scale. We use the following color scheme: red/blue for $\mathbb{Z}_2$-even/odd quasiprimaries, light-orange/light-violet for other $\mathbb{Z}_2$-even/odd eigenvalues. For $a=1$ the spectrum agrees with the ``$\lambda$ of $\nabla R$'' column of \cref{tab:isospec1}.
		\label{fig:eigenvals_vs_Jratio} }
\end{figure}

Let us proceed with the second goal. As discussed in \cref{sec:CFTexp} the eigenvalues of the RG Jacobian $\nabla R$ should be of two types: those that correspond to quasiprimaries in the CFT and those that correspond to total derivatives. The former are given by \cref{eq:lambdaCFT} and so should not depend on which fixed point in the two-dimensional manifold we consider. The latter cannot be predicted from the CFT and may depend on the fixed point at which we compute the Jacobian. In \cref{fig:eigenvals_vs_Jratio} we plot the first few eigenvalues of $\nabla R$ as a function of $\Jratio$. The eigenvalues corresponding to quasiprimaries are shown in red/blue for the $\mathbb{Z}_2$-even/odd sectors. These eigenvalues are essentially independent of $\Jratio$ and fit the CFT predictions of $2^{15/8}$, $2$ and $1$ very well. The eigenvalues corresponding to total derivatives are shown in light-orange/light-violet for the $\mathbb{Z}_2$-even/odd sectors. For these eigenvalues one can see significant variation with $\Jratio$. For $a=1$ the spectrum agrees with the ``$\lambda$ of $\nabla R$'' column of \cref{tab:isospec1}.\footnote{For all values of $\Jratio$ in \cref{fig:eigenvals_vs_Jratio} there is a pair of complex-conjugate eigenvalues in the lower part of the plot. For most values of $\Jratio$ the pair is the 9th and 10th eigenvalues, but for a few values of $\Jratio$ it is the 8th and 9th.  For complex eigenvalues we plot their absolute value. Since the two eigenvalues in this pair have the same absolute value, the result is an ``$\times$'' and a ``+'' right on top of each other.}

We thus see broad agreement with the intuition expounded in \cref{sec:intuition}.

\begin{remark} In \cref{fig:eigenvals_vs_Jratio} there are no $\mathbb{Z}_2$-odd eigenvalues shown for $\Jratio=0.62$. For this value of $\Jratio$ the computation of the $\mathbb{Z}_2$-odd eigenvalues produced nonsensical results which we do not show. We believe that this is a result of the issues with discrete gauge fixing in the $\mathbb{Z}_2$-odd sector that are discussed in \cref{sec:comment_gauge_fix}. 
\end{remark}

\section{Results for the rotating Gilt-TNR}
\label{sec:with-rot}

Here we will finally present what we consider our most exciting numerical results. We will consider the rotating Gilt-TNR algorithm. We will describe how the critical fixed point is found with the Newton method, in agreement with the general discussion in \cref{sec:intuition}. We will study the rate of convergence of the Newton method, and we will use the fixed point to extract CFT quasiprimary dimensions using the Jacobian. All the above will be done in \cref{sec:rot-iso} for the isotropic Ising model, and the discussion here will parallel Sections \ref{sec:Isofp} and \ref{sec:jac-norot} for the non-rotating algorithm. (The 3-state Potts model will be similarly considered in \cref{app:the-3-state-potts-model}.) We will consider the anisotropic Ising model in \cref{sec:oscillations} where we exhibit period-2 oscillations for the rotating Gilt-TNR algorithm when starting from an anisotropic critical tensor. 

\subsection{Numerical strategies for the inverse Jacobian}
\label{sec:inv-jac}

We start by describing how to deal numerically with the inverse Jacobian in the Newton iteration map \eqref{eq:newton}. We need to solve the linear system
\beq 
J(A) x=b,\quad J(A)\equiv I-\nabla R^\circ(A)\,,
\label{eq:linear-solve}
\eeq
where $A=A_{(m)}$, $b=f(A_{(m)})$. 

Our code does not store the Jacobian $\nabla R^\circ(A)$ as a matrix, but realizes it as a linear operator $v\mapsto \nabla R^\circ(A) v$, computed with a second-order finite difference approximation (see \cref{sec:jac-norot} for a fully analogous discussion for the non-rotating RG map). We may then solve \eqref{eq:linear-solve} via an iterative algorithm, GMRES \cite{GMRES} being the standard choice. For unclear reasons, in our numerical tests using {\tt KrylovKit.jl} \cite{haegeman_krylovkit_2024}, GMRES convergence was not always reliable. We therefore used a different strategy, which we found numerically stable.

We use the Arnoldi method to find $s$ largest in absolute value eigenvalues of $\nabla R^\circ(A)$, and the corresponding eigenvectors. Let $P_{s}$ be the orthogonal projector on the subspace spanned by them. We then approximate the Jacobian by 
\beq
\label{eq:approxJ}
J_\approx(A) = I- P_s \nabla R^\circ(A) P_s.
\eeq
We take $s\ge 3$, so that the relevant eigenvalue corresponding to the $\epsilon$ and the two $-1$ eigenvalues corresponding to $T,\bar{T}$ are subtracted. The remaining eigenvalues of $\nabla R^\circ(A)$ being less than 1 in absolute value, the approximation is expected to be good. Increasing $s$ further improves the approximation and speeds up convergence, as we will see below.

Computing the approximation \eqref{eq:approxJ} with the Arnoldi method takes a few minutes on a laptop for $s=10$, $\chi=30$. Within this approximation, solving linear system \eqref{eq:linear-solve} is easy as it amounts to inverting an $s\times s$ matrix. We checked that the results agree closely with GMRES, when the latter converges. We also checked that the accuracy of this approximation, and the agreement with GMRES, are further improved by including the first few terms of the Neumann series expansion in the difference $J(A)-J_\approx(A)$. This improvement is not used in the results shown below. 

To speed up computations, we use one more approximation. Namely, we do not recompute the approximate Jacobian $J(A)$ at every Newton step, but compute it just once, for $A=A_{(0)}$. This preserves good local convergence properties of the Newton method.

\subsection{Isotropic 2D Ising fixed point by the Newton method}
\label{sec:rot-iso}

So let $R^\circ$ be the rotating Gilt-TNR map, constructed by adding a $\pi/2$ rotation to the non-rotating Gilt-TNR. The numerical results below are obtained using the same Gilt parameters as in \cref{sec:Isofp}: $\chi=30$, $\epsilon_{\rm gilt}=6\times 10^{-6}$.

We would like to solve the fixed point equation for $R^\circ$ via the Newton method. We set up the Newton method as in \cref{sec:newton} and compute the inverse Jacobian as described in \cref{sec:inv-jac}. The key ingredients here are the Jacobian $\nabla R^\circ$ and its eigendecomposition, which are computed similarly to the non-rotating Gilt-TNR, \cref{sec:jac-norot}.

Ideally, we would like the Newton method to converge to the fixed point when starting from a generic point on the critical manifold, e.g.~from the bond dimension 2 tensor $A_{\rm NN}(t=1)$ obtained when translating the critical isotropic 2D Ising model to the tensor network form (\cref{app:init}). However it turns out somewhat nontrivial to set up such a faraway convergent Newton method, for reasons discussed in \cref{sec:attempts} below. So in this work we settle for convergence in a smaller domain. Our search for the fixed point will consist of two stages: 
\begin{enumerate}
	\item
	Find an approximate fixed point tensor $A^{(n')}$, obtained by performing $n'$ RG steps starting from $A_{\rm NN}(t_0)$, where $t_0$ is near the finite-$\chi$ critical temperature $t_c\approx 1$. In our main results we use $n'=23$ initial RG steps from $t_0=1.00001313120$,  obtained by bisecting to $\Delta t = 10^{-10}$.\footnote{Tensor $A^{(n')}$ in this section should not be confused with tensor $A^{(n_*)}$ in \cref{sec:Isofp}. They are different tensors because the RG map is different. Also, $n_*=23$ in \cref{sec:Isofp} which was chosen to minimize $\|R(A^{(n)})- A^{(n)}\|$, while here we choose $n'=23$ simply as a sufficiently large number.} These parameters can be somewhat relaxed; see \cref{sec:attempts} below. 
	
	\item 
	Start Newton iterations from $A^{(n')}$, which are then seen to converge to a fixed point.
\end{enumerate}

	Denoting $A_{(0)}=A^{(n')}$, the Newton iterations are given by
	\beq
	A_{(m+1)}= g_\approx(A_{(m)}),\qquad n=0,1,2,\ldots,
	\label{eq:Newton-iter}
	\eeq
	where $g_\approx$ is the approximate Newton map, obtained by replacing the exact Jacobian $J(A)$ by the approximate Jacobian $J_\approx$ with projector rank $s\ge 3$ and evaluated at $A_{(0)}$, as described in \cref{sec:inv-jac}. We will see below how various choices of $s$ influence the performance.
	
Before starting the Newton iterations, it is instructive to consider $\nabla R^\circ$ eigenvalues at $A_{(0)}$. Recall (\cref{tab:isospec1}, column 3) that $\nabla R$ had two eigenvalues close to 1 at $A^{(n_*)}$: $1.0015$ and $0.9980$. In \cref{sec:addrot} we argued that adding the rotation should flip the sign of the two eigenvalues 1, which are associated with the stress tensor components. Our numerical check fully confirms that intuition: we find that $\nabla R^\circ(A_{(0)})$ has precisely 2 eigenvalues close to $-1$: $-1.0006$ and $-0.9985$, and no eigenvalues close to 1.

\begin{figure}
	\centering
	\includegraphics[width=\textwidth]{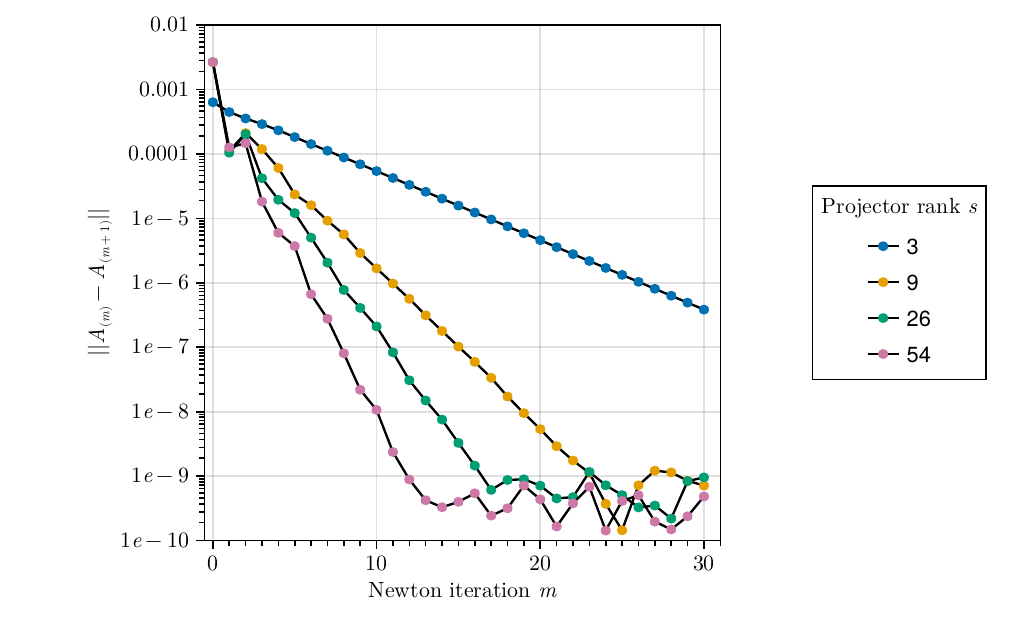}
	\caption{ 
		Convergence of the Newton method for approximations of Jacobian that use different rank $s$ of the $P_s$ projector in \cref{eq:approxJ}. 
		\label{fig:newton_conv}
	}
\end{figure}

Given this encouraging news, we expect Newton iterations \eqref{eq:Newton-iter} to converge, and this is indeed what we observe numerically. 
\cref{fig:newton_conv} demonstrates the rate of convergence, for several projector ranks $s=3,9,26,54$. In that plot, we show the Hilbert-Schmidt distance $\|A_{(m+1)}-A_{(m)}\|$ as a function of the Newton iteration number $m=0,1,2,\ldots$.
We see exponential convergence for all values of $s$, improving for larger $s$. Recall that the Newton method with the exact Jacobian is expected to converge super-exponentially fast. Using $J_{\approx}$ reduces the convergence to exponential. As usual, the rate of the exponential convergence can be understood in terms of the largest eigenvalue of the Jacobian of the map $g_\approx(A)$ at the fixed point. Actually, the convergence rate improvement with $s$ seen in \cref{fig:newton_conv} saturates around the maximal $s=54$ used in that figure. At the end of this section we will give a more detailed discussion of the convergence rate, to understand why it first improves with $s$ and then saturates, and how to improve it further.

Whatever $s$, the exponential convergence stops and is replaced by a chaotic behavior when $\|A_{(m+1)}-A_{(m)}\|$ becomes $\sim 10^{-9}$ or a bit smaller (for $s=3$ this happens outside of the shown range of $m$). We assign this behavior to roundoff errors. In \cref{app:GiltDiff}, we provide evidence that Gilt-TNR has roundoff order $\sim10^{-11}$. It is not unreasonable that this error is enhanced to $\sim10^{-9}$ in the Newton method due to additional involved operations, such as evaluating the Jacobian.

All tensors corresponding to points at the bottom of \cref{fig:newton_conv} should approximate the critical tensor equally well. We have checked their pairwise distances and they are all $\sim 10^{-9}$, see \cref{fig:pairwise}.

\begin{figure}
	\centering
	\includegraphics[scale=0.7]{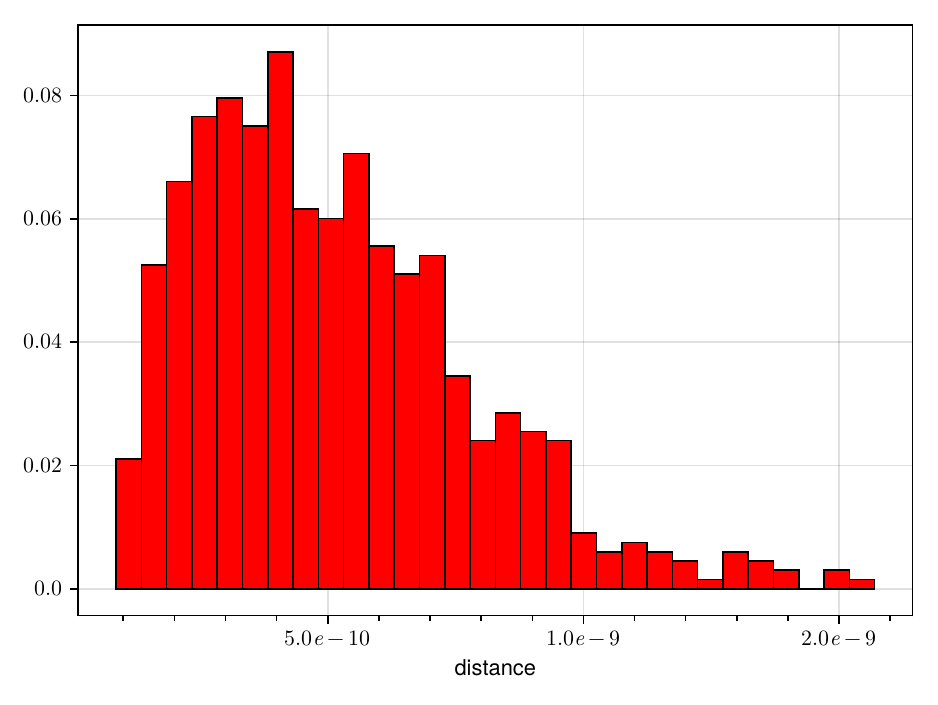}
	\caption{The histogram of distances between the $37$ tensors in \cref{fig:newton_conv} which fall below the $10^{-9}$ line.
		\label{fig:pairwise}}
\end{figure}

For further numerical tests we consider the $s=54$ Newton iteration sequence and pick from it the tensor $A_{(m_*)}$ with $m_*=14$, which is the last iteration when $\|A_{(m+1)}-A_{(m)}\|$ decreases. From the above discussion, we expect that $A_{(m_*)}$ is an excellent approximation to the (unique) fixed point $A_*^{\circ[30]}$ of rotating $\chi=30$ Gilt-TNR (for the given Gilt parameters), with error 
\beq
\| A_{(m_*)} - A_*^{\circ[30]} \| \sim 10^{-9}\,.
\label{eq:distest}
\eeq 
Applying the RG map to $A_{(m_*)}$, we obtain that this shifts the tensor by
\beq
\| R^\circ(A_{(m_*)}) - A_{(m_*)} \| \approx 3\times 10^{-10}\,,
\eeq
providing a further check of \eqref{eq:distest}. We also report the distance from the the critical nearest-neighbor Ising tensor
\begin{equation}
	\|A_{(m_*)}-A_{\rm NN}(1)\|\approx 0.89\,,
	\label{eq:d13}
\end{equation}
as well as the distance from the initial tensor for Newton iterations
\beq
\|A_{(m_*)}-A_{(0)}\|\approx 3\times 10^{-3},
\label{eq:dist_init}
\eeq
which gives an initial idea of the size of the convergence domain for our implementation of the Newton method (see \cref{sec:attempts} below).

In \cref{tab:fp30}, last column, we show the first few tensor elements of $A_{(m_*)}$. We see that they are rather similar, although not identical, to those of $A^{(n_*)}$, the approximate fixed point tensor of the non-rotating Gilt-TNR. 
The distance is, see \cref{fig:sphere}:
\begin{equation}
	\|A_{(m_*)}-A^{(n_*)}\|\approx 0.18\,.
		\label{eq:d23}
\end{equation}
This is in agreement with our general intuition that while the critical exponents are universal, RG fixed points are not---different RG maps will in general have different fixed points. 
\begin{figure}
	\centering
	\includegraphics[scale=0.5]{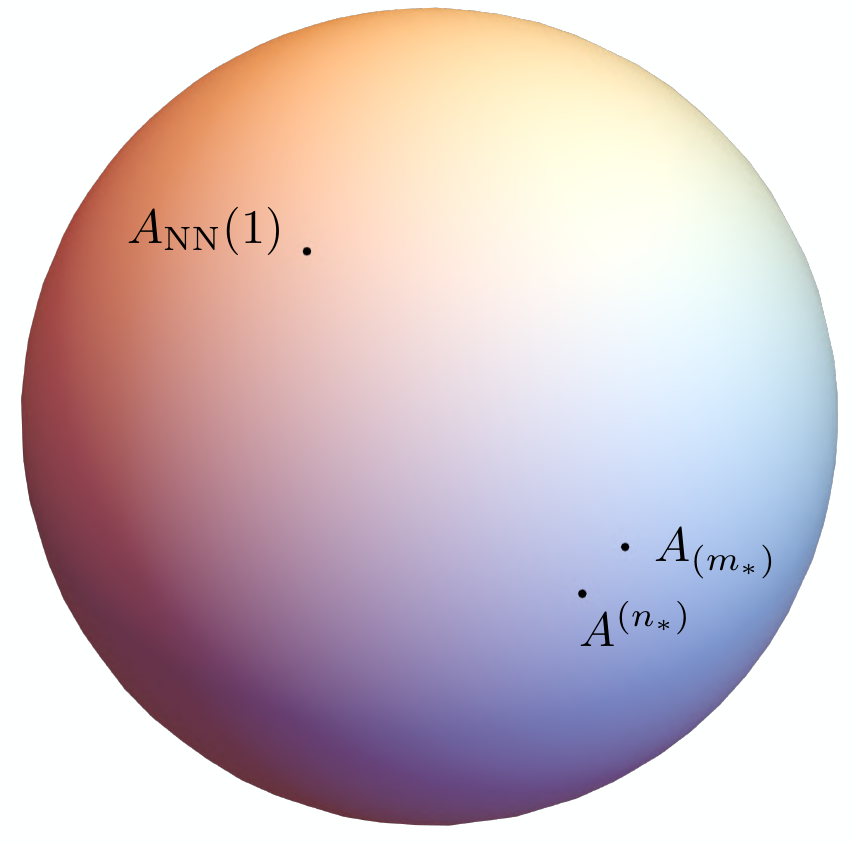}
	\caption{ 
		{The geometry on the unit Hilbert-Schmidt sphere corresponding to the distances \eqref{eq:d12},\eqref{eq:d13},\eqref{eq:d23}. The Hilbert-Schmidt sphere has large dimension, but since there are three tensors involved, the geometry can be shown on a two-dimensional unit sphere as in this figure.}
		\label{fig:sphere} 
	}
\end{figure}

Just like $A^{(n_*)}$, tensor $A_{(m_*)}$ does not respect reflections and rotations. The general comments in this regard made in \cref{sec:Isofp} still apply. We found that the distribution of tensor elements in their absolute magnitude has the same general shape for $A_{(m_*)}$ as for $A^{(n_*)}$, so we do not repeat \cref{fig:tensor_elements_magnitudes}.

Let us next discuss the eigenvalues of the Jacobian $\nabla R^\circ$ at $A_{(m_*)}$. The first few largest eigenvalues are reported in \cref{tab:isospec1}, last column. The main point here is that the quasiprimary eigenvalues are in good agreement with the CFT and with the observation that $T,\bar{T}$ eigenvalues should flip sign. 

Tensor $A_{(m_*)}$ is a much better approximation for the fixed point of rotating Gilt-TNR than $A^{(n_*)}$ was for the non-rotating Gilt-TNR ($10^{-9}$ vs $5\times 10^{-5}$, Eq.~\eqref{eq:5e5}). Does this lead to a better agreement of quasiprimary eigenvalues with CFT? Comparing the $\nabla R$ and $\nabla R^\circ$ columns of \cref{tab:isospec1}, we see that $\sigma$ eigenvalue agreement improves by two orders of magnitude, while $\epsilon$, $T$, $\overline{T}$ eigenvalues do not improve. Apparently, the error in these eigenvalues is dominated by the finite-$\chi$ truncation effects, and not by the error in the determination of the finite-$\chi$ fixed point tensors. It is not clear to us why the $\sigma$ eigenvalue is special in this regard.

As for the non-rotating case, we also see in \cref{tab:isospec1} a bunch of non-universal eigenvalues of $\nabla R^\circ$, associated with derivative interactions. Of some interest here is the leading irrelevant $\mathbb{Z}_2$-even eigenvalue $\lambda_{\rm irr}\approx 0.78$, as it would control the rate of approach to the fixed point if the traditional "shooting" method were used instead of the Newton method, see \cref{rem:TvsN} below.

Note that the subleading $\mathbb{Z}_2$-odd eigenvalues $0.0027\pm 1.5364i$ of $\nabla R^\circ$ in \cref{tab:isospec1} are close in absolute value to the pair of subleading $\mathbb{Z}_2$-odd eigenvalues 1.5328 and 1.5300 of $\nabla R$ reported in the same table, but are almost imaginary rather than real. These eigenvalues are likely associated with the CFT derivative operators $\partial \sigma$, $\overline\partial \sigma$, with the additional phase due to $\pi/2$ rotation of an operator with spin $\pm 1$. This correspondence between eigenvalues of $R$ and $R^\circ$ would be exact if the fixed points were the same. The fixed points are not exactly the same but close, and it seems that the Jacobian varies sufficiently slowly for the reasoning to be approximately true. Therefore, while the magnitude of derivative operator eigenvalues is not predictable, their relative phase between $R$ and $R^\circ$ may be at least approximately predictable.\footnote{We thank Naoki Kawashima for a discussion.}  Similarly, there are almost imaginary $\nabla R^\circ$ eigenvalues in the $\mathbb{Z}_2$-even sector which may be associated with $\partial \epsilon$, $\overline\partial \epsilon$ and with $\partial T$, $\overline\partial\, \overline T$ (not shown in \cref{tab:isospec1} since they belong to the lower part of the spectrum). We will see in \cite{paper-DSO} that the same pattern holds rigorously for the eigenvalues of the lattice dilatation operator.

\begin{remark} \label{rem:TvsN} It is interesting to discuss what it would take to approach the critical fixed point tensor with the traditional ``shooting'' method as closely as $10^{-9}$ achieved here with the Newton method. To use the shooting method we would need to determine the critical temperature sufficiently precisely and then use a sufficiently large number of iterations to approach the fixed point. The final accuracy $\delta_*$ is given in terms of the leading irrelevant $\mathbb{Z}_2$-even eigenvalue $\lambda_{\rm irr}$ in Eq.~\eqref{eq:general}. Using $\lambda_{\rm irr}\approx 0.63$ for the non-rotating Gilt-TNR map, we find that $\delta_*\sim 10^{-9}$ requires $\Delta t\sim 10^{-23}$, below the $\sim 10^{-16}$ machine precision for the double precision arithmetic. (For the rotating Gilt-TNR the required $\Delta t\sim 10^{-34}$ is much smaller due to $\lambda_{\rm irr}\approx 0.78$ being closer to 1.) Thus, to match the precision level of the Newton method with the shooting method would require multiple precision arithmetic, or a tensor RG map with a smaller $\lambda_{\rm irr}$. The theoretical lower limit here is $\lambda_{\rm irr}=0.25$, corresponding to quasiprimaries $T^2$ and $\overline{T}^2$ (see \cref{foot:noTTbar}).
\end{remark}

Finally, let us come back to the convergence rate in \cref{fig:newton_conv}. We checked numerically that, for all $s$ in that figure, the convergence rate is compatible with the largest eigenvalue of the Jacobian of the map $g_\approx$ which is being iterated (as it should be).
Here we would like to provide qualitative understanding of why the convergence rate improves with $s$ and why it eventually saturates, as mentioned above.
 
From the definition of $g_\approx$, its Jacobian at the fixed point is
\beq
(\nabla g_{\approx})(A_*^\circ)=I-J_\approx^{-1} (I-\nabla R^\circ(A_*^\circ))\,.
\eeq
We would like to understand the eigenvalues of this matrix. Recall that $J_\approx$ is evaluated not at the fixed point but at $A_{(0)}$ where we start the Newton iterations. Let us consider first an approximation where we neglect this difference:
\beq
J_\approx \approx I-P_s \nabla R^\circ(A_*^\circ) P_s\,.\label{eq:appr-lemma}
\eeq
Then we have to study the largest eigenvalue of the matrix
\beq
M_s=I-(I-P_s M P_s)^{-1} (I-M),
\eeq
where we denoted $M=\nabla R^\circ(A_*^\circ)$ and where $P_s$ is the orthogonal projector on the first $s$ eigenvectors of $M$. We have the following 
\begin{lemma}\label{lem:eigs} Assume $M$ is diagonalizable with the eigenvalues $\lambda_1$, $\lambda_2$,\ldots arranged in order of decreasing absolute value.
	Suppose none of the first $s$ eigenvalues are 1. Then $M_s$ is diagonalizable with eigenvalues $0, \cdots, 0, \lambda_{s+1}$, $\lambda_{s+2}$ \ldots, where eigenvalue $0$ has multiplicity $s$.
	\end{lemma}
The elementary proof is given below. Hence, the predicted convergence rate is $\lambda_{s+1}$. For $s=3$ this agrees well with \cref{fig:newton_conv}, but for larger $s$ the convergence in \cref{fig:newton_conv} is slower. We assign this discrepancy to the breakdown of the approximation \eqref{eq:appr-lemma}. We observed that re-evaluating $J_\approx$ after a few Newton iterations brings the convergence rate closer to the predicted value. In particular, with $J_\approx$ evaluated at $A_{(0)}$, convergence rate saturates around $s=54$, but after re-evaluation convergence keeps improving even for higher $s$.
	
\begin{proof}[Proof of Lemma \ref{lem:eigs}] The simple case is when the eigenvectors of $M$ are orthogonal. Then, $M_s$ expressed in the basis of eigenvectors is diagonal with 
	\beq
	0,\cdots,0,\lambda_{s+1},\lambda_{s+2}\ldots
	\label{eq:diagMs}
	\eeq
	on the diagonal, proving the lemma. In the general case, we work in the basis obtained by the Gram-Schmidt process from the eigenvectors of $M$. In this basis, all the involved matrices are upper-triangular with eigenvalues on the diagonal. In particular $M_s$ is upper-triangular with \eqref{eq:diagMs} on the diagonal.
	\end{proof}

\subsubsection{Can we increase the convergence domain of the Newton method?}
\label{sec:attempts}

The Newton method described above has a somewhat larger local domain of convergence than what is given by Eq.~\eqref{eq:dist_init}. One can start the Newton iterations from $A_{(0)}=A^{(n')}$ corresponding to a smaller number $n'$ of initial RG steps. One can also relax the accuracy $\Delta t$ of the initial critical temperature. Our tests show that $\Delta t=10^{-7}$ and $n'=15$ is already enough for the subsequent Newton iterations to converge. This corresponds to the distance $2.4\times 10^{-2}$ between the fixed point tensor and the initial tensor of the Newton iterations, an order of magnitude larger than \eqref{eq:dist_init}.

At the moment we are unable to further increase the Newton method convergence domain. The following simple changes of the algorithm do not help: re-evaluate $J_\approx$ after every Newton step; use larger $s$; apply the inverse Jacobian via the GMRES algorithm rather than via our approximate procedure; using ``damping'' i.e. decreased step size if the full Newton step does not lead to a decrease in the error $\|R^\circ(A)-A\|$.  Having examined this issue closely, we traced it to the lack of differentiability of the RG map in a larger domain, due to the discrete gauge fixing ambiguity. Our discrete gauge fixing procedure, described in \cref{app:GaugeFixing}, consists in choosing a set $I_+$ of tensor elements and making them all positive by discrete sign flips on the individual legs. This stops being differentiable when one of the tensor elements crosses zero. We observed that this does tend to happen when one considers the RG map in a sufficiently large domain. We are hopeful that if this problem is resolved, the Newton method convergence domain could be much increased, all the way to include the critical NN Ising tensor $A_{\rm NN}(1)$. Of course, for a faraway convergence one will have to recompute the inverse Jacobian at every RG step, at least until one gets sufficiently close to the fixed point.

\subsection{Anisotropy: period-2 oscillations}
\label{sec:oscillations}

\begin{figure}
	\centering
	\includegraphics[width=0.5\textwidth]{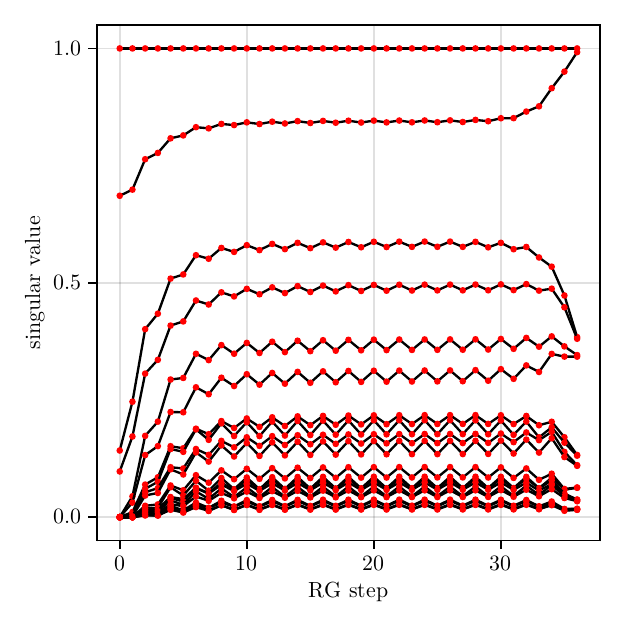}
	\caption{
	  \label{fig:sval_traj_r} Period-2 behavior for the rotating RG map. The figure shows the flow of tensors starting from $A^{(0)}$ corresponding to the Ising model at the reduced temperature  $t=1.0000094857$ and with anisotropy $a=0.8$. As in \cref{fig:sval_traj_nr} we plot the 20 largest singular values along the diagonal normalized by the first one.
}
\end{figure}

 Recall that for the non-rotating RG map the eigenvalues of the Jacobian corresponding to the stress tensor components $T$,$\overline T$ are $1$, and we expect a two-dimensional manifold of fixed points for this map. (A one-dimensional slice of this manifold was studied in \cref{sec:anisotropy}.) On the contrary, for the rotating RG map these eigenvalues are $-1$ and we expect an isolated fixed point embedded in a two-dimensional manifold of period-2 orbits (see \cref{sec:addrot}). We can see this period-2 behavior by studying the flow of the singular values of $A$ across the diagonal as we did in \cref{fig:sval_traj_nr} for the non-rotating RG map. We start with an initial $A$ representing the anisotropic Ising model with $a=0.8$. The flow of the resulting singular values shown in \cref{fig:sval_traj_r} clearly shows period-2 behavior. \cref{fig:sval_traj_r} only shows a single value of $a$, but we observed that the amplitude of the oscillations increases for larger $|a-1|$. These results confirm our intuition about the structure of the RG flow.

\section{Summary and open problems}
\label{sec:conclusions}

In this paper we showed that the fixed point equation of tensor RG can be solved via the Newton method. This needs a judicious modification of the RG map, to include a $\pi/2$ rotation of the lattice. The rationale is to avoid eigenvalue 1 deformations of the fixed point, which are problematic for the Newton method. The corresponding eigenvectors are associated with the stress tensor CFT deformations, which have spins $\ell=\pm2$. The $\theta=\pi/2$ rotation turns their eigenvalues from $1$ into $e^{i\theta\ell}= -1$, eliminating the problem.

Our numerical tests of this general idea used Gilt-TNR \cite{GILT}, one of several known tensor RG algorithms which implement the disentangling procedure necessary to solve the CDL problem. Working at bond dimension $\chi=30$, the Newton method was shown to converge to the fixed point in the Hilbert-Schmidt distance, with accuracy limited only by roundoff errors. To our knowledge, this is the first time that the Newton method was used to find the critical tensor RG fixed point. Working in double precision arithmetic, we achieved the accuracy $\delta \sim 10^{-9}$. This is much better than what was previously achieved with the traditional "shooting" method of recovering the fixed point, where one fine-tunes the relevant temperature perturbation, and removes the irrelevant perturbations by repeated RG steps   
($\delta\sim 10^{-2}$ \cite{Lyu:2021qlw};
$\delta\sim 3\times 10^{-4}$ \cite{PhysRevE.109.034111};
$\delta\sim 5\times 10^{-5}$ in our \cref{fig:crit_convergence}). We emphasize that this improvement is a result of using our Newton method rather than the usual shooting method, not an artifact of Gilt-TNR. We expect our method for finding the fixed point will produce similar improvement for other algorithms.

We would like to conclude by listing several open problems raised by our work:
\begin{enumerate}
	\item

	Can one solve the discrete gauge-fixing issue mentioned in \cref{sec:attempts} and obtain a globally differentiable tensor RG map, for which the Newton method convergence domain could be increased beyond what is achieved here, perhaps to include the critical NN Ising tensor $A_{\rm NN}(1)$?
	
\item It would be interesting to apply the Newton search to locate and study fixed points corresponding to other universality classes, beyond the Ising and 3-state Potts considered in our work, such as:

\begin{itemize}
	\item The tricritical Ising fixed point, initializing the algorithm by a tensor corresponding to the tricritical point of the Blume-Capel model, whose location is known with good accuracy from lattice Monte-Carlo studies \cite{Blume-Capel-MC}. The tricritical fixed point having two relevant symmetric deformations, its study would be more challenging using the shooting method, but should pose no problem to the Newton method. 
	
	\item The Lee-Yang fixed point corresponding to the Ising model in a critical imaginary magnetic field \cite{Fisher:1978pf}. This model is not amenable to Monte-Carlo studies due to the sign problem. See \cite{basu2022yanglee} for a previous tensor RG study using the shooting method.
	
	\item Fixed points which do not correspond to exactly solved CFTs, unlike all fixed points mentioned so far. One example is the coupled 3-state Potts model \cite{Dotsenko:1998gyp}.
\end{itemize}

\item In this paper, we relied on Gilt-TNR, with static $Q$-choice (\cref{app:rigid}) to ensure differentiability. Can the Newton method be set up for other tensor RG algorithms with disentangling, e.g.~for TNR \cite{Evenbly-Vidal}, combining them with $\pi/2$ rotation? Our logic about the rotation getting rid of eigenvalues $1$ should be robust. But is TNR  differentiable, or even continuous, in spite of the iterative optimization of disentanglers which may be trapped in a local minimum (\cref{note:TNRdanger})?

\item 
In this paper we argued on general grounds and provided numerical evidence that the RG map Jacobian at the fixed point has two kinds of eigenvalues. Eigenvalues associated with quasiprimary CFT operators are universal and predictable from CFT, while eigenvalues of derivative operators are not universal and their values cannot be predicted. Our \cref{tab:fp30} contains several eigenvalues which do not agree with any CFT quasiprimary, and we associated their eigenvectors with derivative operators. Is there a test via an explicit computation involving the eigenvectors and the fixed point tensor which can confirm this assignment?

\item  One somewhat technical issue concerns finding a fully robust gauge-fixing procedure for symmetry-breaking perturbations of the fixed point, see \cref{sec:comment_gauge_fix}.

\item Why is gauge fixing with orthogonal gauge transformations enough for a tensor RG map to have a fixed point? This has a satisfactory explanation for maps preserving reflection symmetry, \cref{app:why-orth}, but it is puzzling why this works for many maps used in the literature and in our work, which do not preserve it. In \cref{sec:without-reflection-symmetry} we tried to explain this conceptual puzzle via an argument appealing to a "phantom" RG map which has not been provided. Can one either exhibit this map, or find another argument which does not postulate its existence?

\item Are there practical tensor RG maps which manifestly preserve rotations while solving the CDL problem? One such map of unclear practicality is mentioned in \cref{note:isotropic} (it acts on a set of three tensors and has a larger than usual associated set of gauge symmetries). If that map, or any other rotation-invariant map, were available, one could completely circumvent the eigenvalues 1 problem, implementing the Newton method directly within the subspace of rotation-invariant tensors, as discussed in \cref{sec:rotation-symmetry-breaking}. 

\item
(See also \cref{rem:chi-fp} and the end of Section \ref{sec:Isofp}.) In this paper we studied the fixed point tensors for a fixed value of $\chi=30$ (and sometimes also for $\chi=10,20$). This was enough to defend our main points. However, as discussed in \cref{sec:parameter-choices-and-scaling}, higher bond dimensions up to $\chi=120$ used in \cite{GILT} may well be accessible, provided that the Jacobian is evaluated via the automatic differentiation instead of finite differences used here. One could then get evidence for or against the hypothesis that the fixed point tensors converge as $\chi\to\infty$ to a limiting tensor (which is then necessarily Hilbert-Schmidt). This hypothesis may or may not be true for the Gilt-TNR map, but hopefully there exists a tensor RG map for which it is true.

\item 
Finally, our work is also an important stepping stone towards proving the existence of the critical tensor RG fixed points at the mathematical level of rigor, the program started in \cite{paper1,paper2}. For this dream to be realized, the tensor RG fixed point equation has to be reformulated as a fixed point equation for a map which is a contraction in a neighborhood in an appropriate Banach space of infinite-dimensional tensors. In this paper we have found a numerical realization of a baby version of this program in a finite-dimensional space of truncated tensors.
\end{enumerate}

\begin{acknowledgments}

	We thank Glen Evenbly for communications and for his tensor network software and tutorials \cite{Evenbly-TNR-website,Evenbly-tensor-trace}. We especially thank Cl\'ement Delcamp for numerous useful communications and insightful discussions concerning tensor RG algorithms, and for passing to us minimal working versions of the Gilt-TNR algorithm from \cite{GILT,Gilt-TNR-code} and of Markus Hauru's TNR code from \cite{Hauru:2015abi} (the latter code was not used in this work but it was useful in preliminary stages). We thank Naoki Kawashima and Xinliang Lyu for comments.

This research was supported in part by the Simons Foundation grant 733758 (Simons Bootstrap Collaboration). This research was supported in part by grant no. NSF PHY-2309135 to the Kavli Institute for Theoretical Physics (KITP).

\end{acknowledgments}

\appendix

\section{Tensor conventions}
\label{app:conventions}

\subsection{Tensor RG}
\label{app:TRGconventions}
Let $A$ be a four-legged tensor of bond dimension $\chi$. Consider a periodic tensor network contraction made out of $A$ arranged in an $N\times M$ grid, denoted by
\begin{equation}
	Z(A,N\times M)\,,
	\label{eq:ZAN}
\end{equation}
and called the partition function of the tensor network. We assume that it represents the partition function of a lattice model of spins with periodic boundary conditions on the $N\times M$ lattice with short-range interactions, such as Ising or 3-state Potts.\footnote{Sometimes when the lattice model is translated to a tensor network representation, the spin lattice is rotated by $\pi/4$ with respect to the tensor network, as in \cref{app:init}. The lattice model then has somewhat unusual periodic boundary conditions. This microscopic detail does not play a role in our paper.}

We normalize $A$ to have $\|A\|=1$ where $\|A\|$ is the Hilbert-Schmidt norm.

Tensor RG maps considered in this paper fall into two categories: 1) non-rotating and 2) involving rotating the network by $\pi/2$ (``rotating''). Non-rotating tensor RG maps an $N\times M$ tensor network \eqref{eq:ZAN} made of $A$ to an $N/b\times M/b$ tensor network made of another tensor $A'=\mathcal{R}(A)$ so that the partition function remains invariant:
\begin{equation}
	Z(A,N\times M) = Z(A', N/b\times M/b)\,\qquad(\text{non-rotating}).
	\label{eq:ZAN1}
\end{equation}
On the other hand, for tensor RG maps with a $\pi/2$ rotation, invariance of the partition function is written as:
\begin{equation}
	Z(A,N\times M) = Z(A', M/b\times N/b)\,\qquad(\text{rotating}).
	\label{eq:ZAN1mod}
\end{equation}
There are of course many different tensor RG maps, differing in how $A'$ is obtained from $A$. Some concrete maps were discussed in the main text. Here we keep things general.

In general $A'$ is not normalized. Multiplying the tensor of which the tensor network is made by a constant factor $C$ is a relatively minor change, since such a factor can always be factored out from the partition function. It changes the partition function by $C^N$ where $N$ is the number of tensors in the network. However this does not change the phase of the microscopic model that the network represents. To discuss phases and fixed points, we have to normalize tensors somehow. We define the normalized RG map by
$
	R(A)=A'/\|A'\|
$,
so that the RG flow happens in the space of unit-norm tensors. We will most often work with $R(A)$ which is the most important part of the non-normalized RG map $\mathcal{R}$. We will sometimes abuse the terminology and refer to $R$ as the RG map, dropping ``normalized''. An RG fixed point for the map $R$ is a tensor $A_*$ satisfying
$
	R(A_*)=A_*
$.

\subsection{Tensor network symmetries}
\label{app:symmetries}

A symmetry transformation of a tensor network is a faithful linear action of a group $G$ on the space of tensors $A$:
\beq
T_g:A\mapsto A'\qquad(g\in G)\,,
\eeq
which preserves the partition function of a tensor network built of $A$:
\beq
Z(A,N\times M)=Z(T_g A,N\times M)\,.
\eeq
We say that the tensor network built out of a particular tensor $A_0$ respects this symmetry, if $A_0$ is invariant under this action: $T_g A_0=A_0$ for any $g\in G$.

We will consider on-site (global) symmetries and spatial symmetries. An on-site symmetry transformation multiplies tensor $A$ by matrices on each leg. E.g.~an on-site $\bZ_2$ acts as
\beq
\label{eq:figZ2}
\myinclude{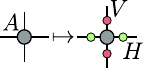}\,
\eeq
for some matrices $H,V$ such that $H^2=V^2=I$. 

A spatial symmetry permutes the legs of a tensor and simultaneously multiplies by matrices on each leg. Spatial symmetries mentioned in this paper, on the square lattice, involve a combination of axis reflections and rotations, forming the group $D_4$ or one of its subgroups $\bZ_2$, $\bZ_2\times\bZ_2$, $\bZ_4$. E.g. the horizontal reflection acts as:
\beq
\myinclude{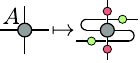}\,,
\label{eq:Z2flip}
\eeq
while rotation by $\pi/2$ acts as
\beq
\myinclude{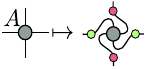}\,
\eeq
The matrices multiplying the legs must be chosen to get a representation of the spatial symmetry group. Note that they are not in general equal to the identity.

\section{Initial tensors for tensor RG evolution}
\label{app:init}

\begin{figure}
	\centering
	\includegraphics[scale=0.8]{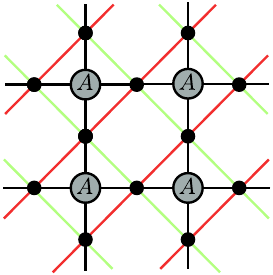}
        \caption{The Ising spins (black dots) are located at the midpoints of the bonds in the tensor network. The coupling constants for the Ising model are $J_y$ for the red bonds and $J_x$ for the green bonds. }
\label{fig:ising-tensor1}
\end{figure}

Here we describe the initial tensor for the RG evolution used in our studies, both in the isotropic and anisotropic case. There is more than one way to represent the nearest-neighbor (NN) Ising partition function as the contraction of a tensor network. The method we use in this paper is the same as used in \cite{TEFR}. In this method the lattice for the Ising model is rotated by $\pi/4$ with respect to the tensor network. In \cref{fig:ising-tensor1} the sites for the Ising spins are represented by black dots. The coupling constants on the bonds between these spins are equal to $J_y$ for the red bonds and $J_x$ for the green bonds. Note that the Ising spins are located at the midpoints of the bonds in the tensor network. Each leg in the tensor network only takes on the values $-1,+1$, corresponding to the values of the Ising spin. The contraction of the tensor network is then equal\footnote{Because of the lattice rotation by $\pi/4$, the equality requires non-standard periodic boundary conditions on the Ising side, assuming the usual periodic boundary conditions on the tensor network side.} to the NN Ising partition function if we define the components of the tensor $A$ by
\beq
\myinclude{fig-Eq2_2_and_EqB_1.pdf} \,  =
\exp((J_y \sigma_1 \sigma_2 + J_x \sigma_2 \sigma_3 + J_y \sigma_3 \sigma_4
+ J_x \sigma_4 \sigma_1)/T)\,.
\label{app:ising-tensor2}
\eeq
The final step in the definition of the initial $A$ is to apply a gauge transformation that transforms the states $\ket{+}$,$\ket{-}$ to the $\mathbb{Z}_2$ even and odd states given by
\beq
\ket{0} = \frac{1}{\sqrt{2}} (\ket{+} + \ket{-}), \quad 
\ket{1} = \frac{1}{\sqrt{2}} (\ket{+} - \ket{-})
\eeq

We define $\Jratio = J_y/J_x$.
For the isotropic case we set $J_x = J_y = 1$ and then define the reduced
temperature by $t = T/T_c$ so that $t = 1$ is the critical point. 
For the anisotropic case we set $J_y = \Jratio, J_x = 1$, and again let 
$t = T/T_c$ where now the critical temperature depends on $\Jratio$.\footnote{The critical temperature is given implicitly by
  $\sinh(2 J_y / T_c) \sinh(2 J_x / T_c) =1$.}
As in the isotropic case, the critical point is given by $t=1$. Only the parameters $t$ and $\Jratio$ are used in our discussions of the numerical results. In the numerical results the critical point is not exactly at $t=1$ since the bond dimension is finite.

Since the Ising lattice is rotated with respect to the tensor network, in the anisotropic case the initial $A$ is not invariant under horizontal or vertical reflections. However, it is invariant under the composition of these two reflections.

The initial tensor, which we normalize to have unit Hilbert-Schmidt norm, will be denoted $A_{\rm NN}(t,a)$, or $A_{\rm NN}(t)$ if $a=1$.

Another method for representing the Ising partition function as the contraction of a tensor network can be found in \cite{HOTRG,GILT}. In this approach the
Ising lattice is not rotated with respect to the tensor network.

\section{More details about Gilt-TNR}
\label{app:moreGilt}

We reviewed the general structure of the non-rotating Gilt-TNR algorithm in \cref{sec:GiltReview} and here we will provide some more details.
\begin{table}
	\centering
	\begin{tabular}{l l}
		\toprule
		$\chi$ & tensor bond dimension\\
		$\epsilon_{\rm gilt}$& soft truncation parameter in \eqref{eq:epsglit}\\
		$\epsilon_{\rm conv}$ & convergence parameter, see \eqref{eq:epsconv},\eqref{eq:epsconv1}\\
		\bottomrule
	\end{tabular}
	\caption{\label{tab:Giltparams}Parameters of Gilt-TNR. Results in this paper used $\epsilon_{\rm conv}=10^{-2}$.}
\end{table}

The basic procedure for choosing $Q=Q_1$ (see \cite[Sec.~IV.B]{GILT} and  \cite[App.~A]{Lyu:2021qlw}) consists of four steps:
\begin{itemize}
	\item Consider the singular value decomposition of an environment shown in \cref{fig:appD}(a), built out of four $A$ tensors contracted over 3 out of 4 bonds around the plaquette, omitting the bond into which $Q_1$ will be inserted. \item
	Form vector $t$ by contracting two ingoing legs of $U$, \cref{fig:appD}(b). 
	\item
	Perform soft truncation of vector $t$, defining vector $t'$ by
\beq
t_i'=t_i \frac {s^2_i}{s_i^2+\epsilon_{\rm gilt}^2},
\label{eq:epsglit}
\eeq
where $s_i$ are the singular values in \cref{fig:appD}(a), and $\epsilon_{\rm gilt}$ is one of the three parameters of the Gilt-TNR algorithm (\cref{tab:Giltparams}). Basically, this equation sets to near zero the components of $t'_i$ corresponding to $s_i\ll \epsilon_{\rm gilt}$. 
\item
Finally define matrix $Q$ from tensor $U^\dagger$ and the vector $t'$ as in \cref{fig:appD}(c).  
\end{itemize}

The rationale why this construction is a good way to filter short-range correlations from the plaquette was given in \cite{GILT}.

\begin{figure}
	\centering
	  \includegraphics{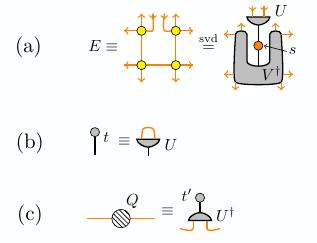}\\
	\caption{\label{fig:appD}Defining the matrix $Q\equiv Q_1$: (a) Environment and singular value decomposition;
		(b) Vector $t$; (c) Matrix $Q$. Figures from \cite{Lyu:2021qlw}.}
	\end{figure}

In practice however the above description is just the first step in an iterative procedure of determining $Q_1$, $Q_2$, $Q_3$, $Q_4$. We need two versions of the algorithm, which differ slightly but essentially in how this is implemented:
\begin{itemize} 
\item	
``Dynamic $Q$-choice''. That's the one implemented
in the original code \cite{Gilt-TNR-code}.
\item 
``Static $Q$-choice''. This is our modification, essential for ensuring differentiability.
\end{itemize}

\subsection{Dynamic $Q$-choice} 
\label{app:flexible}
This is only outlined in \cite[Sec.~IV.B]{GILT}. Here we will give a more detailed description, which corresponds to what is actually done in \cite{Gilt-TNR-code}. The dynamic algorithm depends on a parameter $\epsilon_{\rm conv}$, whose default value is $10^{-2}$; it appears in \cref{eq:epsconv} below. 

We start by determining $Q=Q_1$. To do that, one carries out steps 1-9 for $k=0,1,2,\ldots,k_{\max}$. Here $k_{\max}$ is determined dynamically during the run of the algorithm, and it depends on $\eps_{\rm conv}$ as explained below. Smaller $\eps_{\rm conv}$ will lead to larger $k_{\max}$.
\begin{enumerate}
    \item Start with the environment $M^{(k)}$ and the matrix $Q^{(k)}$. At the zero step $M^{(0)}=EE^{\dagger}$ where $E$ is in \cref{fig:appD}(a) and $Q^{(0)}=1$. By forming this environment out of $E$, one avoids having to compute matrix $V$ in \cref{fig:appD}(a), which is not needed in the algorithm.
    
    \item Do approximate SVD of $Q^{(k)}\approx (Q^{(k)})_L S_{Q^{(k)}} (Q^{(k)})_R$. Here, all the singular values smaller than $10^{-3}\epsilon_{\rm gilt}$ are thrown away. The factor $10^{-3}$ is hard-coded in \cite{Gilt-TNR-code}. 
    	
    \item Define $M^{(k+1)}$ by absorbing $(Q^{(k)})_L \sqrt{S_{Q^{(k)}}}$ and $\sqrt{S_{Q^{(k)}}} (Q^{(k)})_R$ into $M^{(k)}$.
    \item By the eigenvalue decomposition of $M^{(k+1)}$ determine $U^{(k+1)}$ and $S^{(k+1)}$.
    \item Form the vector $t$.
    \item Form the vector $t'$.
    \item Define $Q^{(k+1)}$.
    \item Do approximate SVD of $Q^{(k+1)} \approx Q^{(k+1)}_L S_{Q^{(k+1)}} Q^{(k+1)}_R$. Again, all the singular values smaller than $10^{-3}\epsilon_{\rm gilt}$ are thrown away. 
    \item Check if 
    \beq
    \max_i |(S_{Q^{(k+1)}})_i-1| < \epsilon_{\rm conv}.\label{eq:epsconv}
    \eeq
If this condition does not hold then go to step 1 with  $k \mapsto k+1$. Otherwise, we set $k_{\max}=k$ {(this is just to store the total number of carried out iterations)} and go to the next step. The rationale for condition \eqref{eq:epsconv} is that we would like \eqref{eq:Qproduct} below to converge.
    \item Define $Q$ via $Q^{(0)},\ldots,Q^{(k_{\max})}$ by combining all absorbed pieces in the appropriate order:
 \begin{equation}
 Q= \prod_{k=0}^{k_{\max}} \left(Q^{(k)}_L \sqrt{S_{Q^{(k)}}}\right) \prod_{k={k_{\max}}}^{0} \left(\sqrt{S_{Q^{(k)}}} Q^{(k)}_R\right)\,.
 \label{eq:Qproduct}
 \end{equation}
    \item Define $B_1,B_2$ tensors by splitting $Q$ (using approximate SVD, throwing away  eigenvalues $ < 10^{-3}\epsilon_{\rm gilt}$ again) and absorbing its halves into $A$.
\end{enumerate}

This ends the procedure for $Q_1$. One then proceeds to do the same for matrices $Q_3$, $Q_2$, $Q_4$, in this order.\footnote{\label{note:horvert1}That $Q_1$ and $Q_3$ are chosen before $Q_2$ and $Q_4$ is the likely reason for why the horizontal reflection symmetry is broken weaker than the vertical one (\cref{note:horvert}).} Of course, when we proceed to $Q_3$, we use tensors $B_1$ and $B_2$ to construct the relevant environments. Once $Q_3$ is found, we update $B_1$ and $B_2$, and proceed to $Q_2$, etc.
	
Once the $Q_a$ ($a=1,\ldots,4$) are constructed, we say that we ``completed a lap around the plaquette''. We then check the spectra of $Q_a=Q_{aL} S_{Q_a} Q_{aR}$ (again truncated SVD) and ask if 
	\beq
	\max_i |(S_{Q_a})_i-1| < \epsilon_{\rm conv}\,. \label{eq:epsconv1}
	\eeq If this condition does not hold for one or more bonds around the plaquette, we are supposed to run the above cycle again for those bonds, redefining yet again matrices $Q_a$ and tensors $B_1$, $B_2$. This will be the second, third, etc lap. In practice one usually needs two laps.
	
	Let us give a couple of examples of how this works in practice. Both examples are for $\chi=30$, $\epsilon_{\rm gilt}=6\times 10^{-6}$, $\eps_{\rm conv}=10^{-2}$. 
	\begin{enumerate}
		\item 
		Applied to $A^{(n_*)}$ in \cref{tab:fp30}, the dynamic algorithm does two laps, with the following number $k_{\max}$ of $k$ iterations for $Q_1,Q_3,Q_2,Q_4$:
	\begin{align}
		\text{Lap 1:}\  (19,29,10,28);\qquad
		 \text{Lap 2:}\  (1,1,1,1)\,.
		 \label{eq:rigid}
		\end{align}
			\item 
		Applied to $A_{(0)}$ in \cref{sec:rot-iso}, the dynamic algorithm does two laps, with the following number $k_{\max}$ of $k$ iterations for $Q_1,Q_3,Q_2,Q_4$:
			\begin{align}
			\text{Lap 1:}\  (22,8,24,20);\qquad
			\text{Lap 2:}\  (1,1,1,1)\,.
			\label{eq:rigid_for_newton}
		\end{align}
		\end{enumerate}

\subsection{Static $Q$-choice}
\label{app:rigid}
	In the static algorithm, we do not need the $\eps_{\rm conv}$ parameter. Instead, we specify the ``static instruction'', which includes 1) the number of laps and 2) $k_{\max}$ values for $Q_1,Q_3,Q_2,Q_4$ for each lap. The checks of conditions \eqref{eq:epsconv} and \eqref{eq:epsconv1} of the dynamic algorithm are omitted. The rest of the static algorithm is the same as the dynamic algorithm. This modification ensures that the static algorithm is continuous and differentiable.
	
 The static algorithm is used when we need the RG map to be differentiable, as when the Jacobian is involved. To select a static instruction, we typically run the dynamic algorithm once for a tensor around which we intend to compute the Jacobian. 
	
\subsection{Parameter choices and scaling}
\label{sec:parameter-choices-and-scaling}

In our paper the bond dimension $\chi$ was varied in the range 10-30, and the parameter $\epsilon_{\rm gilt}$ was varied from $10^{-4}$ for $\chi=10$ to $6\times 10^{-6}$ for $\chi=30$.

We used the dynamic $Q$-choice in calculations which did not need the Jacobian, such as the determinations of the critical temperature via bisection, and Figs.~\ref{fig:sval_traj_nr}, \ref{fig:crit_convergence}, \ref{fig:sval_traj_r}. In all dynamic $Q$-choice calculations the parameter $\epsilon_{\rm conv}$ was fixed to $10^{-2}$ (the default value of \cite{Gilt-TNR-code}).

We used the static $Q$-choice when computing the Jacobian eigenvalues, and when using the Newton method which needs the map to be differentiable. When computing the Jacobian eigenvalues at $A^{(n_*)}$ from \cref{tab:fp30}, we used the static instruction \eqref{eq:rigid} obtained by running once the dynamic algorithm at that tensor. In \cref{sec:rot-iso}, we used the dynamic $Q$-choice for the RG iterations from $A^{(0)}$ to $A^{(n')}$. Then, we switched to the static $Q$-choice with the static instruction \eqref{eq:rigid_for_newton} from the dynamic algorithm at $A_{(0)}=A^{(n')}$. We used this static instruction for all subsequent calculations in that section: 
\begin{itemize}
	\item 
	computing the eigenvalue decomposition of $\nabla R^\circ(A_{(0)})$ which enters $J_\approx$;
	\item 
	running the Newton method and producing \cref{fig:newton_conv};
	\item 
	computing the Jacobian eigenvalues $\nabla R^{\circ}$ at $A_{(m_*)}$.
	\end{itemize}

In this paper we limited ourselves to $\chi\le 30$ as this was enough to demonstrate our main points regarding the applicability of the Newton method. Ref.~\cite{Lyu:2021qlw} also used $\chi=30$. In the future, it would be interesting to carry out our analysis for larger $\chi$. (Ref.~\cite{GILT} went up to $\chi=120$.) In the rest of this section, we will discuss how various steps of our algorithms scale with $\chi$ and which other difficulties may arise when increasing $\chi$.

The computational cost of one Gilt-TNR algorithm step scales as $O(\chi^6)$. Evaluating the first $s$ eigenvalues and eigenvectors of $\nabla R(A)$ at any point $A$ using the Arnoldi method scales as $O(s \chi^6)$.\footnote{The Arnoldi method requires $O(s)$ evaluations of directional derivatives $\nabla R^\circ.v$, each of which takes $O(\chi^6)$ with our finite difference approximation method.} This operation is required only once when setting up our Newton method with approximate Jacobian, at $A=A_{(0)}$. For $\chi=30$ and $s=9$, it takes $14$ minutes on a laptop. One Newton method step costs $O(s^2 \chi^4+\chi^6)$, where $\chi^6$ comes from the Gilt-TNR step and $s^2\chi^4$ from applying $(\mathds{1}-P_s \nabla R(A_{(0)})P_s)^{-1}$ to $A-R(A)$. For example, for $\chi=30$ and $s=9$, a laptop can perform 30 Newton method steps from $A_{(0)}$, reaching the approximate fixed point, in 4 minutes (yellow curve in \cref{fig:newton_conv}). 

The most time-consuming part of our calculations was tuning the critical temperature, which remains necessary as the Newton algorithm may not converge to the critical point from a remote starting tensor. Finding the critical temperature using the bisection method with $\Delta t = 10^{-10}$ takes around $80$ minutes on a laptop. 

A potential issue for going to larger $\chi$ has to do with the error in evaluating the Jacobian via finite differences. As discussed in \cref{app:GiltDiff}, this error cannot be made arbitrary small due to roundoff errors. For $\chi=30$, the error for an optimal step size $h$ was quite small (see \cref{fig:diff_test_results}), at least for derivatives in large eigenvalue directions.  However, we observed that the optimal error attainable for the finite difference method increases with $\chi$. Extrapolating from $\chi=10,20,30$, we suspect that for e.g.~$\chi=60$ it may reach values which will make the finite difference method unusable. Automatic differentiation techniques \cite{autodiff} famously avoid roundoff errors inherent for the finite difference method. We expect that a rewrite of our code using automatic differentiation would allow to go to higher $\chi$.

\subsection{Gilt-HOTRG \cite{Lyu:2021qlw}}
	\label{app:Tokyo}
	In this section we will describe the algorithm used in Ref.~\cite{Lyu:2021qlw,TokyoCode}, which somewhat differs from Gilt-TNR used in \cite{GILT} and our work. We will refer to it as Gilt-HOTRG. Having this description is useful to understand and appreciate some details of the results of Ref.~\cite{Lyu:2021qlw}, such as for example the fact that their linearized RG map coincides with the lattice dilatation operator and hence its eigenvalues agree with (exponentiated) CFT scaling dimensions even for descendants \cite{paper-DSO}. 
	 
	As the name suggests, Gilt-HOTRG \cite{Lyu:2021qlw} combines Gilt with HOTRG \cite{HOTRG}, i.e.~it uses Gilt to filter local correlations, and HOTRG for coarse-graining, while Gilt-TNR uses TRG for coarse-graining. The main equation which defines their RG transformed tensor in terms of the original tensor is (figure from \cite{Lyu:2021qlw}, Eq.~(31)):
    \beq
    \myinclude{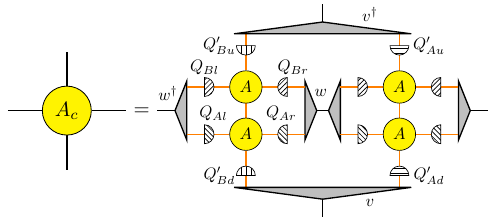}\,.  \label{eq:Tokyo1}
    \eeq
   Let us explain this equation further. As a first step, one chooses Gilt matrices $Q_B$ and $Q_A$, applying the Gilt procedure to horizontal links of a subset of plaquettes (figure from \cite{Lyu:2021qlw}, Fig.~11)):
    \beq
   \myinclude[width=0.9\textwidth]{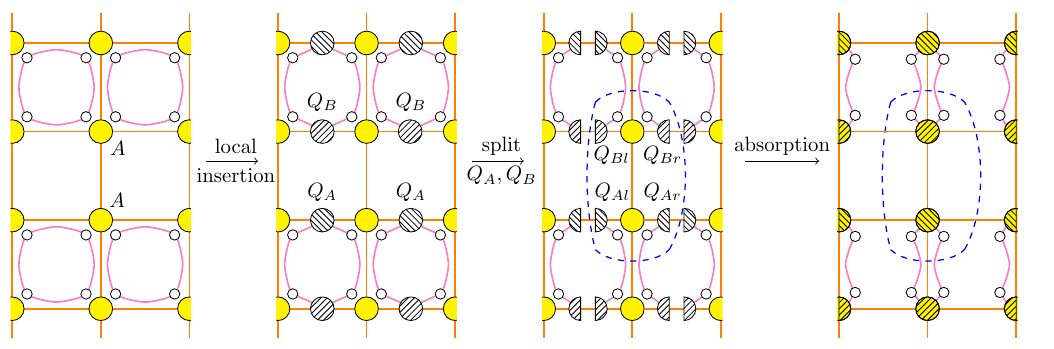}\,. \label{eq:Tokyo2}
   \eeq

   As one can see in \eqref{eq:Tokyo2}, matrices $Q_B$ and $Q_A$ are then SVD-decomposed, and their halves are absorbed into $A$ (this step breaks reflection symmetry since in general $Q_{Br}\ne (Q_{Bl})^t$, $Q_{Ar}\ne (Q_{Al})^t$). One then performs HOTRG coarse-graining in the vertical direction, using isometries $w,w^\dagger$ in \eqref{eq:Tokyo1} (see also \cite{Lyu:2021qlw}, Eq.~(A6)).
	 
	 As a second step, the whole procedure is repeated in the horizontal direction. The final result is an RG map which reduces the lattice size by factor 2 in both directions.
	 
	 Ref.~\cite{Lyu:2021qlw} do not mention if they use iterative optimization of Gilt matrices.
	 
    Gilt-HOTRG is a non-rotating RG map, so it has eigenvalues 1 around the fixed point. Like Gilt-TNR, it can be accompanied by a rotation, turning eigenvalues 1 to $-1$ and making the Newton method applicable.
	 
\section{Gauge fixing}
\label{app:GaugeFixing}

As mentioned in \cref{sec:Isofp}, to see that the tensors $A^{(n)}$ converge to a fixed point, we need to apply after each Gilt-TNR step an appropriately chosen gauge transformation. Tensors related by a gauge transformation are physically equivalent as they yield the same partition function. We need to choose a unique tensor from each gauge equivalence class to discuss the fixed point tensor. This can be done by employing a gauge fixing procedure.

A general gauge transformation preserving the partition function would consist in multiplying the tensor on horizontal legs with matrices $G^{-1}_h$ and $G_h$ where $G_h$ is an invertible $\chi\times\chi$ matrix, and similarly on the vertical legs with $G^{-1}_v$ and $G_v$. In this work, we used a restricted class of gauge transformation where $G_v, G_h$ are orthogonal matrices.

Why are orthogonal gauge transformations enough? This may be expected by looking at the trajectories of singular values along the diagonal, which show plateaux (\cref{fig:sval_traj_nr}). Indeed, singular values are invariant under orthogonal gauge transformations, see \cref{fig:singval}, while they would not be invariant under more general gauge transformations with invertible matrices. For further rationale see \cref{app:why-orth}.

For orthogonal gauge transformations we have $G^{-1}_h=G_h^T$, $G^{-1}_v=G_v^T$. Graphically such a gauge transformation acts on a four-legged tensor $A$ of bond dimension $\chi$ as follows:
\begin{equation}\label{eq:GaugeFixing1}
	\myinclude[scale=1]{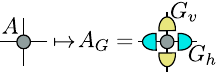},
\end{equation}
where $G_v, G_h \in O(\chi)$.

"Gauge fixing" means that we will apply the gauge transformation \eqref{eq:GaugeFixing1} with some appropriately chosen $G_v, G_h$. Our gauge fixing procedure involves continuous gauge fixing followed by discrete gauge fixing. Continuous gauge fixing is done as follows. Consider the following two contractions of two copies of $A$, called environments:
\begin{equation}\label{eq:GaugeFixing2}
	\myinclude[scale=1]{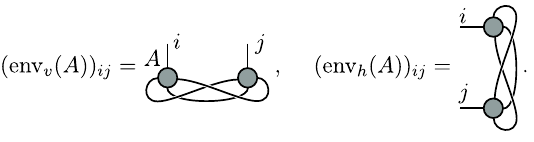}
\end{equation}
The contractions in \eqref{eq:GaugeFixing2} are chosen in such a way that both ${\rm env}_v(A)$ and ${\rm env}_h(A)$ are symmetric and can be diagonalized by orthogonal matrices. We choose $G_v$ and $G_h$ that diagonalize these environments. Then we perform the gauge transformation \eqref{eq:GaugeFixing1}. The tensor $A_G$ is the result of continuous gauge fixing.

There is some ambiguity in the choice of $G_v$ and $G_h$. If the spectra of $\mathrm{env}_v(A)$ and $\mathrm{env}_h(A)$ are nondegenerate, as is generically true, then this ambiguity reduces to multiplying $G_v\to G_v S_v$ and $G_h\to G_h S_h$ where $S_v$ and $S_h$ are diagonal matrices with $\pm 1$ on the diagonal.
In terms of $A_G$, this amounts to applying an additional gauge transformation:
\begin{equation}\label{eq:GaugeFixing3}
	\myinclude[scale=1]{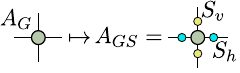},
\end{equation}
Discrete gauge fixing amounts to choosing the matrices $S_v$ and $S_h$ appropriately. The tensor $A_{GS}$ is the final result of our gauge fixing procedure.

Note that there are some replacements of $S_v,S_h$ which leave $A_{GS}$ invariant:
\begin{itemize}
	\item
	      $S_v \mapsto -S_v$;
	\item
	      $S_h \mapsto -S_h$;
	\item
	      the simultaneous replacement $S_v, S_h \mapsto V S_v, H S_h$ provided that, as is the case for us, the tensor $A_G$ enjoys a $\mathbb{Z}_2$ symmetry
	      \begin{equation}\label{eq:GaugeFixing35}
		      \myinclude[scale=1]{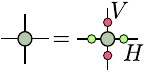}
	      \end{equation}
	      with $H,V$ diagonal matrices with $\pm1$ on the diagonal (see \cref{eq:figZ2}).
\end{itemize}

We will perform discrete gauge fixing by imposing the following gauge fixing condition on the tensor $A_{GS}$:
\begin{equation}\label{eq:GaugeFixing5}
	(A_{GS})_{i j k l} > 0, \qquad \forall (i,j,k,l) \in I_+.
\end{equation}
for an appropriately chosen set of tensor components $I_+ \subset \{1,\ldots,\chi\}^4 $. In words, this says that after the discrete gauge transformation, a certain subset of (nonzero) tensor components should be positive. Equivalently, \cref{eq:GaugeFixing5} can be written as
\begin{equation}\label{eq:GaugeFixing4}
	(S_v)_i(S_h)_j(S_v)_k(S_h)_l=\mathrm{sign}((A_G)_{i j k l}), \qquad \forall (i,j,k,l) \in I_+,
\end{equation}
where we introduced a shorthand notation $(S_v)_i$ for $(S_v)_{ii}$.

The set $I_+$ should be not too large (because otherwise \eqref{eq:GaugeFixing4} will not have a solution) and not too small (because otherwise \eqref{eq:GaugeFixing4} will have multiple solutions and the gauge fixing will not be complete). In other words, $I_+$ should be such that $S_v, S_h$ satisfying \eqref{eq:GaugeFixing4} are unique (up to the above-mentioned replacements which do not change $A_{GS}$).

How do we choose a suitable $I_+$ and solve \eqref{eq:GaugeFixing4}? Let $I_{\neq 0}$ be the set of all components such that $(A_G)_{ijkl} \neq 0$ for all $(i,j,k,l) \in I_{\neq 0}$. Consider \emph{all} equations of the form
\begin{equation}\label{eq:GaugeFixing6}
	(S_v)_i(S_h)_j(S_v)_k(S_h)_l=\mathrm{sign}((A_G)_{i j k l}),
\end{equation}
where $(i,j,k,l) \in I_{\neq0}$. Let us parametrize
\begin{equation}
	(S_v)_i=(-1)^{v_i}, \qquad (S_h)_i=(-1)^{h_i}, \qquad \mathrm{sign}((A_G)_{i j k l}) = (-1)^{a_{i j k l}},
\end{equation}
where $h_i, v_i, a_{ijkl} \in \{0,1\}$.
Then we rewrite equations \eqref{eq:GaugeFixing6} as linear equations in modular arithmetic:
\begin{equation}\label{eq:GaugeFixing7}
	v_i+h_j+v_k+h_l=a_{i j k l} \mod 2.
\end{equation}
Note that the number of these equations is much higher than the number of unknowns, thus we cannot hope to solve all of them (in other words, we cannot hope to find a discrete gauge transformation that makes \emph{all} nonzero elements of $A_{GS}$ positive). Instead, we will choose a specific subset $I_+\subset I_{\ne 0}$ and solve the corresponding subset of equations. To choose $I_+$, we first order equations \eqref{eq:GaugeFixing7} in the order of decreasing $|(A_G)_{ijkl}|$. Then, we form the matrix $M$ so that (ordered) \eqref{eq:GaugeFixing7} can be written as
\begin{equation}\label{eq:GaugeFixing8}
	M \vec x = \vec a,
\end{equation}
where $\vec x$ is the vector of $2\chi$ unknowns which combines $h_i$'s and $v_i$'s. Next, we perform Gaussian elimination on $M$ to identify a maximal independent set of rows, starting with row 1 and then trying subsequent rows one by one and adding them to the set if they are linearly independent from the previously added rows, until no more rows can be added. The resulting set of linearly independent rows is what gives us the set $I_+$. We find that the size of $I_+$ is, as expected, $2\chi-3$, corresponding to $2\chi$ unknowns minus $3$ transformations which were discussed above and which determine the kernel of $M$.

Note that the choice of $I_+$ described so far depends only on the ordering of equations \eqref{eq:GaugeFixing7} in the order of decreasing $|(A_G)_{ijkl}|$ that we chose. Now we take system \cref{eq:GaugeFixing8}, restrict it to rows $I_+$, and solve for $\vec x$ in terms of $\vec a$. The solution is not unique but as discussed different solutions are equivalent since they give the same $A_{GS}$. In practice, for this last step we use a {\tt Julia} \cite{Julia-2017, Julia} package {\tt AbstractAlgebra.jl} \cite{AbstractAlgebra.jl-2017, AbstractAlgerba}, which includes a linear solver for finite fields.\footnote{The linear solver returns one solution and, optionally, the kernel of the matrix, corresponding in our case to vectors that realize symmetries discussed before \cref{eq:GaugeFixing35}.} We use this solution to fix the discrete gauge.

From the described procedure, it can be seen that $I_+$ as defined above will remain constant in a small neighborhood $U_1$ of $A$ in which the ordering of the elements $|(A_G)_{ijkl}|$ which enter into its definition does not change. This is a good feature. Indeed, if $I_+$ changed as $A$ varied, this could lead to discontinuities in the RG map. We cannot allow such discontinuities: we want our RG map to be continuous and even differentiable. 
In our final algorithm, we enforce $A$-independence of $I_+$ in an even larger neighborhood $U_2\supset U_1$ by first computing $I_+$ for a tensor $A_\approx $ approximating the critical fixed point tensor, and then using exactly the same $I_+$ everywhere in $U_2$. The size of $U_2$ is determined by the condition that no tensor element $|(A_G)_{ijkl}|$ with $(i,j,k,l) \in I_+$ vanishes in $U_2$.

\subsection{Comment about gauge fixing for $\mathbb{Z}_2$-odd perturbations}
\label{sec:comment_gauge_fix}

The previous discussion solves the gauge-fixing problem with the $\mathbb{Z}_2$-even subspace to which the fixed point tensor $A_*$ belongs, and within which the RG evolution converging to $A_*$ takes place. 
For a tensor $A$ not belonging to  the $\mathbb{Z}_2$-even subspace, $R(A)$ as described up to now would be computed up to a $\mathbb{Z}_2$ symmetry transformation. It is important to fix this ambiguity consistently, to get a differentiable map in a neighborhood of a fixed point.\footnote{In practice, this need arises when evaluating the Jacobian eigenvalues around the fixed point in the $\mathbb{Z}_2$-odd sector. This requires computing $R(A_*+\delta A)$ for small $\mathbb{Z}_2$-odd perturbations $\delta A$.} In our algorithm, this is fixed as follows. The potential source of ambiguous signs is the singular value decomposition (SVD) operation $M=U \Lambda V $ since multiplying row $i$ of $U$ by a sign $s_i$ while simultaneously multiplying column $i$ of $V$ by the same sign gives an equally valid SVD. It is hard to control how these signs are chosen in a built-in SVD function. Instead, our algorithm uses a custom-made SVD function which, after calling the built-in SVD, applies the above transformation with
\beq
s_i = \text{sign}\Bigl(\sum_j w_j U_{ij}\Bigr), \label{eq:row-sum}
\eeq
where $w_j$ are some chosen fixed weights, e.g.~$w_j=1$.

This simple trick, combined with the gauge-fixing procedures described above, is sufficient to get a differentiable map in a neighborhood of $A_*$, including the $\mathbb{Z}_2$-odd directions with some rare exceptions.  When one considers a family of problems, like in \cref{sec:anisotropy}, where we varied the anisotropy parameter $a$, this method fails for a few parameter values. The method can also fail for a few values of $\Jratio$ when one varies the step size, $s$, in the numerical differentiation. We conjecture this happens because at some step of the algorithm a row of $U$ occurs whose elements nearly sum to zero, so that \eqref{eq:row-sum} becomes numerically ill defined. In some but not all cases this problem can be solved by using a different weight vector in \eqref{eq:row-sum}, such as e.g.~$w_j=1/(1+j)$. Further work is needed to fully understand this problem.

\subsection{Relation with the approach of Ref.~\cite{Lyu:2021qlw}}

Let us now discuss how our gauge fixing procedure relates to the procedure used in Ref.~\cite{Lyu:2021qlw}. 
Their RG map, Gilt-HOTRG, was reviewed in \cref{app:Tokyo} (without gauge fixing step) . It can be seen from the definition of Gilt-HOTRG that if the tensor $A$ is acted upon by an orthogonal gauge transformation, this has no effect on $A'$. This is because various unitary matrices which enter into SVD decompositions and the HOTRG step transform in such a way that the final result remains unchanged (see \cite[Appendix B]{Lyu:2021qlw} for a discussion). Based on this observation, Ref.~\cite{Lyu:2021qlw} observe that there is no need to perform continuous gauge fixing in their algorithm.\footnote{As a matter of fact, this observation also applies to Gilt-TNR used to our work. In fact we have checked that the fixed point is reached even if the continuous gauge fixing step is omitted in our algorithm. Nevertheless, we decided to keep the continuous gauge-fixing step, since it is not expensive. Also, having this step prompted us to reflect more deeply why the orthogonal gauge fixing is actually sufficient in our algorithm, \cref{app:why-orth}.}

Of course, discrete gauge fixing is essential for Ref.~\cite{Lyu:2021qlw} as it is for us. Regarding discrete gauge fixing, there are two main differences between our procedure and the procedure in Ref.~\cite{Lyu:2021qlw}. The first one is that instead of the search of $I_+$, Ref.~\cite{Lyu:2021qlw} makes a particular choice of this subset. While they do not provide exhaustive details, it appears to us that they pick $I_+$ to contain elements of the form $(i^+,1^+,1^+,1^+)$, $(1^+,i^+,1^+,1^+)$, $(i^-,1^-,1^+,1^+)$, $(1^-,i^-,1^+,1^+)$, where the superscript refers to the $\mathbb{Z}_2$ quantum number. The second difference is that instead of absolute gauge fixing condition \eqref{eq:GaugeFixing5}, Ref.~\cite{Lyu:2021qlw} imposes a relative condition:
\begin{equation}\label{eq:GaugeFixing12}
 \mathrm{sign}  ((A'_S)_{i j k l}) = \mathrm{sign} (A_{i j k l}), \qquad \forall (i,j,k,l) \in I_+. 
\end{equation}
The choice of $I_+$ as above allows to solve \cref{eq:GaugeFixing12} easily (so no need to solve a system of linear equations in modular arithmetic, as in our algorithm). Specifically, we can fix $(S_h)_{1^+}=(S_v)_{1^+}=(S_v)_{1^-}=1$ using the symmetries discussed after \cref{eq:GaugeFixing3}. Then, the solution of \cref{eq:GaugeFixing12} is 
\begin{align}
 (S_h)_{i^+}&= \mathrm{sign}  (A'_{i^+1^+1^+1^+})\mathrm{sign}  (A_{i^+1^+1^+1^+}),\nonumber\\
 (S_v)_{i^+}&= \mathrm{sign}  (A'_{1^+i^+1^+1^+})\mathrm{sign}  (A_{1^+i^+1^+1^+}),\nonumber\\
 (S_h)_{i^-}&= \mathrm{sign}  (A'_{i^-1^-1^+1^+})\mathrm{sign}  (A_{i^-1^-1^+1^+}),\\
 (S_v)_{i^-}&= \mathrm{sign}  (A'_{1^-i^-1^+1^+})\mathrm{sign}  (A_{1^-i^-1^+1^+})(S_h)_{1^-}.\nonumber
\end{align}

In our opinion, our discrete gauge fixing method is better protected from vanishing tensor elements in $I_+$ than the procedure in Ref.~\cite{Lyu:2021qlw}. In practice, however, both procedures work well and allow one to demonstrate the convergence of the RG flow to the critical tensor (see \cref{fig:crit_convergence} and Ref.~\cite[Fig. 9]{Lyu:2021qlw}).

\subsection{Why is orthogonal gauge fixing enough?}
\label{app:why-orth}
We would like to reflect here, non-rigorously, on why gauge fixing via orthogonal gauge transformations was enough to get, in \cref{sec:no-rot,sec:with-rot}, RG maps having fixed points. 

An RG map is designed to preserve the partition function. When an RG evolution starts from an initial tensor corresponding to the critical temperature, it is reasonable to expect that, with a bit of care and luck,\footnote{By care we mean e.g.~the need to solve the CDL problem. By luck we mean e.g.~that it is hard to a priori exclude such situations as cycles of finite length. For example, discrete dynamical systems are known to generically undergo period-doubling bifurcations.} a fixed-point tensor will be reached, as long an appropriate gauge fixing is incorporated. However, we may expect that a gauge fixing with general invertible matrices may be needed, as this is the most general transformation of tensors preserving the partition function. We know of only one work \cite{PhysRevE.109.034111} which used such a general gauge fixing, by going to the Minimal Canonical Form \cite{Acuaviva:2022lnc}. The works which studied the singular values along the diagonal without doing explicit gauge fixing (e.g.~\cite{GILT,2023arXiv230617479H}) also provide evidence that orthogonal gauge fixing is enough, since general gauge fixing does not leave these singular values invariant. So while general gauge fixing may be useful, it does not appear indispensable.

Let us assume that we are given an RG map $R$ which has a fixed point tensor $A_*$, up to gauge fixing with general invertible matrices: 
\beq
R(A_*)=\myinclude[scale=0.6]{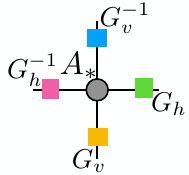}\,.
\label{eq:why-orth1}
\eeq
(So that if we redefine $R$ by including gauge fixing into it, $A_*$ would become a fixed point.) 

Our strategy will be as follows. First, we will discuss tensors and RG maps satisfying some additional assumptions, namely the presence of reflection symmetry. Under those assumptions, we will show that if \eqref{eq:why-orth1} holds, then a similar equation holds with orthogonal gauge fixing matrices, i.e.~orthogonal gauge fixing is enough. Then we will consider the case of weakly broken reflection symmetry, as is the case for the maps in \cref{sec:no-rot,sec:with-rot}.
 
 \subsubsection{With reflection symmetry}
Suppose that the following conditions hold in addition to \eqref{eq:why-orth1}:
\begin{itemize}
	\item
	Tensor $A_*$ respects horizontal reflection symmetry in the form:
	\beq
	\myinclude[scale=0.6]{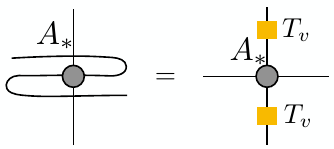}\,,
	\label{eq:why-orth2}
	\eeq
	where $T_v$ is a diagonal matrix with $\pm1$ on the diagonal. Note that this is not the most general form of reflection symmetry compared to \eqref{eq:Z2flip} as we don't include matrices on the horizontal legs. 
	\item
	Tensor $A_*$ also respects vertical reflection symmetry in the form of equation analogous to \eqref{eq:why-orth2} where crossing of vertical legs is equivalent to multiplying horizontal legs by $T_h$.
	\item Tensor $R(A_*)$ also respects both horizontal and vertical reflection symmetry, in the form of equations like \eqref{eq:why-orth2} (with $A_*$ replaced by $R(A_*)$) and of its vertical analogue, and with exactly the same $T_v$, $T_h$. 
	\item $G_v$ and $G_h$ in \eqref{eq:why-orth1} commute with $T_v$ and $T_h$, respectively. This is very natural. The eigenvalues of $T_v$ and $T_h$ have the meaning of the reflection quantum numbers of the basis elements of the vertical and horizontal bond Hilbert spaces. We expect that gauge transformations may mix only vectors having the same quantum numbers (both reflections and $\mathbb{Z}_2$). For reflections, this implies $[G_v,T_v]=[G_h,T_h]=0$.
	\end{itemize}
	
	Under the above assumptions, we can prove that an analogue of \eqref{eq:why-orth1} holds with orthogonal gauge fixing matrices.
\begin{proof} Consider the polar decomposition of $G_h, G_v$ of the form
	\beq
	G_h = P_h R_h,\qquad G_v=P_v R_v, 
	\eeq
	where $R_h,R_v$ are orthogonal and $P_h,P_v$ are symmetric positive matrices. Note that since in this paper we work with real tensors, we assume that $G_h,G_v$ are real.
	
	Now consider the following string of equations:
	\beq
	\myinclude[scale=0.6]{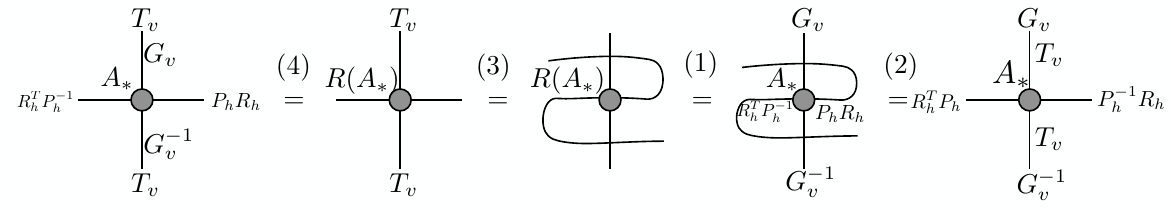}\,.
	\label{eq:why-ort3}
	\eeq
We read it from the center towards the sides. In (1) we used \eqref{eq:why-orth1} replacing $G_h=P_hR_h$. In (2) we used reflection symmetry of $A_*$. In (3) we used reflection symmetry of $R(A_*)$. In (4) we used \eqref{eq:why-orth1}. Now compare the rightmost and leftmost term in \eqref{eq:why-ort3}. Using that $[G_v,T_v]=0$ on the vertical legs, and multiplying by $P_h^{-1} R_h^{-T}$ and $R^{-1}_hP_h$ on the horizontal legs, we obtain:
	\beq
	\myinclude[scale=0.6]{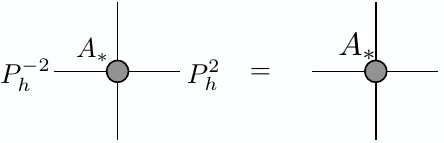}\,.
	\label{eq:why-ort4}
	\eeq
	This equation means that $A_*$ commutes with $P_h^2$, when multiplied on horizontal legs, for any values of vertical indices. By standard arguments it then also commutes with $f(P_h^2)$ for any $f$. In particular choosing $f(x)=\sqrt{x}$ we obtain that  $A_*$ commutes with $P_h$. 
	
	Repeating the above argument in the vertical direction we obtain that $A_*$ also commutes with $P_v$ when multiplied on vertical legs. Combining the two properties, we obtain that the r.h.s.~of \eqref{eq:why-orth1} in fact equals
	\beq
	\myinclude[scale=0.6]{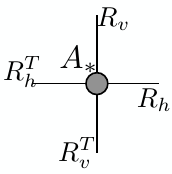}\,.
	\label{eq:why-ort5}
	\eeq
	Thus we obtain that $R(A_*)$ equals $A_*$ up to an orthogonal gauge transformation.
	\end{proof}
	
\begin{remark} That orthogonal gauge fixing should be enough to get a fixed point for RG maps preserving reflection symmetry was stated without proof in e.g.~\cite{Evenbly-review,Lyu:2023ukj}. We don't know a place where this statement was proven before our work.
	\end{remark}
	
	\subsubsection{Without reflection symmetry}\label{sec:without-reflection-symmetry}
	
	The above result would be good for us if our RG maps preserved reflections. Unfortunately they do not. They do not even preserve reflections in a weaker form \eqref{eq:Z2flip} with matrices on both horizontal and vertical legs. Breaking of reflections occurs for two reasons. First, the matrices $Q_1$, $Q_2$, $Q_3$, $Q_4$ are chosen consecutively in \eqref{eq:gilt1}. Second, when these matrices are SVD-decomposed in the r.h.s.~of \eqref{eq:gilt1}, the left and the right parts of the decompositions are not transposes of each other, unless the $Q_i$'s are positive definite, which is not generally the case. 
	
We note in passing that the RG map in \cite{GILT} (which is the same as our RG map without rotation) and the RG maps in \cite{Lyu:2021qlw} also break reflections.\footnote{We believe the statement to the contrary below Eq.~(B5b) in \cite{Lyu:2021qlw} to be incorrect.} Why then is orthogonal gauge fixing enough for all these maps, including ours, to exhibit a fixed point?

We would like to advance here one possible reason, which has to do with our maps breaking reflections only weakly. Consider first the rotating RG map from \cref{sec:with-rot}. It has a fixed point attainable by the Newton method. If the map preserved reflections, the existence of a fixed point with orthogonal gauge fixing would not be surprising given the result from the previous subsection. We know that the reflections are broken, but rather weakly, as can be seen from the fact that the fixed point tensor elements in \cref{tab:fp30}, column 3, which are related by reflections do differ, but not too much. Given this, we may speculate that our map is perhaps close to some "phantom" RG map which does preserve reflections (and hence has a fixed point with orthogonal gauge fixing). The existence of a fixed point for our map would then be natural, given that the Jacobians of our map and of the ``phantom'' RG map are close, and that it is the Jacobian which enters into the Newton method. The non-rigorous part of this argument is that we haven't proved the existence of the ``phantom'' RG map.

As to the RG map without rotation, we can try to think along the same lines (with \cref{tab:fp30}, column 2 instead of column 3), up to one further subtlety. Namely, as discussed many times in this paper, the map without rotation has eigenvalues 1 due to the existence of aspect-ratio changing deformations of the fixed point. Consider the set of critical tensors corresponding to aspect ratio 1 (rotation symmetry). These are the tensors for which the anisotropy ellipse is a circle. It is a nonlinear submanifold of the space of all tensors, over which we have little control apart from its existence. The RG map preserves the criticality and the aspect ratio, hence acts within this submanifold (see \cref{sec:rotation-symmetry-breaking}). Within the submanifold there are no eigenvalues 1, and the above argument may apply. 

We stress again that the just given arguments are highly non-rigorous. Unfortunately we don't have better arguments for the moment.

\section{Gilt-TNR differentiability}
\label{app:GiltDiff}

Let $R$ be the Gilt-TNR map described in \cref{sec:GiltReview} followed by the gauge fixing procedure from \cref{app:GaugeFixing}. In this appendix, we will explore the smoothness properties of this RG map. In particular, we will provide evidence that the map $R$ is thrice differentiable in a neighborhood of the critical fixed point tensor $A_*$.

Let $A$ be an approximation of $A_*$ and $v$ be a four-legged tensor with the same bond dimension. We can try to evaluate the derivative of $R$ at $A$ in the direction $v$ using the symmetric finite difference approximation. This approximation depends on the step size $h$ and is given by Eq.~\eqref{eq:symfinite} which we copy here
\begin{equation}
	\label{eq:fd_formula}
	D_h = \frac{R(A+hv)-R(A-hv)}{2h}.
\end{equation}
Without loss of generality, we assume that $v$ is unit-normalized.

Standard considerations \cite[Sec.~5.7]{Press2007} suggest that we will have
	\beq\label{eq:fd_formula_error}
	D_h = (\nabla R).v + \eps_{t} +\eps_r\,,
	\eeq
	where $\eps_t$ is the truncation error (from truncating the Taylor series), and $\eps_r$ is the roundoff error due to accumulated numerical errors in the evaluation of $R(A\pm sv)$. For $R$ thrice differentiable in a neighborhood of $A_*$, we may Taylor-expand $R(A\pm hv)$, up to $h^3$ terms given the differentiability assumption. The $O(h^0)$ and $O(h^2)$ terms canceling when using the symmetric difference, we obtain $\eps_t\sim h^2$. On the other hand $\eps_r$ can be estimated as $\eps_r\sim \eps_R/h$ where $\eps_R$ is the roundoff error in evaluation of $R$.
	We thus expect
\begin{equation}\label{eq:Diff1}
	D_h=(\nabla R)(A).v  + O(h^2) + O(\eps_R/h)\,.
\end{equation}
 Note that the $O(h^2)$ truncation error is expected to scale smoothly with $h$, with overall size which depends on the direction of $v$, while the $O(\eps_R/h)$ roundoff error is expected to be random but of a magnitude which is largely $v$ independent.
 
Our first goal is to test this equation. We will do it by computing $D_h$ for a discrete sequence of values of $h$, $h_q \rightarrow 0$ geometrically, and monitoring the differences $D_{h_q}-D_{h_{q+1}}$. 
\cref{fig:diff_test_results} shows the result of this test for Gilt-TNR parameters: $\chi=30, \epsilon_{\rm gilt}=6\times 10^{-6}$, for $A=A^{(n_*)}$ from \cref{tab:fp30}, and for two choices of a unit-normalized direction $v$. The first direction is along the eigenvector of $\nabla R$ with the largest $\mathbb{Z}_2$-even eigenvalue, $\lambda \approx 2$.\footnote{We obtained the eigenvector using the Arnoldi method. For this, we approximated $\nabla R.v$ using finite difference formula \cref{eq:fd_formula} with $h=10^{-4}$.} The second direction is a random tensor. (We repeated this test for several random $v$'s, with similar results.) On the vertical axis of \cref{fig:diff_test_results} we show the difference $\|D_{h_q}-D_{h_{q+1}}\|$ as a function of $h_q=10^{-3-0.05q}$.

\begin{figure}
	\centering
	\includegraphics[width=\textwidth]{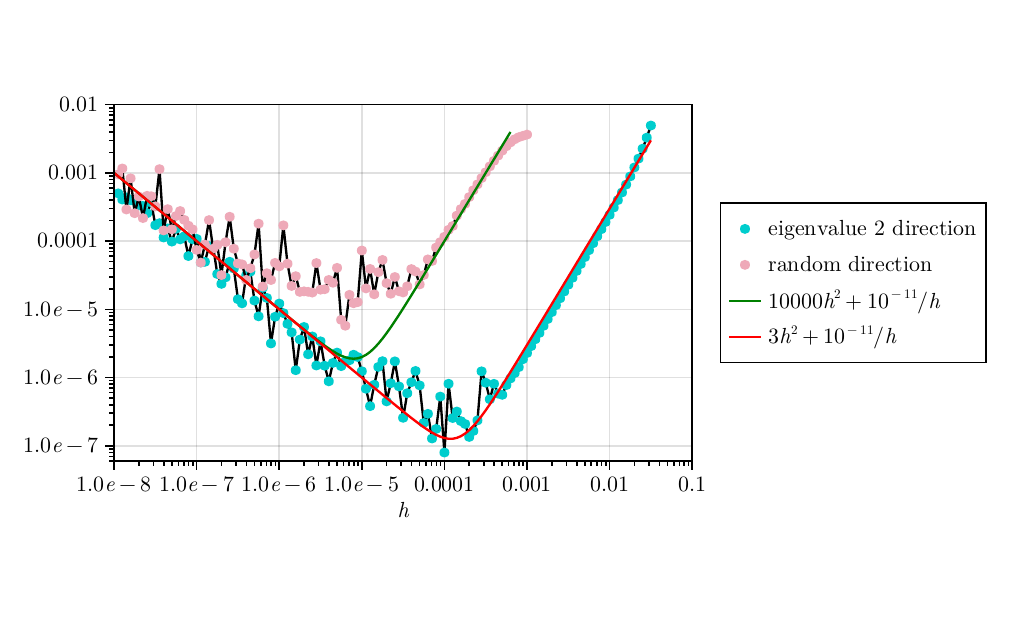}
	\caption{The result of the differentiability test for Gilt-TNR at the approximate fixed point for $\chi=30$ given in \cref{tab:fp30}. The green and red lines illustrate the qualitative behavior of asymptotics and serve to guide the eye. The coefficients for the $h^2$ and $1/h$ terms in both lines have been fixed manually to get a visually reasonable agreement, and not as a result of a fit. The $10^{-11}$ coefficient of the $1/h$ term gives an idea of the roundoff error of Gilt-TNR.
		\label{fig:diff_test_results}}
\end{figure}

Let us discuss what we observe in \cref{fig:diff_test_results}. Consider first the test result for the random direction. In this case, $D_h$ converges as $O(h^2)$ for $7\times10^{-5} < h < 5\times 10^{-4}$. For $h > 5\times 10^{-4}$, the points deviate from $O(h^2)$, indicating that higher-order terms become significant. For $h<7\times 10^{-5}$, the points exhibit chaotic behavior until around $h=10^{-6}$, where they start aligning with the $O(\epsilon_R/h)$ line as predicted. The poor agreement of the test result with \cref{eq:Diff1} for $10^{-6}<h<7\times 10^{-5}$ may be due to the fact that the random vector $v$ has a large overlap with the lower part of the $\nabla R$ spectrum, which is perhaps more sensitive to rounding errors. As to the test result for the eigenvalue $2$ direction, it agrees with \cref{eq:Diff1} very well over the full range, except for $h>0.03$, where points start deviating from the $O(h^2)$ line, which again must be due to higher-order corrections. 

We consider these results to agree well with \cref{eq:Diff1}, providing evidence that our RG map is at least thrice differentiable near the fixed point. Note that the coefficients of the $1/h$ chaotic behavior at small $h$ are roughly similar in \cref{fig:diff_test_results} for both directions, in agreement with \cref{eq:Diff1} for $\eps_R\sim 10^{-11}$. On the other hand the coefficients of the $h^2$ truncation error are vastly different, which suggests that the third derivative along a random direction is much larger than along a typical direction corresponding to low-lying eigenvalues.

We also used \cref{fig:diff_test_results} to optimize the choice of $s$ in \eqref{eq:fd_formula}. 
We chose $h=10^{-4}$ for our computations, because: 1)
This is very close to the optimal spot for the eigenvalue $2$ direction, where $O(h^2)+O(\epsilon_R/h)$ error is minimal; 2) The derivative in the random direction has effectively converged at $h=10^{-4}$, and rounding errors have not yet spoiled the result.  Based solely on the result for the random direction, a better choice could appear to be $h\approx 6\times 10^{-5}$ where, despite the noise, $\|D_{h_q}-D_{h_{q+1}}\|$ takes the smallest values. However, we believe $h=10^{-4}$ is more appropriate as we are primarily concerned with derivatives in the directions of large eigenvalue vectors of $\nabla R$.

It may be possible to improve the Jacobian evaluation by using automatic differentiation (as mentioned in \cref{sec:parameter-choices-and-scaling}), which would remove the $\eps_R/h$ error, replacing it by $O(\eps_R)$. Or one could use a finite difference approximation of higher order, which still has $\eps_R/h$ error, but reduces it since a larger step $h$ can be chosen for the same accuracy. In this paper we have not pursued these because the error in the fixed point determination that we achieved was $10^{-9}$, already quite close to the error $\eps_R\sim 10^{-11}$ of a single RG step. Thus we believe that single RG step error largely determines the final accuracy, and optimizing Jacobian evaluation is not likely to be worthwhile at this stage (but may be useful in the future).

\section{The 2D 3-state Potts model fixed point via the Newton method}\label{app:the-3-state-potts-model}

Prior to this work, a tensor-RG study of the 2D 3-state Potts model was provided in \cite{Li_2022}, where the loop-TNR \cite{LoopTNR} algorithm was used to determine critical temperature, central charge, and scaling dimensions of the critical theory to a good precision. Here, we are interested in slightly different questions. This section is dedicated to finding the RG critical fixed point and studying the RG map linearization around it. We begin with a brief description of the 3-state Potts model and its tensor network representation. Then, we demonstrate the result of the Newton method and examine the upper part of the spectrum of the RG map's Jacobian at the critical fixed point to verify consistency with the CFT predictions.

\subsection{3-state Potts model tensor network representation}

The 2D 3-state Potts model is defined by the following partition function:
\begin{equation}\label{eq:potts0}
Z=\sum_{s_{x,y}=\pm 1, 0} e^{\frac{1}{T} \sum_{x,y} \cos \left( \frac{2\pi}{3} (s_{x+1,y}-s_{x,y}) \right) + \cos \left( \frac{2\pi}{3} (s_{x,y+1}-s_{x,y}) \right) }
\end{equation}
It has a phase transition at $T=T_c=\frac{3}{2 \ln (1+\sqrt{3})}$ separating the ordered and disordered phases. Following the procedure outlined for the Ising model in \cref{app:init}, we construct a four-legged tensor $A$ that reproduces partition function from \cref{eq:potts0}. Specifically, after rotating the lattice by $\pi/4$, we encode the interactions between spins $s_1,s_2,s_3,s_4$ around a single plaquette in the tensor components:
\begin{equation}\label{eq:potts1}
\myinclude[scale=1]{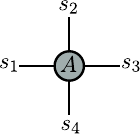}=e^{ \frac{1}{T} \left( \cos\left( \frac{2\pi}{3}s_{12} \right)+\cos\left( \frac{2\pi}{3}s_{23} \right)+\cos\left( \frac{2\pi}{3}s_{34} \right)+\cos\left( \frac{2\pi}{3}s_{41} \right) \right)}.
\end{equation}
Here, $s_{ij}=s_i-s_j$ denotes the difference between neighboring spins.

The 3-state Potts model exhibits the $S_3=\mathbb{Z}_3 \rtimes \mathbb{Z}_2$ symmetry generated by the transformations:
\begin{align}
&\mathbb{Z}_3: s \mapsto s+1 \mod 3; \\
&\mathbb{Z}_2: s \mapsto -s \ \ \ \mod 3.
\end{align}
Let $\eta$ be the generator of the $\mathbb{Z}_3$ subgroup of $S_3$ acting on the index space of the tensor $A$ defined in \cref{eq:potts1}:
\begin{equation}
\eta=\left( \begin{matrix}
0 & 0 & 1\\
1 & 0 & 0\\
0 & 1 & 0\\
\end{matrix} \right).
\end{equation}
In our numerical analysis, we express $A$ in the orthonormal basis of $\eta$'s eigenvectors $\ket{0}, \ket{1}, \ket{2}$, corresponding to eigenvalues $e^{2\pi i q/3}$ with $q=0,1,2$, respectively. We denote the resulting normalized tensor by $A_{3P}(t)$, where $t=T/T_c$ is the reduced temperature. \footnote{While the basis vectors $\ket{0}, \ket{1}, \ket{2}$ are complex, one may show that the tensor $A_{3P}(t)$ is real due to $S_3$ symmetry of $A$. The initial tensor is invariant under an $S_3$ action, and our RG map preserves this symmetry.}

\subsection{The Newton method results}

\begin{figure}
	\centering
	\includegraphics{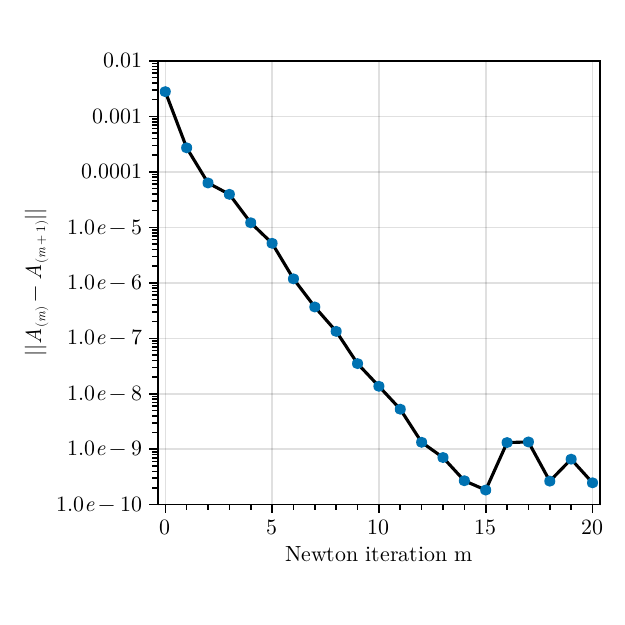}
	\caption{ 
		Convergence of the Newton method for the 3-state Potts model. Gilt parameters: $\chi=30, \epsilon_{\rm gilt} = 3 \times 10^{-5}$. The trajectory is initiated at the tensor $A_{(0)}$ whose distance to the exact fixed point tensor (for $\chi=30$) can be estimated as $\sim 3 \times 10^{-3}$ 
		\label{fig:potts_newton_conv}
	}
\end{figure}

The numerical results presented in this section are obtained using the following Gilt parameters: $\chi=30$, $\epsilon_{\rm gilt}=3\times 10^{-5}$. 

Before applying the Newton method, one needs a sufficiently accurate initial tensor $A_{(0)}$. In particular, $A_{(0)}$ should be close enough to the exact fixed point tensor so that they both belong to one connected region where the RG map is differentiable. To obtain $A_{(0)}$, we employ the shooting method, which yielded the finite-$\chi$ critical temperature $t^{[30]}_c \approx 0.9999063149$ at bond dimension $\chi=30$.  Running the RG trajectory starting at $A_{3P}(t \approx t^{[30]}_c)$ for $20$ steps we get tensor a $A_{(0)}$ that approximates the fixed point tensor with accuracy $\sim 5\times 10^{-3}$. 

\cref{fig:potts_newton_conv} shows the convergence of the Newton method initiated at $A_{(0)}$. For this plot, we used the projector $P_s$ with rank $s=16$ (see \cref{eq:approxJ}). We see exponential convergence similar to the Ising case. All tensors in the lower portion of \cref{fig:potts_newton_conv} should equally well approximate the exact fixed point. For further computations, we choose tensor $A_{(m_*)}$ with $m_*=15$ as our approximation of the critical fixed point, which is the last iteration where $\|A_{m+1}-A_{m}\|$ decreases.

A direct method to quantify the quality of our solution $A_{(m_*)}$ is to check how much the RG step changes it. Applying $R^\circ$, we get:
\begin{equation}
	\| R^{\circ}(A_{(m_*)}) - A_{(m_*)} \| \sim 10^{-13},
\end{equation}
which confirms that we have converged to a solution of the RG equation.

To conclude this section, let us demonstrate that the relevant and marginal part of the Jacobian spectrum agrees with the CFT expectations. We do not consider the irrelevant part of the spectrum, as it is flooded with non-universal eigenvalues coming from total derivatives operators (see the discussion in \cref{sec:CFTexp}). Our map acts in the $\mathbb{Z}_3$-invariant space of tensors. \cref{tab:3Potts} lists the relevant and marginal $\mathbb{Z}_3$-invariant quasiprimaries of the 3-state Potts model.

\begin{table}[ht]
	\centering
	\begin{tabular}{ccccc}
		\toprule
		    &  $\mathds{1}$ & $\epsilon$ & $\Phi, \overline{\Phi}$ & $T, \overline{T}$ \\
		\midrule
		  $\Delta$ & $0$ & $4/5$ & $9/5$ & $2$ \\
		  $\ell$ & $0$ & $0$ & $\pm 1$ & $\pm 2$ \\
		\bottomrule
	\end{tabular}
	\caption{\label{tab:3Potts}2D 3-state Potts CFT Virasoro primary and quasiprimary $\mathbb{Z}_3$-invariant operators of scaling dimension $\Delta\leq 2$, with their quantum numbers (see e.g.~\cite{DiFrancesco:1997nk}).}
\end{table}

\begin{table}[ht]
	\centering
	\begin{NiceTabular}{c|c c c}
		\toprule
		$\mathcal{O}$ & $\lambda_{\rm CFT}=i^{\ell_{\mathcal{O}}}2^{2-\Delta_\mathcal{O}}$ & $\lambda$ of $\nabla R^\circ$ & relative difference  \\
		\midrule
		 $\epsilon$  & $2^{6/5} \approx 2.2974$ & 2.2897 & $0.3 \%$  \\
		 $\Phi$ & $2^{1/5}i \approx 1.1487 i$  & $-0.001+1.1456i$ &  $0.3 \%$ \\
		 $\bar{\Phi}$ & $-2^{1/5}i \approx-1.1487 i$ & $-0.001-1.1456i$ & $0.3 \%$ \\
		 $T$  & $-1$  & $-0.9958$ &  $0.4 \%$     \\
		 $\overline{T}$ & $-1$   & $-0.9686$ & $3.1\%$  \\
		\bottomrule   
	\end{NiceTabular}
	\caption{Columns 1,2: relevant and marginal CFT quasiprimaries ($\mathbb{Z}_3$ invariant sector), and their exact eigenvalues. Column 3: The first few largest, in absolute value, Jacobian eigenvalues at the approximate fixed point for the rotating Gilt-TNR ($\nabla R^\circ$). Column 4: relative differences between the CFT predictions and the Jacobian eigenvalues. It is a bit hard to compare our accuracy to \cite{Li_2022} where the results are presented as a graph and not in numerical form, but it appears comparable.\label{tab:pots_jac_spec}}
\end{table}

The Jacobian eigenvalues should be related to the scaling dimensions and spins of nontrivial quasiprimaries (i.e., excluding $\mathds{1}$) from \cref{tab:3Potts} via \cref{eq:lambda-mod}. \cref{tab:pots_jac_spec} demonstrates agreement with this prediction. However, the relative difference between the CFT predictions and the Jacobian eigenvalues is approximately an order of magnitude higher than in the Ising case. Let us comment on this. 

We observed that the accuracy of physical quantities improves when $\epsilon_{\rm gilt}$ goes down. This is due to the error introduced to the network by Gilt matrices $Q$ becoming smaller for smaller $\epsilon_{\rm gilt}$ (see \cref{app:moreGilt}). However, as one decreases $\epsilon_{\rm gilt}$ to $0$, the effectiveness of Gilt at filtering out CDL tensors decreases, and eventually the RG map will not have a fixed point. Also, recall that Gilt-TNR consists of two steps: Gilt followed by TRG. The Gilt step reduces the bond dimension of the network. Smaller $\epsilon_{\rm gilt}$ would lead to a larger intermediate bond dimension after Gilt, leading to larger errors in the subsequent TRG steps. Thus, for a given bond dimension $\chi$, one should try to choose $\epsilon_{\rm gilt}$ as small as possible but such that the errors of the Gilt and TRG steps are balanced and the fixed point still exists. The existence of the fixed point can be tested by examining the shooting method results. A larger bond dimension $\chi$ generally allows for smaller $\epsilon_{\rm gilt}$.

For the Potts model with $\chi=30$, we found $\epsilon_{\rm gilt}=3 \times 10^{-5}$ to be a good value. This is, however, much larger than $\epsilon_{\rm gilt}=6 \times 10^{-6}$ used for the Ising model. This larger value may explain the lower accuracy of our Potts model results. The reason for Potts model requiring larger $\epsilon_{\rm gilt}$ is not clear. We observe, though, an evident correlation --- the larger number of symmetry sectors ($3$ in Potts vs $2$ in Ising) requires more thorough disentangling and so larger $\epsilon_{\rm gilt}$.

\providecommand{\href}[2]{#2}\begingroup\raggedright\endgroup

\end{document}

%% file: shorthand.tex
\def\eps{\epsilon}
\newcommand{\beq}{\begin{equation}} 
\newcommand{\eeq}{\end{equation}}

\def\bZ {\mathbb{Z}}

\def\calO {{\cal O}}
\def\cO{{\cal O}}

\def\calM {{\cal M}}

\def\ge{\geqslant}
\def\le{\leqslant}

\def\leq{\leqslant}
\def\<{\langle}
\def\>{\rangle}
\def\cO{{\cal O}}

\newcommand{\myinclude}[2][]{\raisebox{0.6ex}{\raisebox{-0.5\height}{\includegraphics[#1]{#2}}}}

\usepackage{enumitem}

\renewcommand{\ge}{\geqslant}
\renewcommand{\le}{\leqslant}